\documentclass[11pt]{article}

\usepackage[T1]{fontenc}
\usepackage{lmodern}
\usepackage[a4paper,margin=1in]{geometry}
\usepackage{amsmath,amssymb,amsthm}
\usepackage{mathtools}
\usepackage{bm}
\usepackage{graphicx}
\usepackage{booktabs}
\usepackage{array}
\usepackage[table]{xcolor}
\usepackage{colortbl}
\usepackage{enumitem}
\usepackage{caption}
\usepackage{tikz}
\usetikzlibrary{calc,arrows.meta,decorations.pathreplacing,positioning}
\usepackage[numbers,sort&compress]{natbib}
\usepackage[colorlinks=true,linkcolor=blue,citecolor=blue,urlcolor=blue]{hyperref}

\theoremstyle{plain}
\newtheorem{Theorem}{Theorem}
\newtheorem{Proposition}{Proposition}
\newtheorem{Lemma}{Lemma}
\newtheorem{Corollary}{Corollary}
\theoremstyle{remark}
\newtheorem{Remark}{Remark}
\theoremstyle{definition}

\newcommand{\R}{\mathbb{R}}
\newcommand{\Sphere}{\mathbb{S}}
\newcommand{\Simplex}{\Delta}
\newcommand{\JSD}{\mathrm{JSD}}
\newcommand{\KL}{D_{\mathrm{KL}}}
\newcommand{\FIM}{\mathbf{F}}
\newcommand{\dFR}{d_{\mathrm{FR}}}
\newcommand{\dang}{d_{\mathrm{ang}}}
\newcommand{\GM}{\mathrm{GM}}
\newcommand{\AM}{\mathrm{AM}}
\newcommand{\HM}{\mathrm{HM}}

\title{\bfseries An Information-Geometric Justification for Composite Coherence in Event-Based Narrative Extraction}
\author{Brian Keith-Norambuena\thanks{Department of Computing \& Systems Engineering,
Universidad Cat\'{o}lica del Norte, Antofagasta, Chile.
ORCID: 0000-0001-5734-8962. Correspondence: \texttt{brian.keith@ucn.cl}}}
\date{}

\begin{document}
\maketitle

\begin{abstract}
Graph-based narrative extraction relies on a coherence function to score
transitions between events, but the coherence metrics in current use are
defined operationally and lack an information-theoretic foundation. We
study the composite metric $C = \sqrt{A \cdot T}$, where $A$ is the
angular similarity of document embeddings and $T = 1 - d_{\mathrm{JS}}$
is the topic proximity through the Jensen-Shannon distance of soft
cluster memberships, and we provide an information-geometric reading of
this metric together with an axiomatic characterization of the
geometric-mean combinator. On the product manifold $\Sphere^{d-1}
\times \Simplex^{K-1}_+$, the negative log-coherence decomposes
additively into an angular and a topic cost. Because the Riemannian
metric tensor induced by the Jensen-Shannon distance on the simplex is
proportional to the Fisher information matrix, the topic component is
locally consistent with the Fisher-Rao metric singled out by Chentsov's
theorem. Within a parametric family of combinators (the compensability
spectrum), the geometric mean is the unique combinator consistent with
four natural axioms (a boundary/veto condition, symmetry,
log-additivity, normalization), and the construction also motivates a
proper product metric $d_\times$ that we use as a reference distance.
Experiments on four corpora spanning news and academic domains
(40~to~6{,}000 documents), three general-purpose embedding families
(GPT-4/ada-002, MPNet, MiniLM-L6) plus citation-aware SPECTER2, and
three alternative topic models (LDA, soft $k$-means, GMM) are
consistent with the framework: the Fisher identity holds with $R \ge
0.99$, the geometric mean tracks $d_\times$ closely ($\rho = 0.999$),
and a downstream LLM-as-judge consistency check shows that the
geometric mean is not empirically dominated by any alternative
combinator or single-channel baseline. Sweeping the compensability
spectrum, the bottleneck-coherence gap between extracted storylines and
random sequences splits into a symmetric component---maximized at the
geometric mean on the four corpora above and a fifth, human-navigation
corpus---and a displacement term; a cross-modal case study on a
human-curated image narrative reproduces the same effect in a second
modality. Together,
these results provide an information-geometric justification for the
composite coherence metric and articulate the conditions under which
the geometric mean is the natural choice.
\end{abstract}


\medskip
\noindent\footnotesize This is a preprint of an article accepted for publication in \emph{Entropy} (MDPI). The main text and the supplementary material are combined into this single document; supplementary sections, tables, and figures are numbered with an ``S'' prefix.\normalsize

\vspace{1em}
\section{Introduction}\label{sec:intro}

Narratives are fundamental structures through which humans organize and communicate knowledge about sequences of events \citep{halverson2011master,keith2020narrative}. Computational narrative extraction seeks to automatically discover these structures from document collections, and has applications in intelligence analysis \citep{shahaf2010connecting}, journalism \citep{vossen2015newsreader}, and social science \citep{keith2023survey}.

At the core of graph-based narrative extraction methods lies a \emph{coherence function} that measures how smoothly one event transitions into another. Following the standard event-based formulation, in which each document represents a single event \citep{keith2023survey}---a common and natural assumption for news corpora, whose scope we discuss in Section~\ref{sec:prelim-narrative}---a corpus of $n$~documents induces a coherence graph $G=(V,E,C)$ with documents as nodes and edge weights $C(d_i,d_j)$ quantifying the transition quality between events $d_i$ and $d_j$. Storylines, the basic unit of a narrative, are paths through this graph that maximize some aggregate of edge coherence \citep{shahaf2010connecting,shahaf2012connecting,keith2020narrative,german2025narrative}.

The survey by Keith~Norambuena et~al. \citep{keith2023survey} documents the diverse, operational ways coherence has been defined: word-influence random walks \citep{shahaf2010connecting}, KL divergence with temporal decay \citep{yan2011evolutionary}, and, in multi-criteria path objectives, a coherence term weighted alongside auxiliary criteria such as relevance and coverage \citep{li2013evolutionary}. Each definition captures an intuitive aspect of narrative smoothness, but none of them is derived from an underlying geometry of the document space, none is shown to satisfy proper metric properties, and none is the consequence of an explicit set of axioms. These three concerns---absence of geometric structure, lack of true metric properties, and absence of an axiomatic foundation---are related but distinct, and the present paper addresses all three in a single framework.

Keith and Mitra \citep{keith2020narrative} introduced, in their Narrative Maps framework, a composite coherence metric
\begin{equation}\label{eq:coherence}
 C(d_i, d_j) \;=\; \sqrt{\,A(e_i, e_j)\;\cdot\; T(\hat{e}_i, \hat{e}_j)\,}
\end{equation}
that combines the angular similarity~$A$ of the document embeddings with the topic similarity $T = 1 - d_{\mathrm{JS}}$ (where $d_{\mathrm{JS}} = \sqrt{\JSD}$ is the Jensen-Shannon distance) of the soft cluster membership distributions, through the geometric mean. While this metric performs well empirically, an information-theoretic justification has not been formally established for it.

This paper makes four contributions:

\textbf{Geometric contribution (Section~\ref{sec:product}).} We give a product-manifold reading of the composite coherence metric on $\Sphere^{d-1} \times \Simplex^{K-1}_+$, on which the negative log-coherence decomposes additively into an angular and a topic cost. Because the Riemannian tensor induced by the Jensen-Shannon distance on the simplex is proportional to the Fisher information matrix, the topic component is locally consistent with the Fisher-Rao geometry singled out by Chentsov's theorem \citep{chentsov1982statistical}. We also define a proper product metric $d_\times$ that we use as a reference distance in the empirical analysis.

\textbf{Axiomatic contribution (Sections~\ref{sec:geomean} and~\ref{sec:properties}).} Within a parametric compensability spectrum of power-mean combinators (Min, HM, GM, AM, Max), we characterize the geometric mean as the unique combinator consistent with four natural axioms (a boundary/veto condition, symmetry, log-additivity, normalization). We discuss explicitly which axioms are load-bearing for this uniqueness statement and how dropping each one enlarges the admissible family. Adjoined to this spectrum is the Quad combinator $1 - d_\times^2$, which is not itself a power mean but arises directly as the similarity-domain counterpart of the product metric. A further proposition decomposes the bottleneck-coherence gap between coherent and incoherent storylines, swept across the compensability spectrum, into an even and an odd part: the even part is maximized at the geometric mean exactly when incoherent storylines are the more channel-imbalanced---the design premise of the coherence metric---and the odd part only displaces the observed peak. Section~\ref{sec:properties} records the formal properties of this combinator family: the power-mean compensability ordering, the position of the Quad combinator, and the metric structure of the bounded dissimilarities (notably that $1 - C_{\GM}$ is not a metric).

\textbf{Information-theoretic contribution (Section~\ref{sec:connections}).} We connect each per-channel cost to an exact classical information-theoretic quantity: the angular cost is the surprisal of a random-hyperplane (SimHash) collision, and the topic cost is a strictly increasing function of a mutual information. A data-processing argument further shows that the two channels capture complementary scales of abstraction (\emph{scale complementarity}), so that the composite is genuinely more informative than either channel alone.

\textbf{Empirical contribution (Section~\ref{sec:experiments}).} We validate the framework on six corpora drawn from the narrative-extraction literature. The metric-level and downstream analyses use four event-based corpora---Cuba, COVID, VisPub, and AMiner ($40$ to $6{,}000$ documents) from the Narrative Maps \citep{keith2020narrative} and Narrative Trails \citep{german2025narrative} repositories---together with three general-purpose embedding families plus citation-aware SPECTER2 and three alternative topic-model families (LDA, soft $k$-means, GMM). On these, the Fisher identity holds with $R \ge 0.99$, the geometric mean tracks $d_\times$ closely ($\rho = 0.999$), and a downstream LLM-as-judge consistency check shows that the geometric mean is not empirically dominated by any alternative combinator or single-channel baseline (Friedman $p \ge 0.099$). The bottleneck-gap analysis then adds the Wikispeedia human-navigation corpus \citep{west2012wikispeedia} as a non-circular, human-grounded anchor and, as a cross-modal case study, the ROGER expedition-photograph corpus \citep{german2025roger}: throughout, the metric's design premise holds, so that the symmetric component of the separation between extracted storylines and random sequences is maximized at the geometric mean. A secondary observation from the downstream comparison is that different combinators induce structurally different storylines at statistically indistinguishable aggregate quality, so the choice of combinator is a modeling decision---about which narrative structures to foreground---rather than a free parameter to be tuned for quality (Section~\ref{sec:discussion}).

The remainder of the paper is organized as follows. Section~\ref{sec:prelim} fixes notation and reviews related work. Section~\ref{sec:product} develops the product-manifold reading and its local Chentsov compatibility; Section~\ref{sec:geomean} gives the axiomatic characterization of the geometric mean together with the bottleneck-gap analysis; Section~\ref{sec:properties} records the metric's formal properties; and Section~\ref{sec:connections} connects the per-channel costs to classical information theory. Section~\ref{sec:experiments} reports the experimental validation; Section~\ref{sec:discussion} discusses implications and limitations; and Section~\ref{sec:conclusions} concludes.

\section{Preliminaries and Related Work}\label{sec:prelim}

\subsection{Event-Based Narrative Extraction}\label{sec:prelim-narrative}
We consider a temporally ordered corpus $\mathcal{D}=\{d_1,\dots,d_n\}$ in which each document represents a single event. Operationalizing events at the document level is the standard assumption in event-based narrative extraction \citep{keith2023survey,german2025narrative}: it is natural for news corpora, where an article---especially a breaking-news article---typically reports one main event, and it lets the narrative be studied at the level of inter-document transitions. Finer (sub-document or multi-event) resolutions are possible \citep{keith2023survey} but are orthogonal to the coherence question studied here. A \emph{coherence graph} $G=(V,E,C)$ has $V=\mathcal{D}$, edges $E$ connecting temporally compatible documents (i.e., $(d_i, d_j)\in E$ only if $i < j$, enforcing forward temporal ordering), and weights $C\colon V\!\times\!V\to[0,1]$ measuring transition quality. A \emph{storyline} is a path $\pi=(d_{i_1},\dots,d_{i_L})$ through $G$---a temporally ordered sequence of $L$ documents ($i_1<\dots<i_L$), so that $L$ is the number of events in the storyline. A \emph{narrative map} is a collection of overlapping storylines forming a DAG \citep{keith2020narrative}, though a DAG structure is not a strict requirement \citep{keith2023survey,german2025narrative}. The quality of $\pi$ depends on the coherence of its edges: common objectives include maximizing the minimum edge coherence (maximin paths \citep{german2025narrative}) or the geometric mean of edge coherences.

The survey by Keith~Norambuena et~al.~\citep{keith2023survey} catalogues nine distinct mathematical formulations of narrative coherence across three event resolution levels (clusters, documents, sentences), none of which is grounded in a principled mathematical framework. Existing metrics are either single-channel (e.g., word-influence random walks \citep{shahaf2010connecting,shahaf2012connecting}, minimum cosine similarity \citep{zhou2014generating}, $n$-gram overlap \citep{tran2013leveraging}) or combine multiple criteria via ad hoc rules (e.g., KL divergence with exponential temporal decay \citep{li2013evolutionary}, sigmoid-transformed JSD \citep{li2015tracking}, weighted multi-criteria scoring \citep{xu2018generating,liu2017growing,liu2020story}). Keith and Mitra \citep{keith2020narrative} introduced the composite metric $C=\sqrt{A\cdot T}$ studied in this paper, combining angular similarity of document embeddings with topic similarity via the Jensen-Shannon distance, which is the only formulation in the surveyed literature that operates on two distinct representation spaces with a multiplicative combinator. German et~al.~\citep{german2025narrative} adopted this metric in their subsequent work on Narrative Trails. Both works specify the metric operationally, without a theoretical justification. This paper provides the information-geometric justification for that specific combination.

\subsection{The Composite Coherence Metric}\label{sec:prelim-metric}
The composite coherence of Equation~\eqref{eq:coherence} combines two
components.

\textbf{Angular similarity.} Each document $d_i$ has a dense embedding $e_i\in\R^d$ obtained from a
language model. The angular similarity is
\begin{equation}\label{eq:angular}
 A(e_i, e_j) \;=\; 1 \;-\; \frac{\arccos\bigl(\cos(e_i, e_j)\bigr)}{\pi} \,,
\end{equation}
where $\cos(e_i,e_j) = e_i^\top e_j / (\|e_i\|\,\|e_j\|)$. Thus $A$ ranges from 0 (antipodal) to 1 (identical direction), and the angular distance $\dang(e_i,e_j)=\arccos(\cos(e_i,e_j))$ is the geodesic distance on the unit sphere $\Sphere^{d-1}$ (see Section~\ref{sec:prelim-sphere}).

\textbf{Topic similarity.} Soft cluster membership distributions $\hat{e}_i\in\Simplex^{K-1}$ are obtained by applying UMAP dimensionality reduction \citep{mcinnes2018umap, mcinnes2018umap_software} (\texttt{n\_neighbors}${}=32$, \texttt{n\_components}${}=48$, \texttt{min\_dist}${}=0.0$, cosine metric, fixed random seed) followed by HDBSCAN clustering \citep{campello2013density} (\texttt{min\_cluster\_size}${}=5$, excess-of-mass cluster selection) with soft membership assignment, where $\Simplex^{K-1} = \{p\in\R^K : p_k\ge 0,\;\sum_k p_k = 1\}$ denotes the probability simplex. The cluster count~$K$ is an emergent output of this configuration rather than a tuned parameter, with $K \ge 4$ the recommended operating regime (see Section~\ref{sec:exp4}, the Supplementary Materials for the cross-corpus and sensitivity analyses, and the Limitations in Section~\ref{sec:limitations}). The topic similarity is
\begin{equation}\label{eq:topic}
 T(\hat{e}_i, \hat{e}_j) \;=\; 1 \;-\; d_{\mathrm{JS}}(\hat{e}_i, \hat{e}_j)\,,
\end{equation}
where $d_{\mathrm{JS}}(p,q) = \sqrt{\JSD(p,q)}$ is the Jensen-Shannon \emph{distance} \citep{endres2003new,osterreicher2003new}, the Kullback-Leibler divergence is $\KL(p\|q) = \sum_k p_k \ln(p_k / q_k)$, and the Jensen-Shannon divergence is $\JSD(p,q) = \tfrac{1}{2}\KL(p\|\,m) + \tfrac{1}{2}\KL(q\|\,m)$ with $m=\tfrac{1}{2}(p+q)$ \citep{lin1991divergence}. Under the base-2 logarithm convention, $\JSD\in[0,1]$ and hence $d_{\mathrm{JS}}\in[0,1]$, so $T\in[0,1]$. We use the JS \emph{distance} (square root of the divergence) rather than the divergence itself for two related reasons. First, $d_{\mathrm{JS}}$ is a proper metric \citep{endres2003new,osterreicher2003new}, which is needed for the product-metric construction of Theorem~\ref{thm:metric-gm}. Second, the topic similarity then has the same first-order structure as the angular similarity: $A = 1 - \dang/\pi$ is linear in the spherical geodesic distance $\dang$, and the relation $d_{\mathrm{JS}} \approx \dFR / \sqrt{2 \ln 2}$ established in Remark~\ref{rem:fisher-connection} makes $T = 1 - d_{\mathrm{JS}}$ linear, to first order, in the Fisher-Rao geodesic distance on the simplex. The two components thus enter the log-coherence decomposition (Proposition~\ref{thm:geodesic}) at the same order in their respective geodesic distances, which is what we mean by ``aligning'' them; using the divergence $\JSD$ in place of $d_{\mathrm{JS}}$ would make the topic component second-order and break this symmetry (Remark~\ref{rem:infinitesimal}).

\textbf{Composite coherence.} With the two components $A$ and $T$ now specified in detail, the composite coherence of Equation~\eqref{eq:coherence} is their geometric mean, $C = \sqrt{A \cdot T}$. Why this particular combinator, applied to these particular components, is well-founded is the question the rest of the paper addresses.

\subsection{Information Geometry}\label{sec:prelim-ig}

The Fisher-Rao metric was introduced by Rao \citep{rao1945information} and studied systematically by Amari \citep{amari1985differential,amari2000methods}. For textbook treatments see Ay et al.\ \citep{ay2017information} and Nielsen \citep{nielsen2020elementary}.

\textbf{Information geometry in text analysis.}
Information-geometric methods have been applied to several NLP tasks, though not previously to narrative coherence. Lafferty and Lebanon \citep{lafferty2005diffusion} introduced diffusion kernels on the multinomial manifold for text classification, exploiting the Fisher-Rao geometry of word distributions. Lebanon \citep{lebanon2005riemannian} extended this to margin classifiers adapted to the simplex geometry. Colombo et al.\ \citep{colombo2022infolm} used the Fisher-Rao distance between masked language model output distributions as an evaluation metric for text generation, finding it outperforms standard $n$-gram metrics. Dhillon et al.\ \citep{dhillon2003divisive} used the generalized Jensen-Shannon divergence for information-theoretic word clustering, minimizing within-cluster JSD. In a different geometric direction, Nickel and Kiela \citep{nickel2017poincare} embedded symbolic data in hyperbolic (constant negative curvature) spaces to capture hierarchical structure. Our work differs in that the geometry arises from the coherence metric's structure rather than being imposed as a design choice: the product manifold $\Sphere^{d-1}\!\times\!\Simplex^{K-1}$ and the Fisher-Rao connection emerge from analyzing an existing metric, not from constructing a new one.

We now recall the key definitions used in the rest of the paper.

\textbf{Jensen-Shannon divergence.}
Lin \citep{lin1991divergence} popularized the JSD as a symmetrized, bounded variant of the KL divergence. It satisfies $\JSD(p,q)\ge 0$ with equality iff $p=q$. Endres and Schindelin \citep{endres2003new} and \"Osterreicher and Vajda \citep{osterreicher2003new} independently proved that $\sqrt{\JSD}$ is a metric. Fuglede and Tops{\o}e \citep{fuglede2004jensenshannon} showed that $\sqrt{\JSD}$ admits an isometric embedding into a Hilbert space.

\textbf{Fisher information matrix.} For a discrete distribution $p=(p_1,\dots,p_K)$ on a finite alphabet, the Fisher information matrix has entries
\begin{equation}\label{eq:fisher}
 \FIM(p)_{kl} \;=\; \frac{\delta_{kl}}{p_k}\,,
\end{equation}
where $\delta_{kl}$ is the Kronecker delta. This is a diagonal matrix with entries $1/p_k$.

\textbf{Fisher-Rao manifold.} The open probability simplex $\Simplex^{K-1}_+ = \{p\in\R^K : p_k>0,\; \sum_k p_k = 1\}$, equipped with the Riemannian metric tensor $g_{\mathrm{FR}}(p) = \FIM(p)$ induced by the Fisher information matrix, is the Fisher-Rao manifold $(\Simplex^{K-1}_+, g_{\mathrm{FR}})$. The Bhattacharyya angle is
\begin{equation}\label{eq:fr-distance}
 \dFR(p,q) \;=\; \arccos\Bigl(\sum_{k=1}^{K}\sqrt{p_k \, q_k}\,\Bigr)
 \;\in\; [0,\,\pi/2] \, .
\end{equation}
The Fisher-Rao manifold has constant positive sectional curvature $\kappa = 1/4$, and via the square-root parameterization $\xi_k = 2\sqrt{p_k}$ it maps isometrically onto a piece of the sphere $\Sphere^{K-1}_+$ of radius~2 \citep{amari2000methods}. Note that $\dFR$ as defined is the Bhattacharyya angle, i.e., the angular separation on the \emph{unit} sphere under $\eta_k = \sqrt{p_k}$. Infinitesimally, $\delta p^\top \FIM(p)\,\delta p = 4\,\dFR^2$. The factor of~4 arises because the metric~\eqref{eq:fisher} corresponds to the radius-2 sphere, whose geodesic distance is $2\dFR$.

\textbf{Chentsov's theorem.} Chentsov's uniqueness theorem \citep{chentsov1982statistical} is a cornerstone of information geometry. Among all Riemannian metrics on the manifold of probability distributions on a finite sample space, the Fisher-Rao metric is the unique metric (up to a global scaling constant) that is invariant under sufficient statistics (equivalently, under the congruent Markov embeddings of the sample space that they induce). See Campbell \citep{campbell1986extended} for an extension and Bauer, Bruveris, and Michor \citep{bauer2016uniqueness} for a modern proof in the infinite-dimensional setting. The scope of this uniqueness statement is important: it characterizes the metric \emph{tensor} (an infinitesimal object), not the global finite-distance distance function, so any monotone reparameterization of arc length yields a different scalar distance with the same underlying Riemannian geometry. We will use Chentsov's theorem in this restricted sense throughout.

\textbf{JSD-Fisher connection.}
The second-order expansion
\begin{equation}\label{eq:jsd-fisher}
 \JSD_{\mathrm{nat}}(p,\,p+\delta p) \;=\; \frac{1}{8}\,\delta p^\top \FIM(p)\,\delta p \;+\; O\!\bigl(\|\delta p\|^3\bigr)
\end{equation}
(equivalently $\tfrac{1}{2}\dFR^2$, since $\delta p^\top \FIM\,\delta p = 4\dFR^2$ for the Bhattacharyya angle~\eqref{eq:fr-distance}) is a well-known identity in information geometry \citep{nielsen2020elementary}. The novelty in the present paper is not the identity itself but its \emph{application}: recognizing that the topic component of an existing narrative coherence metric inherits, at the level of the Riemannian metric tensor, the Fisher-Rao geometry that Chentsov's theorem singles out among invariant metrics on the simplex (Section~\ref{sec:topic-compat}).

\textbf{Logarithm base and JS distance.}
The identity~\eqref{eq:jsd-fisher} is stated for the natural-logarithm convention, written $\JSD_{\mathrm{nat}}$. Throughout the rest of the paper, $\JSD$ (unsubscripted) denotes the base-2 Jensen-Shannon divergence, $\JSD = \JSD_{\mathrm{nat}}/\ln 2 \in [0,1]$, for which the identity becomes
\begin{equation}\label{eq:jsd-fisher-base2}
 \JSD(p,\,p+\delta p) \;=\; \frac{1}{8\ln 2}\,\delta p^\top \FIM(p)\,\delta p \;+\; O\!\bigl(\|\delta p\|^3\bigr)\,,
\end{equation}
so that $\JSD \approx \dFR^2/(2\ln 2)$ for nearby distributions. Since the topic similarity uses the Jensen-Shannon distance $d_{\mathrm{JS}} = \sqrt{\JSD}$, we obtain
\begin{equation}\label{eq:djs-fisher}
 d_{\mathrm{JS}}(p,\,p+\delta p) \;\approx\; \frac{\dFR(p,\,p+\delta p)}{\sqrt{2\ln 2}} \;+\; O\!\bigl(\|\delta p\|^2\bigr)\,,
\end{equation}
which is \emph{first}-order in the Fisher-Rao distance, paralleling the first-order relationship $A = 1 - \dang/\pi$ on the angular side. Two levels of this connection must be distinguished. At the \emph{infinitesimal} level, the relationship is \textbf{exact}: the Riemannian metric tensor induced by $d_{\mathrm{JS}}$ on $\Simplex^{K-1}_+$ is $(1/(8\ln 2))\,\FIM(p)$, i.e., proportional to the Fisher information matrix, so $d_{\mathrm{JS}}$ induces the same Riemannian geometry as the Fisher-Rao metric (up to the global scale factor $1/\sqrt{2\ln 2}$). At \emph{finite distances}, the relationship~\eqref{eq:djs-fisher} is approximate, with higher-order corrections that vanish as $\|\delta p\|\to 0$. Empirically, the finite-distance approximation achieves a Pearson correlation $R \ge 0.99$ across all corpora (Section~\ref{sec:exp4}); throughout, $R$ denotes the Pearson product-moment correlation coefficient (not its square) and $\rho$ a rank correlation.

\textbf{Finite-distance justification.}
Chentsov's theorem guarantees uniqueness of the \emph{metric tensor} (an infinitesimal object), but the coherence metric operates at finite distances. This is principled for three reasons. (a)~The metric tensor induced by $d_{\mathrm{JS}}$ on $\Simplex^{K-1}_+$ is \emph{exactly} $(1/(8\ln 2))\,\FIM(p)$, a constant multiple of the Fisher information matrix, so $d_{\mathrm{JS}}$ and $\dFR/\sqrt{2\ln 2}$ have the same unparametrized geodesics (the global constant rescales arc length and sectional curvature). (b)~The product decomposition (Proposition~\ref{thm:geodesic}) uses $d_{\mathrm{JS}}$ directly, not its Fisher-Rao approximation. (c)~Writing $\mathrm{BC}(p,q)=\sum_k\sqrt{p_kq_k}=\cos\dFR$ for the Bhattacharyya coefficient (exact, by~\eqref{eq:fr-distance}), the known bound $\tfrac{1}{2}(1-\mathrm{BC}) \le \JSD_{\mathrm{nat}} \le 1-\mathrm{BC}$ \citep{topsoe2000some} sandwiches $\JSD_{\mathrm{nat}}$ between two explicit functions of the Fisher-Rao distance at every distance, not merely infinitesimally---equivalently $\tfrac{1}{4}H^2 \le \JSD_{\mathrm{nat}} \le \tfrac{1}{2}H^2$ with $H^2 = 2(1-\mathrm{BC})$ the squared Hellinger distance (Proposition~\ref{prop:hellinger})---and provides explicit non-asymptotic control, confirmed empirically by $R \ge 0.99$ across all four corpora (Section~\ref{sec:exp4}).

\subsection{Spherical and Product-Manifold Geometry}\label{sec:prelim-sphere}
The unit sphere in $\R^d$ is denoted $\Sphere^{d-1} = \{x\in\R^d : \|x\| = 1\}$. For unit vectors $u,v\in\Sphere^{d-1}$, the angular distance $\dang(u,v) = \arccos(u^\top v)$ is the geodesic distance of the round metric on $\Sphere^{d-1}$. The sphere is a Riemannian manifold of constant positive sectional curvature~1, and the normalized angular similarity $A = 1 - \dang/\pi$ maps $[0,\pi]\to[0,1]$.

\label{sec:prelim-product}On the topic side, the relevant construction is the product Riemannian manifold: $M_1\times M_2$ with the product metric $g = g_1 \oplus g_2$ is standard in Riemannian geometry \citep{lee2018introduction}. In multi-view learning, combining similarity functions from different feature spaces is common \citep{xu2013survey}, but the standard approach is additive (kernel sums) rather than multiplicative (geometric mean). The geometric mean of positive-definite matrices has been studied extensively in matrix analysis \citep{bhatia2009positive}, but its use as a combinator of scalar similarities is less explored.

\section{The Product Manifold Structure}\label{sec:product}
We now show that the composite coherence metric admits a natural interpretation as a distance on a product Riemannian manifold. The tools we use (product Riemannian manifolds, the Fisher-Rao metric, and Chentsov's theorem) are well-established in information geometry \citep{amari2000methods,ay2017information}. The contribution is recognizing that they apply to this specific coherence metric from the narrative extraction literature, yielding consequences for its justification and diagnostic use that were not previously apparent.

\textbf{Formal results, interpretations, and empirical checks.}
Sections~\ref{sec:product} to~\ref{sec:connections} mix three kinds of statement, and we keep them typographically distinct so that the logical status of each conclusion is clear. \emph{Theorems, propositions, and corollaries} (with proofs) state what is formally established under explicitly stated hypotheses. Paragraphs headed \textbf{Interpretation} give the conceptual or geometric reading of a formal result; they are not themselves proven and should be read as motivation. Paragraphs headed \textbf{Empirical check}, and all of Section~\ref{sec:experiments}, report measurements on real corpora; these support the framework but do not extend the proofs. Where a result is exact and where it is approximate or corpus-dependent is stated explicitly at each occurrence.

\subsection{Two Geometric Spaces}\label{sec:two-spaces}
The composite coherence metric operates on two distinct Riemannian manifolds. The angular component $A = 1 - \dang/\pi$ measures proximity on $\Sphere^{d-1}$ via the geodesic distance of the round metric: after $\ell^2$-normalization, the embeddings $\bar{e}_i = e_i/\|e_i\|$ lie on $\Sphere^{d-1}$ and $\dang(\bar{e}_i,\bar{e}_j) = \arccos(\bar{e}_i^\top\bar{e}_j)$ is that geodesic distance. The topic component $T = 1 - d_{\mathrm{JS}}$ measures proximity on $\Simplex^{K-1}_+$ via the Jensen-Shannon distance $d_{\mathrm{JS}} = \sqrt{\JSD}$, which is a proper metric \citep{endres2003new,osterreicher2003new} and is locally Fisher-Rao compatible in the sense recorded in Section~\ref{sec:prelim-ig}.

\subsection{Coherence as Product-Manifold Cost}\label{sec:geodesic}
The decomposition that we record next is algebraically immediate---it
follows from $\log\sqrt{AT} = \tfrac{1}{2}(\log A + \log T)$---and its
value lies in the \emph{geometric reading} it enables. Read as a cost on a product Riemannian manifold,
the identity puts the angular and topic terms on the same footing, and
this product-manifold picture is what subsequently motivates the
product metric $d_\times$
(defined and shown to be a proper metric in Theorem~\ref{thm:metric-gm}), the diagnostic
attribution it enables (Section~\ref{sec:five-arguments}), and the
axiomatic characterization of the geometric mean
(Section~\ref{sec:axioms}). Readers familiar with product spaces and
combined distance functions may find Proposition~\ref{thm:geodesic}
unsurprising on its own; its role here is structural, as the pivot
between the metric-construction and axiomatic strands of the framework.

\begin{Proposition}[Product decomposition]\label{thm:geodesic}
 Let $\dang \in [0, \pi]$ be the angular distance on $\Sphere^{d-1}$
 and $d_{\mathrm{JS}} \in [0, 1]$ the base-2 Jensen-Shannon distance on
 $\Simplex^{K-1}_+$. For any documents $d_i$ and $d_j$ with $C(d_i,d_j) > 0$ (i.e.\ $\dang < \pi$ and $d_{\mathrm{JS}} < 1$),
 \begin{equation}\label{eq:exact-decomp}
  -\log C(d_i, d_j)
  \;=\;
  \tfrac{1}{2}\, \underbrace{\bigl[-\log\bigl(1 - \dang/\pi\bigr)\bigr]}_{\text{angular cost}}
  \;+\;
  \tfrac{1}{2}\, \underbrace{\bigl[-\log\bigl(1 - d_{\mathrm{JS}}\bigr)\bigr]}_{\text{topic cost}} \,.
 \end{equation}
 Both cost functions $h(x) = -\log(1 - x)$ are monotone increasing and convex on $[0, 1)$, so $C$ is monotone decreasing in each component distance separately; it is not, however, a monotone function of the product metric $d_\times$ (Theorem~\ref{thm:metric-gm}).
\end{Proposition}

\begin{proof}
 Direct algebraic identity: $-\log C = -\log\sqrt{AT} = -\tfrac{1}{2}(\log A + \log T) = \tfrac{1}{2}[-\log(1 - \dang/\pi)] + \tfrac{1}{2}[-\log(1 - d_{\mathrm{JS}})]$.
\end{proof}

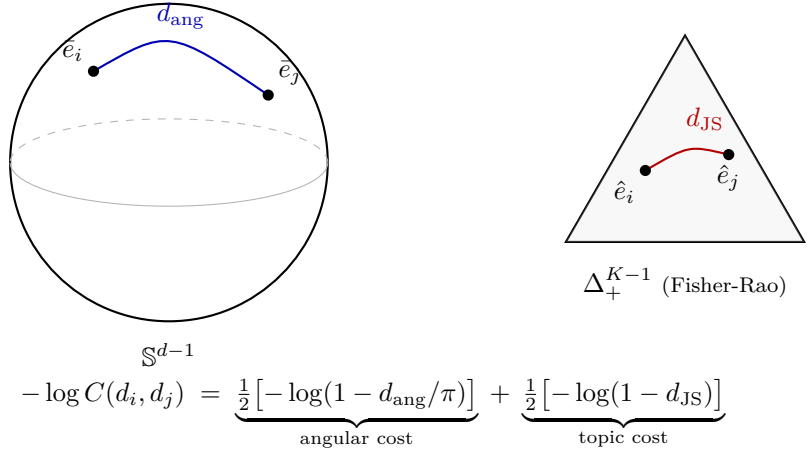
\begin{figure}[t]
 \centering
\begin{tikzpicture}[scale=1.05, font=\small]
  \begin{scope}[shift={(0,0)}]
    \draw[thick] (0,0) circle (2);
    \draw[gray!60] (-2,0) arc[start angle=180, end angle=360, x radius=2, y radius=0.55];
    \draw[gray!60, dashed] (2,0) arc[start angle=0, end angle=180, x radius=2, y radius=0.55];
    \coordinate (di) at (-0.95, 1.15);
    \coordinate (dj) at (1.25, 0.85);
    \draw[blue!70!black, thick] (di) .. controls (0.0, 1.7) .. (dj);
    \fill (di) circle (2pt) node[above left] {$\bar e_i$};
    \fill (dj) circle (2pt) node[above right] {$\bar e_j$};
    \node[blue!70!black] at (0.15, 1.85) {$\dang$};
    \node[below=0.4em] at (0,-2) {$\Sphere^{d-1}$};
  \end{scope}

  \begin{scope}[shift={(6.5,-0.4)}]
    \coordinate (v1) at (-1.5, -0.6);
    \coordinate (v2) at ( 1.5, -0.6);
    \coordinate (v3) at ( 0.0,  2.0);
    \draw[thick] (v1) -- (v2) -- (v3) -- cycle;
    \fill[gray!20, opacity=0.25] (v1) -- (v2) -- (v3) -- cycle;
    \coordinate (pi) at (-0.5, 0.3);
    \coordinate (pj) at ( 0.55, 0.5);
    \draw[red!70!black, thick] (pi) .. controls (0.05, 0.62) .. (pj);
    \fill (pi) circle (2pt) node[below left] {$\hat e_i$};
    \fill (pj) circle (2pt) node[below] {$\hat e_j$};
    \node[red!70!black] at (0.25, 0.95) {$d_{\mathrm{JS}}$};
    \node[below=0.4em] at (0,-0.65) {$\Simplex^{K-1}_+$ \scriptsize (Fisher-Rao)};
  \end{scope}

  \node[anchor=west] at ([yshift=-5pt]-2.0, -3.0) {%
    $\displaystyle -\log C(d_i, d_j)
      \;=\;
      \underbrace{\tfrac{1}{2}\bigl[-\log(1 - \dang/\pi)\bigr]}_{\text{angular cost}}
      \;+\;
      \underbrace{\tfrac{1}{2}\bigl[-\log(1 - d_{\mathrm{JS}})\bigr]}_{\text{topic cost}}$%
  };
\end{tikzpicture}
 \caption{The product-manifold reading of the composite coherence metric. The angular component of the document representation lives on the unit sphere $\Sphere^{d-1}$, and the topic component lives on the open simplex $\Simplex^{K-1}_+$ equipped with the Jensen-Shannon (Fisher-Rao-compatible) geometry. Each document $i$ is represented by an angular vector $\bar e_i$ on the sphere and a topic distribution $\hat e_i$ on the simplex; for a document pair, the geodesic on each factor has length $\dang$ and $d_{\mathrm{JS}}$ respectively. The log-coherence decomposes additively into an angular cost and a topic cost, one per factor of the product manifold.\label{fig:product-manifold}}
\end{figure}

Figure~\ref{fig:product-manifold} illustrates the product-manifold reading: each document pair contributes one geodesic on the sphere factor and one geodesic on the simplex factor, and the log-coherence $-\log C$ adds the two costs with equal weight.

\begin{Remark}[Fisher-Rao connection]\label{rem:fisher-connection}
 The JSD-Fisher identity \eqref{eq:djs-fisher} gives $d_{\mathrm{JS}} \approx \dFR/\sqrt{2\ln 2}$ for nearby distributions. Substituting into the topic term of~\eqref{eq:exact-decomp} yields
 \begin{equation}\label{eq:fisher-substitution}
  -\log C
  \;\approx\;
  \tfrac{1}{2}\bigl[-\log(1 - \dang/\pi)\bigr]
  \;+\;
  \tfrac{1}{2}\bigl[-\log\bigl(1 - \dFR/\sqrt{2\ln 2}\,\bigr)\bigr] \,,
 \end{equation}
 connecting the coherence cost to the Fisher-Rao distance, which is the Riemannian metric on the statistical manifold singled out by Chentsov's theorem (Theorem~\ref{thm:chentsov}). The angular term is left exact. Only the $d_{\mathrm{JS}}\to\dFR$ substitution is approximate, and it is only defined where $\dFR/\sqrt{2\ln 2} < 1$; its finite-distance accuracy on the corpus is reported in the approximation-quality test (Section~\ref{sec:exp1}).
\end{Remark}

\begin{Remark}[Infinitesimal limit]\label{rem:infinitesimal}
 In the limit of vanishing distances ($\dang\to 0$, $d_{\mathrm{JS}}\to 0$), Taylor-expanding each term of~\eqref{eq:exact-decomp} via $-\log(1-x) = x + x^2/2 + O(x^3)$ yields the exact expansion
 \begin{equation}\label{eq:geodesic-approx}
  -\log C
  \;=\;
  \tfrac{1}{2}\bigl( \dang/\pi + d_{\mathrm{JS}} \bigr)
  \;+\;
  \tfrac{1}{4}\bigl( (\dang/\pi)^2 + d_{\mathrm{JS}}^2 \bigr)
  \;+\; O(\epsilon^3) \, .
 \end{equation}
 Both the angular and topic contributions enter at the \emph{same} order: the first-order terms are $\dang/(2\pi)$ and $d_{\mathrm{JS}}/2$. This structural symmetry arises from using the Jensen-Shannon \emph{distance} $d_{\mathrm{JS}} = \sqrt{\JSD}$ (which is first-order in $\dFR$ by eq.~\ref{eq:djs-fisher}) rather than the divergence $\JSD$ (which is second-order in $\dFR$).

 If one further substitutes $d_{\mathrm{JS}} \approx \dFR/\sqrt{2\ln 2}$ (valid at first order), the leading terms become $\tfrac{1}{2}(\dang/\pi + \dFR/\sqrt{2\ln 2})$. This first-order relation cannot, however, be substituted throughout while keeping the second-order expansion exact: since $d_{\mathrm{JS}} = \dFR/\sqrt{2\ln 2} + O(\dFR^2)$, the first-order term $\tfrac12 d_{\mathrm{JS}}$ already carries an $O(\dFR^2)$ correction---of the same order as the second-order terms---so re-expressing the whole expansion in $\dFR$ would drop a second-order contribution. (Substituting inside the second-order term itself, $d_{\mathrm{JS}}^2 \to \dFR^2/(2\ln 2)$, is by contrast harmless: it errs only at $O(\dFR^3)$.) The second-order expansion is therefore exact only in terms of $d_{\mathrm{JS}}$.

 At typical corpus distances ($\dang/\pi \approx 0.35$, see Section~\ref{sec:exp1}), the exact formula~\eqref{eq:exact-decomp} should be used.

 Since $d_\times^2 = \tfrac{1}{2}\bigl[(\dang/\pi)^2 + d_{\mathrm{JS}}^2\bigr]$ (Theorem~\ref{thm:metric-gm}), the expansion~\eqref{eq:geodesic-approx} can be written $-\log C = \tfrac{1}{2}\bigl(\dang/\pi + d_{\mathrm{JS}}\bigr) + \tfrac{1}{2}\,d_\times^2 + O(\epsilon^3)$: the squared product metric appears exactly as the second-order correction to the additive first-order cost.
\end{Remark}

\subsection{Local Chentsov Compatibility of the Topic Component}\label{sec:topic-compat}
\begin{Theorem}[Chentsov]\label{thm:chentsov}
 Among all Riemannian metrics on the manifold $\Simplex^{K-1}_+$ of probability distributions over a finite alphabet of size~$K$, the Fisher-Rao metric is the unique metric (up to a positive multiplicative constant) that is invariant under sufficient statistics---equivalently, under the congruent Markov embeddings $\Simplex^{K'-1}\hookrightarrow\Simplex^{K-1}$ ($K'\le K$) that they induce \citep{chentsov1982statistical}.
\end{Theorem}

This is a classical result of information geometry, not a contribution of the present paper; we restate it here, without proof, only so that the local compatibility statement below can refer to it. See \citet{chentsov1982statistical} for the original, \citet{campbell1986extended} for an extension, and \citet{bauer2016uniqueness} for a modern proof.

\textbf{Interpretation.}
In the restricted (metric-tensor) sense of Section~\ref{sec:prelim-ig}, Theorem~\ref{thm:chentsov} bears on the present setting through the following compatibility statement.

\begin{Corollary}[Local Fisher-Rao compatibility of the topic component]\label{cor:canonical-topic}
 The Riemannian metric tensor induced by the Jensen-Shannon distance $d_{\mathrm{JS}}$ on $\Simplex^{K-1}_+$ is exactly $(1/(8\ln 2))\,\FIM(p)$ (identity~\eqref{eq:jsd-fisher-base2}), that is, proportional to the Fisher information matrix. The topic component $T = 1 - d_{\mathrm{JS}}$ therefore induces a constant multiple, $1/(8 \ln 2)$, of the Fisher-Rao metric tensor on the open simplex: it has the same unparametrized geodesics and the same notion of invariance under sufficient statistics, with arc length rescaled by the square root of that constant and the sectional curvature by its reciprocal (the Fisher-Rao value $1/4$ becoming $2\ln 2$). Any other smooth Riemannian distance whose metric tensor is proportional to $\FIM(p)$ shares these geodesics; at finite distances, however, it need not be a fixed function of $d_{\mathrm{JS}}$---which is itself not a function of the Fisher-Rao arc length, since two pairs with equal Bhattacharyya coefficient (hence equal $\dFR$) can have different $d_{\mathrm{JS}}$ (Section~\ref{sec:prelim-ig}).
\end{Corollary}

The corollary is a \emph{local} statement, in the restricted sense of Section~\ref{sec:prelim-ig}: it identifies the metric \emph{tensor}, and does not claim that $d_{\mathrm{JS}}$ is a monotone function of the Fisher-Rao distance at finite distances. The role of Chentsov's theorem in our framework is therefore to single out the Fisher-Rao geometry as the invariant local geometry on the simplex, and to record that the Jensen-Shannon distance is locally consistent with it; this is what we mean by ``local Chentsov compatibility'' in what follows.

\textbf{Empirical check.}
The corollary is infinitesimal, but the linear relation $d_{\mathrm{JS}} \approx \dFR/\sqrt{2\ln 2}$ it rests on holds at the finite distances of real data (Section~\ref{sec:prelim-ig}; quantified in Section~\ref{sec:exp4}), so the local result is usable in practice. It is not part of the corollary, and the framework does not depend on the relation being exact, since the product decomposition of Proposition~\ref{thm:geodesic} uses $d_{\mathrm{JS}}$ directly.

\textbf{Why $d_{\mathrm{JS}}$ rather than another Fisher-compatible distance.}
Once Theorem~\ref{thm:chentsov} singles out the Fisher-Rao geometry, the
Hellinger distance, the Bhattacharyya distance, and the Jensen-Shannon
distance are all candidate scalar distances that share that local
geometry. The Jensen-Shannon distance $d_{\mathrm{JS}} = \sqrt{\JSD}$
has three operational advantages in the present setting, beyond the proper-metric property already required for Theorem~\ref{thm:metric-gm} (Section~\ref{sec:prelim-metric}). First, its
square $\JSD(p,q) = H(m) - \tfrac{1}{2}[H(p)+H(q)]$ has a direct
information-theoretic reading (mutual information between sample and
label under a uniform mixture; Section~\ref{sec:mi}), which connects
naturally to the notion of a narrative transition cost. Second,
$d_{\mathrm{JS}}$ is bounded and normalized under the base-2 logarithm
($d_{\mathrm{JS}}\in[0,1]$), so that the topic similarity $T = 1 -
d_{\mathrm{JS}}$ lies in $[0,1]$ without additional normalization.
Third,
$\JSD(p,q) = H(m) - \tfrac{1}{2}[H(p)+H(q)]$ with $m=(p+q)/2$ is a
function of the mixture and the per-distribution Shannon entropies, so
it is obtained from the same entropy primitive already used for the
cluster-membership vectors, without the per-coordinate products
$\sqrt{p_k q_k}$ on which the Hellinger and Bhattacharyya distances
rely.

\subsection{Unified Spherical Structure}\label{sec:unified-sphere}
\begin{Proposition}[Product of spheres]\label{prop:spheres}
Under the square-root parameterization $\xi_k = 2\sqrt{p_k}$, the Fisher-Rao manifold $(\Simplex^{K-1}_+, g_{\mathrm{FR}})$ maps isometrically onto the positive orthant of a sphere of radius~2. Both components of the composite coherence metric therefore live on spheres, and the product manifold is $\Sphere^{d-1}\!\times\!\Sphere^{K-1}_+$.
\end{Proposition}

\begin{proof}
 From~\eqref{eq:fisher}, the squared Fisher-Rao line element is $ds^2_{\mathrm{FR}} = \sum_k dp_k^2 / p_k$. Under $\xi_k = 2\sqrt{p_k}$, we have $d\xi_k = dp_k / \sqrt{p_k}$, so $dp_k^2 / p_k = d\xi_k^2$. Thus $ds^2_{\mathrm{FR}} = \sum_k d\xi_k^2$, which is the Euclidean metric restricted to the positive orthant of the sphere $\|\xi\| = 2$ (since $\xi_k > 0$ and $\sum_k \xi_k^2/4 = \sum_k p_k = 1$, so $\|\xi\| = 2$). This is the round metric on a sphere of radius~2---the classical square-root isometry recalled in Section~\ref{sec:prelim-ig} \citep{amari2000methods}, derived here to make the product structure explicit.

 The embedding sphere $\Sphere^{d-1}$ carries the standard round metric. On the product $\Sphere^{d-1}\!\times\!\Sphere^{K-1}_+$, the product metric is $ds^2 = ds^2_{\Sphere^{d-1}} + ds^2_{\Sphere^{K-1}_+}$, which is the sum of two angular distances squared.
\end{proof}

\section{Why the Geometric Mean}\label{sec:geomean}
This section makes the case for the geometric mean. We situate it on the compensability spectrum of candidate combinators (Section~\ref{sec:candidates}), establish its distinguishing structural properties---log-additivity and scale invariance (Section~\ref{sec:five-arguments})---give an axiomatic characterization that singles it out under four stated conditions (Section~\ref{sec:axioms}), and close with a property of extraction output across the spectrum (Section~\ref{sec:gap-profile}).

\subsection{Candidate Combinators}\label{sec:candidates}
We consider the following combinators $F\colon [0,1]^2\to[0,1]$:
\begin{align}
 C_{\AM} &= \tfrac{1}{2}(A + T)      & &\text{(arithmetic mean),}\\
 C_{\HM} &= 2AT/(A+T)           & &\text{(harmonic mean),}\\
 C_{\GM} &= \sqrt{AT}            & &\text{(geometric mean),}\\
 C_{\mathrm{Quad}} &= 1 - d_\times^2 = \tfrac{A(2{-}A) + T(2{-}T)}{2}  & &\text{(metric-squared complement),}\\
 C_{\min} &= \min(A, T)           & &\text{(minimum),}\\
 C_{\max} &= \max(A, T)           & &\text{(maximum).}
\end{align}
More generally, the \emph{power mean} family $M_\alpha(A,T) = ((A^\alpha + T^\alpha)/2)^{1/\alpha}$ includes AM ($\alpha=1$), HM ($\alpha=-1$), and GM ($\alpha\to 0$) as special cases, with Min ($\alpha\to-\infty$) and Max ($\alpha\to+\infty$) as limits. The parameter~$\alpha$ controls a single axis, the \emph{degree of compensability}, between the two channels. The Quad combinator $C_{\mathrm{Quad}} = 1 - d_\times^2$ is \emph{not} a power mean; it arises directly from the product metric and is treated separately below.

\textbf{The compensability spectrum.}
Among the power-mean combinators, as $\alpha$ increases, a high score in one channel increasingly compensates for a low score in the other:
\[
 \underbrace{C_{\min}}_{\alpha\to-\infty}
 \;<\; \underbrace{C_{\HM}}_{\alpha=-1}
 \;<\; \underbrace{C_{\GM}}_{\alpha\to 0}
 \;<\; \underbrace{C_{\AM}}_{\alpha=1}
 \;<\; \underbrace{C_{\max}}_{\alpha\to+\infty}
\]
for all $A \ne T$ in $(0,1)$. At the extremes, Min enforces a \emph{bottleneck} (both channels must be good, AND-semantics), while Max enforces \emph{sufficiency} (either channel being good is enough, OR-semantics). Each intermediate value of~$\alpha$ corresponds to a distinct ontological claim about the relationship between angular similarity and topic similarity.
\begin{itemize}[nosep]
 \item \textbf{HM} ($\alpha = -1$). Appropriate when $A$ and $T$ are \emph{rates} (e.g., averaging speeds over equal distances). Neither coherence nor transition probability is naturally a rate. Moreover, HM exhibits pathological zero-dominance ($\HM(0.9, 0.01) = 0.020$), collapsing to near-Min behavior for asymmetric scores.
 \item \textbf{GM} ($\alpha \to 0$). Appropriate when $A$ and $T$ are \emph{multiplicatively independent}, so that the joint ``pass-through'' quality is $A \cdot T$. This is the natural model when both quantities behave like probabilities or signal transmittances.
 \item \textbf{AM} ($\alpha = 1$). Appropriate when $A$ and $T$ contribute \emph{additively} to quality, with one unit of angular similarity perfectly substitutable for one unit of topic similarity. This implicitly assumes commensurable units, a strong claim that is hard to justify when the two channels measure different geometric quantities.
\end{itemize}

\textbf{Quad as the product-metric complement.}
The Quad combinator $C_{\mathrm{Quad}} = 1 - d_\times^2 = \frac{A(2-A) + T(2-T)}{2}$ is \emph{not} a member of the power-mean family. It arises directly from the product metric: since $d_\times$ is the natural $\ell^2$-distance on the product manifold (Theorem~\ref{thm:metric-gm}), $C_{\mathrm{Quad}} = 1 - d_\times^2$ is the metric-squared complement. Two properties distinguish Quad from the power-mean combinators. First, \emph{functional equivalence with $d_\times$}: since $d_\times^2$ is a monotone transformation of $d_\times$ on $[0,1]$, maximin extraction on $C_{\mathrm{Quad}}$ yields the same paths as minimax on $d_\times$ (maximin extraction is invariant under monotone transforms of the edge weights; Corollary~\ref{cor:maximin-scale-inv}). In other words, Quad is the product metric in similarity-domain clothing. Second, \emph{additive separability}: $C_{\mathrm{Quad}} = f(A) + f(T)$ with $f(S) = S(2-S)/2$, so the total cost $\sum_e (1 - C_{\mathrm{Quad},e})$ decomposes into per-channel contributions for sum-based algorithms such as Dijkstra shortest paths and linear programs (LPs), analogous to the log-additive decomposition of GM (Proposition~\ref{prop:log-additive}) but in linear rather than log space.

Despite this geometric legitimacy, Quad operates in the Euclidean domain while the component similarities $A$ and $T$ are bounded in $[0,1]$ and behave like probabilities, for which the natural geometry is log-Euclidean. The geometric mean is the natural mean in this geometry: $\log \GM(A,T) = \frac{1}{2}(\log A + \log T)$, i.e., the arithmetic mean in log-space. Quad computes the Euclidean complement of a non-Euclidean quantity. This is reflected empirically: $C_{\GM}$ achieves the highest rank correlation with $d_\times$ among the power-mean combinators ($\rho = 0.999$, Table~\ref{tab:triangle}; Quad attains $\rho = 1.000$ trivially, as a monotone function of $d_\times$), despite not being derived from $d_\times$, and the log-transform better captures the product-manifold geometry than the linear complement does.

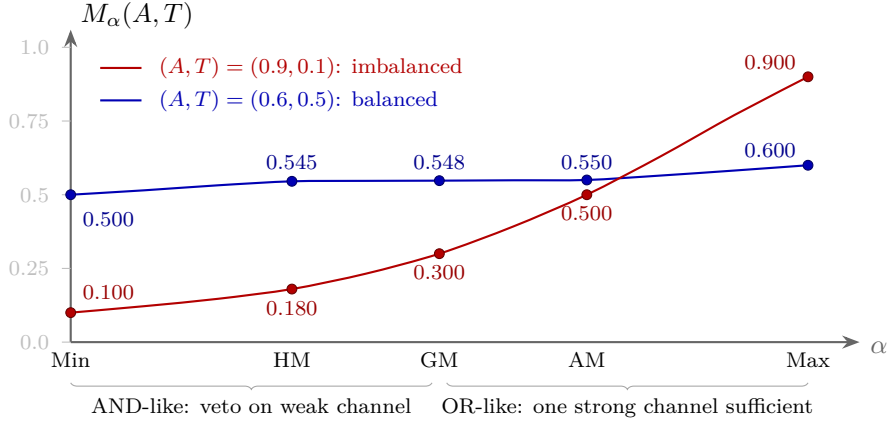
\begin{figure}[t]
 \centering
%
\begin{tikzpicture}[
  x=1.95cm, y=3.9cm, font=\small,
  axis/.style={->, >=Stealth, thick, gray!80!black},
  rmark/.style={circle, fill=red!70!black, draw=red!40!black, inner sep=1.3pt},
  bmark/.style={circle, fill=blue!70!black, draw=blue!40!black, inner sep=1.3pt},
  rlab/.style={font=\scriptsize, text=red!60!black, inner sep=1pt},
  blab/.style={font=\scriptsize, text=blue!55!black, inner sep=1pt},
]
  \def\xMin{-2.5}
  \def\xMax{2.5}
  \def\yMax{1.0}

  \draw[axis] (\xMin, 0) -- (\xMax+0.35, 0) node[below=2pt, right] {$\alpha$};
  \draw[axis] (\xMin, 0) -- (\xMin, \yMax+0.06);
  \node[anchor=south west] at (\xMin, \yMax+0.03) {$M_\alpha(A,T)$};

  \foreach \yv in {0.0, 0.25, 0.5, 0.75, 1.0} {
    \draw[gray!50] (\xMin,\yv) -- ({\xMin-0.06},\yv)
      node[left=1pt, font=\scriptsize] {$\yv$};
  }

  \node[font=\scriptsize, anchor=north] at (\xMin, 0) {Min};
  \node[font=\scriptsize, anchor=north] at (-1, 0) {HM};
  \node[font=\scriptsize, anchor=north] at ( 0, 0) {GM};
  \node[font=\scriptsize, anchor=north] at ( 1, 0) {AM};
  \node[font=\scriptsize, anchor=north] at (\xMax, 0) {Max};

  \draw[blue!70!black, thick, smooth]
    plot coordinates {
      (-2.5, 0.500) (-2.0, 0.515) (-1.5, 0.532) (-1.0, 0.5455)
      (-0.5, 0.547) ( 0.0, 0.5477) ( 0.5, 0.549) ( 1.0, 0.550)
      ( 1.5, 0.565) ( 2.0, 0.582) ( 2.5, 0.600)
    };

  \draw[red!70!black, thick, smooth]
    plot coordinates {
      (-2.5, 0.100) (-2.0, 0.122) (-1.5, 0.147) (-1.0, 0.180)
      (-0.5, 0.232) ( 0.0, 0.300) ( 0.5, 0.392) ( 1.0, 0.500)
      ( 1.5, 0.630) ( 2.0, 0.770) ( 2.5, 0.900)
    };

  \node[bmark] at (-2.5, 0.500) {};  \node[blab, anchor=north west] at (-2.44, 0.452) {$0.500$};
  \node[bmark] at (-1.0, 0.5455) {}; \node[blab, anchor=south]      at (-1.0,  0.575) {$0.545$};
  \node[bmark] at ( 0.0, 0.5477) {}; \node[blab, anchor=south]      at ( 0.0,  0.575) {$0.548$};
  \node[bmark] at ( 1.0, 0.550) {};  \node[blab, anchor=south]      at ( 1.0,  0.578) {$0.550$};
  \node[bmark] at ( 2.5, 0.600) {};  \node[blab, anchor=south east] at ( 2.44, 0.618) {$0.600$};

  \node[rmark] at (-2.5, 0.100) {};  \node[rlab, anchor=south west] at (-2.44, 0.135) {$0.100$};
  \node[rmark] at (-1.0, 0.180) {};  \node[rlab, anchor=north]      at (-1.0,  0.150) {$0.180$};
  \node[rmark] at ( 0.0, 0.300) {};  \node[rlab, anchor=north]      at ( 0.0,  0.270) {$0.300$};
  \node[rmark] at ( 1.0, 0.500) {};  \node[rlab, anchor=north]      at ( 1.0,  0.470) {$0.500$};
  \node[rmark] at ( 2.5, 0.900) {};  \node[rlab, anchor=south east] at ( 2.44, 0.917) {$0.900$};

  \draw[red!70!black, thick] (-2.30, 0.93) -- (-2.02, 0.93);
  \node[red!70!black, font=\scriptsize, anchor=west] at (-1.97, 0.93)
        {$(A,T) = (0.9, 0.1)$: imbalanced};
  \draw[blue!70!black, thick] (-2.30, 0.82) -- (-2.02, 0.82);
  \node[blue!70!black, font=\scriptsize, anchor=west] at (-1.97, 0.82)
        {$(A,T) = (0.6, 0.5)$: balanced};

  \draw[gray!70, decorate, decoration={brace, mirror, raise=2pt}]
        (\xMin, -0.12) -- (-0.05, -0.12)
        node[midway, below=4pt, font=\scriptsize, text=black]
        {AND-like: veto on weak channel};
  \draw[gray!70, decorate, decoration={brace, mirror, raise=2pt}]
        (0.05, -0.12) -- (\xMax, -0.12)
        node[midway, below=4pt, font=\scriptsize, text=black]
        {OR-like: one strong channel sufficient};
\end{tikzpicture}
 \caption{The compensability spectrum of power-mean combinators $M_\alpha$ evaluated at two example similarity pairs. The horizontal axis runs over the spectrum from $\mathrm{Min} = \min(A,T)$ (as $\alpha\to-\infty$) through HM, GM, and AM to $\mathrm{Max} = \max(A,T)$ (as $\alpha\to+\infty$), and each curve is drawn over the full coherence range it attains. For a balanced pair $(A,T) = (0.6, 0.5)$ (blue) the output stays within the narrow band $[0.50, 0.60]$ across the whole spectrum, so the choice of combinator hardly matters; for an imbalanced pair $(A,T) = (0.9, 0.1)$ (red) it sweeps the full range, from the minimum $0.10$ through $0.18$ (HM), $0.30$ (GM), and $0.50$ (AM) to the maximum $0.90$, so the choice of combinator is operationally consequential. The geometric mean sits at $\alpha = 0$ and balances the two regimes.\label{fig:compensability-spectrum}}
\end{figure}

Figure~\ref{fig:compensability-spectrum} illustrates how the spectrum behaves on a balanced and an imbalanced similarity pair. The geometric mean's intermediate position is what permits the veto property at the imbalanced end while preserving meaningful gradation in the balanced regime.

\subsection{Log-Additivity and Scale Invariance}\label{sec:five-arguments}
\begin{Proposition}[Log-additivity]\label{prop:log-additive}
 The geometric mean is the unique combinator in the power-mean family such that $-\log C$ decomposes additively:
 \begin{equation}
  -\log C_{\GM} = -\tfrac{1}{2}\log A - \tfrac{1}{2}\log T\,.
 \end{equation}
 This additive, cross-term-free structure parallels that of the Riemannian product metric $g = g_1 \oplus g_2$ on $\Sphere^{d-1}\times\Simplex^{K-1}_+$, whose line element satisfies $ds^2 = ds^2_1 + ds^2_2$: in both, each factor contributes a separate term, with no cross-channel interaction. The correspondence is structural---each cost $-\log(1-d_i)$ is first-order, not quadratic, in its geodesic distance (Remark~\ref{rem:infinitesimal})---and is precisely the geometric consistency between $C_{\GM}$ and $d_\times$ discussed after Theorem~\ref{thm:metric-gm}; the combinator derived from $d_\times$ itself is Quad (Section~\ref{sec:candidates}).
\end{Proposition}

\begin{proof}
 Throughout, $A, T \in (0,1]$, so that all logarithms are defined. We first verify that the geometric mean is log-additive. The member at $\alpha = 0$ is defined by continuity, and the limit is classical \citep{bullen2003handbook}: applying L'H\^opital's rule to $\log M_\alpha$,
 \[
   \lim_{\alpha \to 0} \frac{\log\bigl((A^\alpha + T^\alpha)/2\bigr)}{\alpha}
   \;=\; \lim_{\alpha \to 0} \frac{A^\alpha \log A + T^\alpha \log T}{A^\alpha + T^\alpha}
   \;=\; \frac{\log A + \log T}{2}\,,
 \]
 so $M_0 = \sqrt{AT} = C_{\GM}$, and $-\log C_{\GM} = -\tfrac{1}{2}\log A - \tfrac{1}{2}\log T$ is additive.

 Next we show that no other member of the family is log-additive. For $\alpha \ne 0$,
 \[
   -\log M_\alpha \;=\; -\tfrac{1}{\alpha}\log\frac{A^\alpha + T^\alpha}{2}\,,
 \]
 which decomposes into a sum of separate functions of $\log A$ and $\log T$ if and only if $A^\alpha + T^\alpha$ factors as $F(A)\,G(T)$ for some functions $F, G$. Suppose such a factorization holds for all $A, T \in (0,1]$. Setting $T = 1$ and $A = 1$ in turn gives $F(A)\,G(1) = A^\alpha + 1$ and $F(1)\,G(T) = 1 + T^\alpha$, with $F(1)\,G(1) = 2 \ne 0$, whence
 \[
   A^\alpha + T^\alpha \;=\; F(A)\,G(T) \;=\; \frac{(A^\alpha + 1)(T^\alpha + 1)}{2}\,.
 \]
 Evaluating at $A = T$ gives $4A^\alpha = (A^\alpha + 1)^2$, i.e.\ $(A^\alpha - 1)^2 = 0$ for all $A \in (0,1]$. Hence $A^\alpha \equiv 1$, forcing $\alpha = 0$, a contradiction. Thus no power mean with $\alpha \ne 0$ is log-additive.

 Finally, the limiting members $M_{+\infty} = \max(A,T)$ and $M_{-\infty} = \min(A,T)$ of the spectrum (Section~\ref{sec:candidates}) are excluded as well. Writing $x = -\log A$ and $y = -\log T$ (so $x, y \ge 0$), we have $-\log M_{+\infty} = \min(x,y)$ and $-\log M_{-\infty} = \max(x,y)$. If $\min(x,y) = f(x) + g(y)$, then evaluating at $(x,y) = (0,0)$, $(1,0)$, $(0,1)$ gives $f(0) + g(0) = 0$ and $f(1) + g(0) = f(0) + g(1) = 0$, so $f(1) + g(1) = 0$, contradicting $\min(1,1) = 1$; the same four points rule out $\max(x,y)$, for which they would force $f(1) + g(1) = 2$ against $\max(1,1) = 1$. The geometric mean is therefore the unique log-additive member of the power-mean family.
\end{proof}

\begin{Theorem}[Associated product metric]\label{thm:metric-gm}
 Write $\dang := \dang(d_i, d_j)$ for the angular distance on $\Sphere^{d-1}$ and $d_{\mathrm{JS}} := d_{\mathrm{JS}}(\hat{e}_i, \hat{e}_j)$ for the Jensen-Shannon distance on $\Simplex^{K-1}_+$. The product distance
 \begin{equation}\label{eq:product-metric}
  d_\times(d_i, d_j)
  \;=\;
  \frac{1}{\sqrt{2}}\, \sqrt{ (\dang/\pi)^2 + d_{\mathrm{JS}}^{\,2} \,}
 \end{equation}
 is a proper metric on the document space\footnote{Throughout, ``metric on the document space'' means a metric on the representation space $\Sphere^{d-1}\times\Simplex^{K-1}_+$; on documents it is a pseudometric, since two distinct documents with identical representations lie at distance zero.} taking values in $[0, 1]$: it is the $\ell^2$-product of two proper metrics, normalized by the diameter $\sqrt{2}$ of the unit square $[0, 1]^2$. The coherence $C_{\GM}$ is monotone decreasing in $\dang$ and in $d_{\mathrm{JS}}$ separately. However, $C_{\GM}$ is not a function of $d_\times$ alone: pairs at the same $d_\times$ can have different coherences depending on how the distance is distributed between channels.
\end{Theorem}

\begin{proof}
 Both $\dang/\pi$ and $d_{\mathrm{JS}}$ \citep{endres2003new,osterreicher2003new} are proper metrics on their respective spaces with values in $[0,1]$. The $\ell^2$-product of two metrics is a metric on the Cartesian product (a standard result in metric geometry \citep{burago2001course}), so $\sqrt{(\dang/\pi)^2 + d_{\mathrm{JS}}^{\,2}}$ is a proper metric on the product space. Dividing by $\sqrt{2}$ (the diameter of $[0,1]^2$ under the $\ell^2$ norm) gives $d_\times \in [0,1]$, which is also a proper metric (positive scaling preserves all metric axioms).

 Monotonicity: $C = \sqrt{A\cdot T} = \sqrt{(1-\dang/\pi)(1-d_{\mathrm{JS}})}$. Both $A = 1-\dang/\pi$ and $T = 1-d_{\mathrm{JS}}$ are monotone decreasing in their respective distances, so $C$ is monotone decreasing in each component distance separately.
\end{proof}

\textbf{Interpretation: three levels of relationship between $C_{\GM}$ and $d_\times$.}
The result above shows that $d_\times$ is a proper metric on the document space and that $C_{\GM}$ is monotone in each component distance, but the relationship between $C_{\GM}$ and $d_\times$ is more nuanced than ``$C_{\GM}$ is derived from $d_\times$.'' Three distinct levels of relationship are worth separating, and we use this terminology consistently in the rest of the paper.
\begin{itemize}[leftmargin=*, nosep, topsep=2pt]
 \item \textbf{Monotone compatibility.} The rank correlation $\rho(1 - C_{\GM}, d_\times) = 0.999$ reported in the combinator comparison (Section~\ref{sec:exp2}) is an empirical observation about the ordering of document pairs, not a derivation. Two scalar functions can produce essentially the same ranking while still being mathematically different objects.
 \item \textbf{Geometric consistency.} Both $C_{\GM}$ and $d_\times$ live on the same product manifold $\Sphere^{d-1} \times \Simplex^{K-1}_+$ and respond in compatible ways to changes in the two component distances: the log-additive decomposition of $C_{\GM}$ mirrors the additive Riemannian product metric that $d_\times$ instantiates (the correspondence stated with Proposition~\ref{prop:log-additive}).
 \item \textbf{Metric derivation.} A metric derivation would be a stronger statement, namely that $C_{\GM}$ is a function of $d_\times$ alone. This is \emph{not} the case in our setting, as Theorem~\ref{thm:metric-gm} records. The Quad combinator $1 - d_\times^2$ is the combinator that \emph{is} derived from $d_\times$, and it differs from $C_{\GM}$ in exactly this respect (Section~\ref{sec:candidates}).
\end{itemize}
The high rank correlation between $C_{\GM}$ and $d_\times$ supports monotone compatibility and geometric consistency; it does not establish metric derivation. We will use this distinction below when discussing why $C_{\GM}$ remains the recommended combinator despite $d_\times$ being a proper metric in its own right.

\textbf{Empirical check.}
Theorem~\ref{thm:metric-gm} is exact, so the quantity worth measuring is the one it does \emph{not} fix: the rank correlation $\rho(1 - C_{\GM}, d_\times) = 0.999$ (the combinator comparison, Section~\ref{sec:exp2}), which measures how closely $C_{\GM}$ tracks the product metric (monotone compatibility). This tracking is \emph{not} implied by the theorem: the theorem holds on any document space, whereas the strength with which $C_{\GM}$ follows $d_\times$---without being a function of it---is a property of the corpora studied.

\textbf{Path cost and diagnostic attribution.}
The log-cost $-\log C$ is the natural edge weight for path extraction. The additive decomposition $-\log C = -\tfrac{1}{2}\log A - \tfrac{1}{2}\log T$ (Proposition~\ref{prop:log-additive}) means that the total cost of any extracted path separates cleanly:
\[
  \sum_{e\in\text{path}} \bigl(-\log C_e\bigr)
  \;=\;
  \underbrace{\frac{1}{2}\sum_{e\in\text{path}} \bigl(-\log A_e\bigr)}_{\text{angular cost}}
  \;+\;
  \underbrace{\frac{1}{2}\sum_{e\in\text{path}} \bigl(-\log T_e\bigr)}_{\text{topic cost}}.
\]
This enables \emph{diagnostic attribution}. Given an extracted storyline, one can identify whether its total cost is dominated by semantic dissimilarity or topical divergence, and localize expensive edges to one channel or the other. The proper metric $d_\times$ is retained for the formal metric-space guarantees needed by triangle-inequality-dependent constructions.

The \emph{operational} force of log-additivity depends on the extraction algorithm. For sum-based extractors (Dijkstra shortest paths, or the linear-programming formulations with coverage and connectivity constraints used in structured narrative extraction \citep{shahaf2010connecting,shahaf2012connecting,keith2020narrative}), $-\log C$ serves as an additive edge cost and the per-channel budgets above directly shape which paths are optimal. For bottleneck-based extractors (maximin/minimax), any monotone transformation of $C$ preserves path optimality (Corollary~\ref{cor:maximin-scale-inv}), so log-additivity is irrelevant to path \emph{selection} but remains valuable as a post-hoc diagnostic.

A heuristic information-theoretic intuition for why the log-additive
structure and the geometric-mean summary fit together can be developed
from a rate-distortion analogy under the assumptions of approximate
channel independence and logarithmic per-channel rate-distortion
functions; we record this analogy in the Supplementary Materials
(Section~S6) as a conceptual reading rather than a derivation, since the
geometric mean is already characterized by Theorem~\ref{thm:axioms} at
the level of an exact algebraic identity.

\begin{Proposition}[Scale invariance]\label{prop:scale-inv}
 For $\lambda > 0$, consider the common power rescaling $A' = A^\lambda$, $T' = T^\lambda$ of the similarity components---equivalently, a common multiplicative rescaling of the per-channel costs, $-\log A' = \lambda\,(-\log A)$ and $-\log T' = \lambda\,(-\log T)$. Call a combinator \emph{ranking-invariant} if, for every $\lambda > 0$, the ranking of document pairs it induces on the rescaled components coincides with the ranking it induces on the originals. Within the power-mean spectrum of Section~\ref{sec:candidates}, the ranking-invariant members are exactly the geometric mean and the limits $M_{\pm\infty}$; the geometric mean is the unique ranking-invariant member with finite exponent---equivalently, the unique one that is strictly increasing in both channels.
\end{Proposition}

\begin{proof}
 \emph{Invariance of the geometric mean.} $C_{\GM}' = (A^\lambda T^\lambda)^{1/2} = (AT)^{\lambda/2}$, and $x \mapsto x^{\lambda/2}$ is strictly increasing for $\lambda > 0$, so the ranking by $C_{\GM}'$ coincides with the ranking by $C_{\GM} = \sqrt{AT}$. In cost terms, $-\log C_{\GM}' = \lambda\,(-\log C_{\GM})$ by the log-additivity of Proposition~\ref{prop:log-additive}.

 \emph{Invariance of the limits.} A strictly increasing map commutes with $\min$ and $\max$, so $\min(A^\lambda, T^\lambda) = \min(A,T)^\lambda$ and $\max(A^\lambda, T^\lambda) = \max(A,T)^\lambda$: each rescaled combinator is a strictly increasing transform of the original, and the induced rankings are unchanged. We note that $M_{\pm\infty}$ are ranking-invariant under any common strictly increasing reparameterization, being order statistics of the two channels; neither, however, is strictly increasing in both channels.

 \emph{Non-invariance for finite $\alpha \ne 0$.} For $\lambda > 0$ and $\alpha \ne 0$,
 \begin{equation}\label{eq:exponent-shift}
   M_\alpha(A^\lambda, T^\lambda)
   \;=\; \Bigl(\tfrac{A^{\lambda\alpha} + T^{\lambda\alpha}}{2}\Bigr)^{1/\alpha}
   \;=\; \Bigl[\Bigl(\tfrac{A^{\lambda\alpha} + T^{\lambda\alpha}}{2}\Bigr)^{1/(\lambda\alpha)}\Bigr]^{\lambda}
   \;=\; M_{\lambda\alpha}(A,T)^{\lambda},
 \end{equation}
 so the ranking on rescaled components is the ranking induced by $M_{\lambda\alpha}$. Take $\lambda = 2$, so $\alpha$ and $2\alpha$ are distinct, nonzero, and of the same sign. By the strict power-mean inequality \citep{bullen2003handbook}, for fixed $x \ne y$ in $(0,1)$ the map $q \mapsto M_q(x,y)$ is strictly increasing on $\mathbb{R}$. Fix $q^*$ strictly between $\alpha$ and $2\alpha$ and $t_- < t_+$ in $(0,1)$, and set $a := M_{q^*}(t_-, t_+)$, $e_1 := (a,a)$, $e_2 := (t_-, t_+)$. Then $M_q(e_1) = a$ for every $q$, while $q \mapsto M_q(e_2)$ is strictly increasing with $M_{q^*}(e_2) = a$. Since exactly one of $\alpha$, $2\alpha$ lies below $q^*$ and the other above, $M_\alpha$ and $M_{2\alpha}$ rank the pair $\{e_1, e_2\}$ oppositely, so the ranking is not invariant under $\lambda = 2$. As a concrete illustration with the arithmetic mean ($\alpha = 1$, $\lambda = 2$): the pairs $(A,T) = (0.9, 0.1)$ and $(0.52, 0.52)$ satisfy $C_{\AM} = 0.50 < 0.52$ before rescaling but $C_{\AM}' = 0.41 > 0.27$ after---a strict reversal (for $\lambda > 1$, convexity of $x^\lambda$ shifts the AM ranking toward imbalanced pairs; for $\lambda < 1$, concavity shifts it toward balanced ones).
\end{proof}

\begin{Corollary}[Scale invariance of maximin extraction]\label{cor:maximin-scale-inv}
 Consider the uniform rescaling $A' = A^\lambda$, $T' = T^\lambda$ with $\lambda > 0$ of Proposition~\ref{prop:scale-inv}. On every coherence graph, the maximin (widest-path) preorder on storylines induced by the geometric mean---together with its set of optimal storylines and the bottleneck edge of every path---is invariant under this rescaling for every $\lambda > 0$, and the same holds for the limits $M_{\pm\infty}$. Conversely, for every power mean $M_\alpha$ with finite $\alpha \ne 0$ there exist a coherence graph and an exponent $\lambda > 0$ for which the rescaling changes the extracted maximin storyline. The geometric mean is thus the unique finite-exponent power mean whose maximin extraction is scale-invariant in this sense---equivalently, the unique member of the spectrum that combines this invariance with strict monotonicity in both channels.
\end{Corollary}

\begin{proof}
 \emph{Step 1 (maximin extraction is ordinal).} Let $G$ be a coherence graph with edge-coherence map $w\colon E \to [0,1]$; the maximin score of a storyline $\pi$ is $\sigma(\pi) = \min_{e\in\pi} w(e)$, and the extractor returns a storyline of maximal score between the prescribed endpoints. If $\varphi\colon [0,1]\to[0,1]$ is strictly increasing, then $\min_{e\in\pi}\varphi(w(e)) = \varphi\bigl(\min_{e\in\pi} w(e)\bigr)$ for every path, since a strictly increasing map commutes with $\min$. Hence $\sigma(\pi)\ge\sigma(\pi') \iff \varphi(\sigma(\pi))\ge\varphi(\sigma(\pi'))$: replacing $w$ by $\varphi\circ w$ leaves the maximin preorder on storylines, its set of maximizers, and the bottleneck edge $\arg\min_{e\in\pi} w(e)$ of every path unchanged.

 \emph{Step 2 (invariance of $M_0$ and $M_{\pm\infty}$).} For each of these three combinators the rescaling sends every edge coherence to its $\lambda$-th power: $M_0(A^\lambda, T^\lambda) = (A^\lambda T^\lambda)^{1/2} = M_0(A,T)^\lambda$, and, since a strictly increasing map commutes with order statistics, $\min(A^\lambda, T^\lambda) = \min(A,T)^\lambda$ and $\max(A^\lambda, T^\lambda) = \max(A,T)^\lambda$ (as in the proof of Proposition~\ref{prop:scale-inv}). In all three cases the rescaled weight map is $\varphi \circ w$ with $\varphi(x) = x^\lambda$ strictly increasing, so by Step 1 the maximin preorder, the set of optimal storylines, and every bottleneck edge are unchanged, on every coherence graph.

 \emph{Step 3 (non-invariance for finite $\alpha \ne 0$).} Fix a finite $\alpha \ne 0$ and take $\lambda = 2$. By the proof of Proposition~\ref{prop:scale-inv}, there exist channel pairs $e_1 = (a,a)$ and $e_2 = (t_-, t_+)$ with components in $(0,1)$ that $M_\alpha$ and $M_{2\alpha}$ rank oppositely; since their components lie in $(0,1)$, $M_q(e_i) < 1$ for every exponent $q$. Build $G$ on temporally ordered vertices $s < m_1 < m_2 < t$ with exactly two $s$--$t$ storylines $\pi_i = (s, m_i, t)$, assigning the channel pair $e_i$ to edge $(s, m_i)$ and $(1,1)$ to edge $(m_i, t)$. Because $M_q(1,1) = 1 > M_q(e_i)$, the bottleneck of $\pi_i$ is its first edge, so the extractor under $M_\alpha$ returns the storyline with the larger $M_\alpha(e_i)$. After the rescaling $A' = A^2$, $T' = T^2$, the exponent-shift identity \eqref{eq:exponent-shift} gives edge coherences $M_{2\alpha}(\cdot)^2$, and by Step 1 (with $\varphi(x) = x^2$) the extraction is governed by the $M_{2\alpha}$-ordering of $\{e_1, e_2\}$; since that ordering is opposite to the $M_\alpha$-ordering, the extracted storyline differs before and after the rescaling.

 Steps 2 and 3 together establish the corollary: within the spectrum, maximin extraction is rescaling-invariant exactly for $M_0$ and $M_{\pm\infty}$, and among these only the geometric mean is strictly increasing in both channels.
\end{proof}

Operationally, Corollary~\ref{cor:maximin-scale-inv} says that the storylines extracted under the geometric mean are insensitive to a uniform multiplicative rescaling of the per-channel costs---the otherwise arbitrary calibration of how strongly raw dissimilarity is penalized---and that it is the only power mean responsive to both channels with this property. The limits Min and Max share the insensitivity, but for the degenerate reason that order statistics commute with any common monotone recalibration, at the price of responding to only one channel at each comparison; for every other combinator the cost scale must be fixed before the extraction output is well-determined.

\subsection{Conditional Axiomatic Characterization}\label{sec:axioms}
The geometric mean is the unique combinator consistent with the four axioms stated below. The characterization is conditional on those axioms; the discussion after the proof identifies which axioms do the most work and what is admitted when each one is dropped.

\begin{Theorem}[Axiomatic characterization]\label{thm:axioms}
 Let $F\colon [0,1]^2 \to [0,1]$ be a coherence combinator
 satisfying:
 \begin{enumerate}
  \item[\textbf{A1}] \textup{(Boundary/veto)} $F(0,t) = F(t,0) = 0$ for all $t\in[0,1]$, and $F(1,1)=1$.
  \item[\textbf{A2}] \textup{(Symmetry)} $F(a,t) = F(t,a)$ for all $a,t\in(0,1]$.
  \item[\textbf{A3}] \textup{(Log-additivity)} There exist functions $f,g\colon[0,\infty)\to[0,\infty)$ such that, for all $a,t\in(0,1]$, $-\log F(a,t) = f(-\log a) + g(-\log t)$ (in particular $F>0$ on $(0,1]^2$, so the left-hand side is finite).
  \item[\textbf{A4}] \textup{(Normalization)} $F(a,a)=a$ for all
   $a\in[0,1]$.
 \end{enumerate}
 Then $F(a,t) = \sqrt{a\cdot t}$.
\end{Theorem}

\begin{proof}
 From \textbf{A3}, letting $x=-\log a$ and $y=-\log t$ (so $x,y\ge 0$), we have $-\log F = f(x) + g(y)$ with $f,g\colon[0,\infty)\to[0,\infty)$, so in particular $f,g\ge 0$. The corner value $F(1,1)=1$ from \textbf{A1} gives $f(0)+g(0)=0$, and since $f,g\ge 0$ this forces $f(0)=g(0)=0$.

 From \textbf{A2}, $f(x)+g(y)=f(y)+g(x)$ for all $x,y\ge 0$, which implies $f-g$ is constant. Since $f(0)-g(0)=0$, we conclude $f=g$.

 From \textbf{A4}, $F(a,a)=a$ for all $a\in(0,1]$, so $-\log a = f(-\log a) + g(-\log a) = 2f(-\log a)$. Dividing both sides by~2: $f(x) = x/2$ for all $x\ge 0$.

 Substituting: $-\log F(a,t) = \frac{-\log a}{2} + \frac{-\log t}{2} = -\frac{1}{2}\log(at)$, so $F(a,t) = (at)^{1/2} = \sqrt{at}$ on the open square $(0,1]^2$. The boundary axiom \textbf{A1} then fixes the axes, $F(0,t)=F(t,0)=0=\sqrt{0\cdot t}$, so $F(a,t)=\sqrt{at}$ on all of $[0,1]^2$.
\end{proof}

\textbf{Scope of the characterization.}
Theorem~\ref{thm:axioms} is stated for an arbitrary $F\colon[0,1]^2\to[0,1]$, so axioms \textbf{A1}--\textbf{A4} single out the geometric mean among \emph{all} coherence combinators---not only within the power-mean compensability spectrum used to introduce it. That spectrum (Min, HM, GM, AM, Max) situates the geometric mean among familiar combinators and grounds the compensability discussion of Section~\ref{sec:candidates}, but the uniqueness statement itself imposes no power-mean restriction.

\textbf{No regularity is assumed.}
The characterization requires no continuity, monotonicity, or measurability of $f$ and $g$. Functional-equation characterizations of means---of Cauchy, Pexider, or bisymmetry type \citep{aczel1989functional}---typically impose such hypotheses to exclude pathological (e.g., non-measurable) solutions; here the value $f(x)$ is pinned pointwise for every $x \ge 0$ (axiom~\textbf{A4} on the diagonal, with $f=g$ from~\textbf{A2}, gives $2f(-\log a) = -\log a$, and $a \mapsto -\log a$ maps $(0,1]$ onto $[0,\infty)$) rather than left to satisfy a functional equation, so pathological solutions cannot arise.

\textbf{Interpretation: each axiom and its role.}
Each axiom has independent motivation. The boundary/veto condition (\textbf{A1}) enforces that a zero in either component yields zero coherence while perfect similarity on both yields one\footnote{The clause $F(1,1)=1$ is implied by axiom~\textbf{A4} (at $a=1$); it is retained in~\textbf{A1} so that the boundary axiom remains self-contained when \textbf{A4} is ablated (Table~\ref{tab:axiom-ablation}).}---a single failed channel cannot be compensated. Symmetry (\textbf{A2}) reflects that the two representation spaces play symmetric roles in our setting. Log-additivity (\textbf{A3}) is needed because sum-based narrative extraction algorithms (Dijkstra on $-\log C$ edge costs, LP formulations) require the total path cost to decompose into independent per-factor contributions. Normalization (\textbf{A4}) is a calibration requirement: when both components agree on a value $a$, the composite must equal $a$.

\textbf{Which axioms are load-bearing.}
Each of the four axioms does independent work: dropping any one admits a combinator that is not the geometric mean, as Table~\ref{tab:axiom-ablation} records. Log-additivity (\textbf{A3}) and normalization (\textbf{A4}) carry the most weight---\textbf{A3} alone cuts the entire compensability spectrum down to a single member, and \textbf{A4} fixes the calibration exponent---while symmetry (\textbf{A2}) fixes the equal-weight balance between the two channels and the boundary/veto condition (\textbf{A1}) fixes the behavior on the axes, without which the characterization pins $F$ only on the open square $(0,1]^2$.

\begin{table}[h]
 \centering
 \small
 \caption{Axiom ablation for Theorem~\ref{thm:axioms}. For each row, the axiom in column~1 is dropped while the other three are retained, and column~2 gives at least one combinator that satisfies the remaining axioms but is not the geometric mean. The geometric mean is the unique combinator that satisfies all four.\label{tab:axiom-ablation}}
 \begin{tabular}{@{}l p{0.45\linewidth} p{0.30\linewidth}@{}}
  \toprule
  \textbf{Axiom dropped} & \textbf{A combinator admitted by the three remaining axioms} & \textbf{Role of the dropped axiom} \\
  \midrule
  \textbf{A1} Boundary/veto & The surviving-channel extension $F(0,t)=F(t,0)=t$, equal to $\sqrt{at}$ on the open square but returning the non-zero channel on the axes (only $F(0,0)=0$ is forced, by A4); the veto $F=0$ at a zero channel is no longer enforced. & Fixes the axes (the veto); without it $F$ is determined only on the open square. \\
  \textbf{A2} Symmetry & Weighted GM $F_w(a,t)=a^w t^{1-w}$, any $w \in (0,1)$, $w \neq 1/2$. & Fixes equal weights ($w = 1/2$). \\
  \textbf{A3} Log-additivity & Veto-respecting non-GM power means, e.g.\ $\HM$ ($\alpha=-1$) and the minimum ($\alpha\to-\infty$); i.e.\ $M_\alpha$ with $\alpha<0$. & Picks GM out of the veto-respecting compensability spectrum. \\
  \textbf{A4} Normalization & $F_\beta(a,t)=(at)^{\beta}$, any $\beta > 0$, $\beta \neq 1/2$. & Fixes the exponent to $\beta = 1/2$. \\
  \bottomrule
 \end{tabular}
\end{table}

The empirical weight-sensitivity analysis in the Supplementary Materials (Section~S2) corroborates the role of \textbf{A2}: across cluster-count configurations, the optimal angular--topic weight $w^\ast$ (the exponent of the weighted geometric mean $A^w T^{1-w}$, distinct from the compensability exponent $\alpha$) tracks the equal-weight value $0.5$ closely (Cohen's $d \le 0.25$ away from the degenerate $K = 2$ boundary), so the data are consistent with the symmetric weighting that \textbf{A2} imposes. Theorem~\ref{thm:axioms} is therefore best understood as a \emph{conditional} characterization: it identifies the geometric mean as the unique combinator simultaneously consistent with log-additivity, calibration, symmetry, and the boundary/veto condition. The intent is not to argue that the geometric mean is universally optimal, only that it is the natural choice under the stated operational requirements.

\subsection{The Bottleneck-Gap Profile}\label{sec:gap-profile}
The arguments above are properties of the geometric mean as a combinator. We close Section~\ref{sec:geomean} with a property of extraction \emph{output}: across the compensability spectrum $M_\alpha$, which exponent most sharply separates coherent storylines from incoherent ones, measured on the maximin (bottleneck) objective that narrative extractors optimize?

Fix $\varepsilon\in(0,1)$ and let $D_{\mathrm H},D_{\mathrm R}$ be Borel probability measures on $[\varepsilon,1]^2$---the laws of the channel pair $(A,T)$ of the \emph{bottleneck edge} (the minimum-coherence edge under the geometric mean $M_0$, held fixed across $\alpha$) of, respectively, a coherent and an incoherent storyline. With $M_0=\sqrt{AT}$ the geometric mean, define the \emph{bottleneck-gap profile}
\begin{equation}\label{eq:gap-profile}
 g(\alpha)\;=\;\mathbb{E}_{D_{\mathrm H}}[M_\alpha(A,T)]\;-\;\mathbb{E}_{D_{\mathrm R}}[M_\alpha(A,T)],\qquad\alpha\in\R .
\end{equation}

\begin{Proposition}[Bottleneck-gap profile]\label{prop:gap-profile}
 Let $g=E+O$ be the decomposition of~\eqref{eq:gap-profile} into an even and an odd function of $\alpha$. Write $\mu_\bullet=\mathbb{E}_{D_\bullet}[M_0]$, $\Delta=\log A-\log T$, and $\varphi_\alpha=\log M_\alpha-\log M_0$. Then:
 \begin{enumerate}
  \item[\textup{(a)}] $\varphi_\alpha$ is odd in $\alpha$, and $E(\alpha)=\mathbb{E}_{\mathrm H}[M_0\cosh\varphi_\alpha]-\mathbb{E}_{\mathrm R}[M_0\cosh\varphi_\alpha]$, $O(\alpha)=\mathbb{E}_{\mathrm H}[M_0\sinh\varphi_\alpha]-\mathbb{E}_{\mathrm R}[M_0\sinh\varphi_\alpha]$.
  \item[\textup{(b)}] $E(0)=\mu_{\mathrm H}-\mu_{\mathrm R}$, and $E(0)-E(\alpha)=\mathbb{E}_{\mathrm R}[Y_\alpha]-\mathbb{E}_{\mathrm H}[Y_\alpha]$ with $Y_\alpha=M_0(\cosh\varphi_\alpha-1)\ge 0$. Hence the even part is maximized at $\alpha=0$ (the geometric mean) if and only if $\mathbb{E}_{\mathrm R}[Y_\alpha]\ge\mathbb{E}_{\mathrm H}[Y_\alpha]$ for every $\alpha$.
  \item[\textup{(c)}] $g$ is smooth, with $g'(0)=\tfrac18\bigl(\mathbb{E}_{\mathrm H}[M_0\Delta^2]-\mathbb{E}_{\mathrm R}[M_0\Delta^2]\bigr)$ and $g''(0)=\tfrac1{64}\bigl(\mathbb{E}_{\mathrm H}[M_0\Delta^4]-\mathbb{E}_{\mathrm R}[M_0\Delta^4]\bigr)$. Suppose $g''(0)<0$, and fix (by continuity) $\delta>0$ and $m>0$ with $g''\le -m$ on $[-\delta,\delta]$. If $|g'(0)|<m\delta$, then $g$ has a unique critical point $\alpha^\ast$ in $(-\delta,\delta)$; it is a strict local maximum and satisfies $|\alpha^\ast|\le|g'(0)|/m$ together with
  \[
   \Bigl|\alpha^\ast+\frac{g'(0)}{g''(0)}\Bigr|
   \;\le\;\frac{L_3}{2m^3}\,g'(0)^2,
   \qquad L_3:=\max_{[-\delta,\delta]}\lvert g'''\rvert .
  \]
  In particular
  \[
   \alpha^\ast=-\frac{g'(0)}{g''(0)}+O\bigl(g'(0)^2\bigr)
   =-\frac{8\bigl(\mathbb{E}_{\mathrm H}[M_0\Delta^2]-\mathbb{E}_{\mathrm R}[M_0\Delta^2]\bigr)}
          {\mathbb{E}_{\mathrm H}[M_0\Delta^4]-\mathbb{E}_{\mathrm R}[M_0\Delta^4]}
    +O\bigl(g'(0)^2\bigr),
  \]
  with implied constant $L_3/(2m^3)$.
 \end{enumerate}
\end{Proposition}

\begin{proof}
 A direct computation gives the reflection identity $M_\alpha M_{-\alpha}=AT=M_0^2$, so $\log M_\alpha+\log M_{-\alpha}=2\log M_0$ and $\varphi_{-\alpha}=-\varphi_\alpha$: $\varphi_\alpha$ is odd in $\alpha$. With $u=\log A,\,v=\log T\in[\log\varepsilon,0]$ one has $\log M_\alpha=K(\alpha)/\alpha$, where $K(\alpha)=\log\bigl(\tfrac12(e^{\alpha u}+e^{\alpha v})\bigr)$ is the cumulant generating function of the uniform law on $\{u,v\}$; $K$ is real-analytic on $\R$ with $K(0)=0$, so $\log M_\alpha$ extends real-analytically across $\alpha=0$ and $M_\alpha$ is real-analytic in $\alpha$. On the compact domain $[\varepsilon,1]^2$, $M_\alpha$ and all its $\alpha$-derivatives are bounded uniformly on compact $\alpha$-intervals, so $g$ is $C^\infty$ and its derivatives may be taken under the integral sign (dominated convergence).

 \emph{(a)} Since $M_{\pm\alpha}=M_0\,e^{\pm\varphi_\alpha}$, equation~\eqref{eq:gap-profile} gives $g(\alpha)=\mathbb{E}_{\mathrm H}[M_0e^{\varphi_\alpha}]-\mathbb{E}_{\mathrm R}[M_0e^{\varphi_\alpha}]$, and $\tfrac12\bigl(g(\alpha)\pm g(-\alpha)\bigr)$ yields the stated $E$ and $O$; $\cosh\varphi_\alpha$ is even and $\sinh\varphi_\alpha$ odd in $\alpha$ because $\varphi_\alpha$ is odd.

 \emph{(b)} $\cosh\varphi_0=1$ gives $E(0)=\mu_{\mathrm H}-\mu_{\mathrm R}$, and $E(0)-E(\alpha)=\mathbb{E}_{\mathrm R}[M_0(\cosh\varphi_\alpha-1)]-\mathbb{E}_{\mathrm H}[M_0(\cosh\varphi_\alpha-1)]$. As $\cosh\ge 1$, $Y_\alpha\ge 0$, so $E(\alpha)\le E(0)$ for all $\alpha$ if and only if $\mathbb{E}_{\mathrm R}[Y_\alpha]\ge\mathbb{E}_{\mathrm H}[Y_\alpha]$ throughout.

 \emph{(c)} From $K(0)=0$ and $K''(0)=\operatorname{Var}\{u,v\}=\tfrac14\Delta^2$, the analytic odd function $\varphi_\alpha$ satisfies $\varphi_\alpha=\dot\varphi\,\alpha+O(\alpha^3)$ with $\dot\varphi=\tfrac12K''(0)=\tfrac18\Delta^2$. Since $E$ is even and $O$ odd, $g'(0)=O'(0)$ and $g''(0)=E''(0)$; differentiating the closed forms of~(a) at $\alpha=0$ (using $\varphi_0=0$, $\varphi'_0=\dot\varphi$, $\varphi''_0=0$) gives $O'(0)=\mathbb{E}_{\mathrm H}[M_0\dot\varphi]-\mathbb{E}_{\mathrm R}[M_0\dot\varphi]$ and $E''(0)=\mathbb{E}_{\mathrm H}[M_0\dot\varphi^2]-\mathbb{E}_{\mathrm R}[M_0\dot\varphi^2]$, which are the stated $g'(0)$, $g''(0)$ after substituting $\dot\varphi=\Delta^2/8$. If $g''(0)<0$, continuity yields $\delta>0$ and $m>0$ with $g''\le -m$ on $[-\delta,\delta]$, so $g'$ is strictly decreasing there. Assume $|g'(0)|<m\delta$. If $g'(0)\ge 0$ then $g'(\delta)\le g'(0)-m\delta<0$; if $g'(0)\le 0$ then $g'(-\delta)\ge g'(0)+m\delta>0$. In either case $g'$ changes sign on $(-\delta,\delta)$ and, being strictly decreasing, has a unique zero $\alpha^\ast$ there, with $g'>0$ to its left and $g'<0$ to its right, so $\alpha^\ast$ is a strict local maximum of $g$. Integrating $g''\le -m$ between $0$ and $\alpha^\ast$ gives $|g'(0)|\ge m|\alpha^\ast|$, i.e.\ $|\alpha^\ast|\le|g'(0)|/m$. Finally, a second-order Taylor expansion of $g'$ gives $0=g'(\alpha^\ast)=g'(0)+g''(0)\,\alpha^\ast+\tfrac12 g'''(\xi)\,\alpha^{\ast 2}$ for some $\xi$ between $0$ and $\alpha^\ast$, whence $\alpha^\ast+g'(0)/g''(0)=-\bigl(g'''(\xi)/2g''(0)\bigr)\alpha^{\ast 2}$; using $|g''(0)|\ge m$ and $|\alpha^\ast|\le|g'(0)|/m$ yields the stated bound.
\end{proof}

\textbf{Interpretation: the geometric mean and extraction output.}
Proposition~\ref{prop:gap-profile} isolates what the geometric mean does for extraction \emph{output}. The even part of the gap---its symmetric, combinator-balanced component---is maximized \emph{exactly} at the geometric mean precisely when $\mathbb{E}_{\mathrm R}[Y_\alpha]\ge\mathbb{E}_{\mathrm H}[Y_\alpha]$, i.e.\ when incoherent storylines carry more channel imbalance than coherent ones. This is not an additional assumption: it is the design premise of the coherence metric itself---a coherent transition is good in \emph{both} channels (balanced), an incoherent one fails at least one (imbalanced; a transition coherent in only one channel is vetoed, $C_{\GM}=0$, Section~\ref{sec:properties}). The odd part only \emph{displaces} the observed peak, by $-g'(0)/g''(0)$, an explicit ratio of imbalance moments of the data. Unlike the axiomatic characterization (Theorem~\ref{thm:axioms}), this argument refers to extracted storylines rather than to the combinator in isolation.

\textbf{Empirical check.}
Experiment~7 (Section~\ref{sec:exp7}) measures $g(\alpha)$ on five corpora. The design-premise condition $\mathbb{E}_{\mathrm R}[Y_\alpha]\ge\mathbb{E}_{\mathrm H}[Y_\alpha]$ holds on every one; the even part is accordingly maximized at the geometric mean, and the observed peak of $g$ itself is interior to the spectrum, displaced from $\alpha=0$ by the odd part. The decomposition $g=E+O$ and the closed forms of Proposition~\ref{prop:gap-profile} are exact; that the even part is maximized at $\alpha=0$ is conditional on the design premise, which Experiment~7 verifies rather than assumes.

\section{Formal Properties}\label{sec:properties}
The compensability ordering of Section~\ref{sec:candidates} is the classical power-mean inequality: for all $A,T\in[0,1]$ (with the convention $2AT/(A+T):=0$ at $A=T=0$),
\begin{equation}\label{eq:pm-chain}
 \min(A,T)\;\le\;\tfrac{2AT}{A+T}\;\le\;\sqrt{AT}\;\le\;\tfrac{A+T}{2}\;\le\;\max(A,T),
\end{equation}
with equality throughout iff $A=T$. This is the monotonicity of $q\mapsto M_q(A,T)$ in the exponent \citep{bullen2003handbook}, evaluated at $q=-\infty,-1,0,1,+\infty$---the same fact used in the proof of Corollary~\ref{cor:maximin-scale-inv}.

\textbf{Interpretation.} The geometric mean occupies a balanced middle ground: stricter than the arithmetic mean (penalizing imbalance between the components) but less harsh than the minimum (which ignores the stronger channel entirely).

\begin{Lemma}[Position of the Quad combinator]\label{lem:quad-position}
 The Quad combinator $C_{\mathrm{Quad}} = 1 - d_\times^2$ is not a power mean and does not join the chain~\eqref{eq:pm-chain} at a fixed position. For all $A,T\in[0,1]$,
 \begin{equation}\label{eq:quad-am}
  C_{\mathrm{Quad}}-C_{\AM}\;=\;\tfrac12\bigl[A(1-A)+T(1-T)\bigr]\;\ge\;0,
 \end{equation}
 so $C_{\GM}\le C_{\AM}\le C_{\mathrm{Quad}}$ everywhere, with $C_{\AM}=C_{\mathrm{Quad}}$ iff $A,T\in\{0,1\}$. Its position relative to the maximum, by contrast, depends on the channel imbalance: writing $S_1=\max(A,T)$ and $S_2=\min(A,T)$,
 \begin{equation}\label{eq:quad-max}
  C_{\mathrm{Quad}}\;\le\;C_{\max}\iff S_1^2+(1-S_2)^2\;\ge\;1,
 \end{equation}
 which holds for sufficiently imbalanced pairs (e.g.\ $(0.9,0.1)$, where $C_{\mathrm{Quad}}=0.59<0.9$) and fails for every balanced interior pair: $A=T\in(0,1)$ gives $C_{\mathrm{Quad}}=A(2-A)>A=C_{\max}$.
\end{Lemma}

\begin{proof}
 With $C_{\mathrm{Quad}}=\tfrac12[A(2-A)+T(2-T)]$, equation~\eqref{eq:quad-am} is immediate and vanishes iff each nonnegative term does. For~\eqref{eq:quad-max}, assume $A\ge T$, so $C_{\max}=A$; then $C_{\mathrm{Quad}}-A=\tfrac12\bigl[T(2-T)-A^2\bigr]=\tfrac12\bigl[1-(1-T)^2-A^2\bigr]$, which is $\le 0$ iff $A^2+(1-T)^2\ge 1$.
\end{proof}

Intuitively, squaring a value in $[0,1]$ shrinks it, so the metric-squared complement $1-d_\times^2$ is the most generous combinator in the table over balanced pairs---and the same squaring, applied to the proper metric $d_\times$, is what breaks the triangle inequality for $1-C_{\mathrm{Quad}}=d_\times^2$ (the $8.21\%$ violation rate in the combinator comparison, Section~\ref{sec:exp2}).

\begin{Proposition}[Metric structure of bounded dissimilarities]\label{prop:bounded-metrics}
 Let $d_1 = \dang/\pi$ and $d_2 = d_{\mathrm{JS}}$ be the normalized component distances. Then:
 \begin{enumerate}[nosep]
  \item $1 - C_{\AM} = (d_1 + d_2)/2$ is a proper metric (scaled $\ell^1$ product);
  \item $1 - C_{\min} = \max(d_1, d_2)$ is a proper metric ($\ell^\infty$ product);
  \item $1 - C_{\GM} = 1 - \sqrt{(1-d_1)(1-d_2)}$ is \emph{not} a metric in general.
 \end{enumerate}
\end{Proposition}

\begin{proof}
 For~(1): $1-\tfrac{A+T}{2}=\tfrac{(1-A)+(1-T)}{2}=\tfrac{d_1+d_2}{2}$; a sum of metrics is a metric, and positive scaling preserves all metric axioms. For~(2): $1-\min(A,T)=1-\min(1-d_1,1-d_2)=\max(d_1,d_2)$, and the $\ell^\infty$-product of metrics is a metric: for any $x,y,z$, $\max_i d_i(x,z)\le\max_i\bigl(d_i(x,y)+d_i(y,z)\bigr)\le\max_i d_i(x,y)+\max_i d_i(y,z)$.

 For~(3), consider first an abstract three-point space $\{x,y,z\}$ with $d_1(x,y)=d_1(y,z)=0.01$, $d_1(x,z)=0.02$ and $d_2(x,y)=0.99$, $d_2(y,z)=0.01$, $d_2(x,z)=1.00$. Both components satisfy the triangle inequality (with equality in $d_2$, which is permitted). Then $D_{\GM}(x,z)=1.000$ while $D_{\GM}(x,y)+D_{\GM}(y,z)=0.901+0.010=0.911<1.000$, violating the triangle inequality; symmetry and the identity axiom hold, so only the triangle inequality fails. The configuration is realized on the actual component spaces: place the three embeddings along a great circle of $\Sphere^{d-1}$ at normalized angular distances $\dang/\pi = 0.01,\,0.01,\,0.02$ (so $A_{xy}=A_{yz}=0.99$, $A_{xz}=0.98$) and take $K=2$ topic distributions $\hat e_x=(1,0)$, $\hat e_y=(0.01,0.99)$, $\hat e_z=(0,1)$. The base-2 Jensen--Shannon distances are then $d_2(x,y)=0.980$, $d_2(y,z)=0.071$, $d_2(x,z)=1$, giving $D_{\GM}(x,y)=0.858$, $D_{\GM}(y,z)=0.041$, $D_{\GM}(x,z)=1.000$ and the violation $1.000>0.858+0.041=0.899$. A sufficiently small perturbation of $\hat e_x$ and $\hat e_z$ into the open simplex $\Simplex^{1}_{+}$ preserves the strict violation by continuity (the margin at the vertex configuration is $\approx 0.10$). In all cases the violation arises when the intermediate point $y$ has extreme channel imbalance: close on one channel, far on the other.
\end{proof}

\textbf{Empirical check.}
In practice, HDBSCAN soft membership concentrates mass on one or two clusters but retains numerically small probabilities across the others, so membership vectors rarely approach the simplex vertices closely enough to produce the extreme channel imbalance the counterexample requires. Accordingly, no violations were detected among $10^5$ sampled triplets in corpora with $K\ge 11$ (the cross-corpus validation, Section~\ref{sec:exp4}); violations appear only under degenerate clustering ($K=2$, where membership vectors concentrate near the simplex vertices).

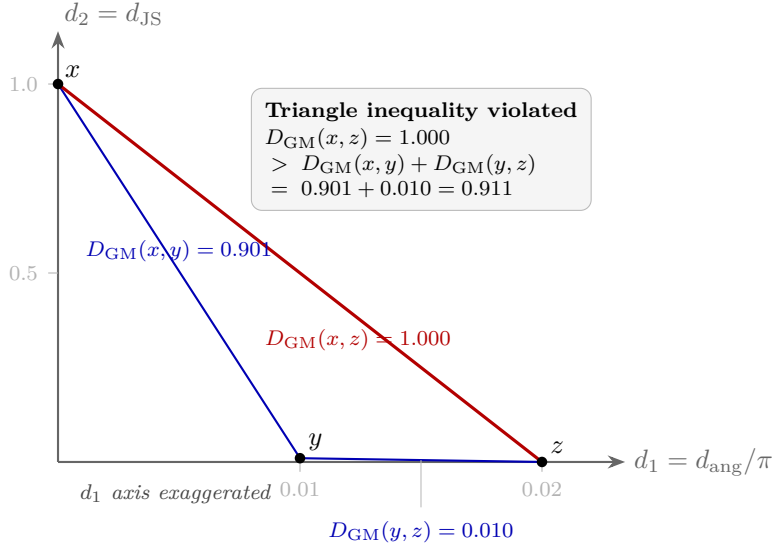
\begin{figure}[t]
 \centering
%
\begin{tikzpicture}[
  font=\small,
  axis/.style={->, >=Stealth, thick, gray!80!black},
  cpt/.style={circle, fill=black, inner sep=1.4pt},
  lead/.style={gray!55, thin},
]
  \draw[axis] (0,0) -- (7.5,0) node[right, inner sep=3pt] {$d_1 = \dang/\pi$};
  \draw[axis] (0,0) -- (0,5.7) node[anchor=south west, inner sep=2pt] {$d_2 = d_{\mathrm{JS}}$};

  \draw[gray!60] (3.2,0) -- (3.2,-0.13) node[anchor=north, font=\scriptsize] {$0.01$};
  \draw[gray!60] (6.4,0) -- (6.4,-0.13) node[anchor=north, font=\scriptsize] {$0.02$};
  \draw[gray!60] (0,2.5) -- (-0.13,2.5) node[anchor=east, font=\scriptsize] {$0.5$};
  \draw[gray!60] (0,5.0) -- (-0.13,5.0) node[anchor=east, font=\scriptsize] {$1.0$};

  \coordinate (x) at (0.00, 5.00);   
  \coordinate (y) at (3.20, 0.05);   
  \coordinate (z) at (6.40, 0.00);   

  \draw[blue!70!black, thick] (x) -- (y);
  \draw[blue!70!black, thick] (y) -- (z);
  \draw[red!70!black, very thick] (x) -- (z);

  \node[blue!70!black, font=\scriptsize, sloped, anchor=south]
        at ($(x)!0.5!(y)$) {$D_{\GM}(x,y) = 0.901$};
  \node[red!70!black, font=\scriptsize, sloped, anchor=north]
        at ($(x)!0.62!(z)$) {$D_{\GM}(x,z) = 1.000$};
  \draw[lead] (4.80,0.025) -- (4.80,-0.60);
  \node[blue!70!black, font=\scriptsize, anchor=north] at (4.80,-0.64)
        {$D_{\GM}(y,z) = 0.010$};

  \node[cpt] at (x) {};
  \node[cpt] at (y) {};
  \node[cpt] at (z) {};
  \node[anchor=south west, inner sep=2.5pt] at (x) {$x$};
  \node[anchor=south west, inner sep=3pt]   at (y) {$y$};
  \node[anchor=south west, inner sep=3pt]   at (z) {$z$};

  \node[anchor=north west, align=left, font=\scriptsize,
        draw=gray!55, rounded corners, fill=gray!8, inner sep=5pt]
    at (2.55,4.95) {%
      \textbf{Triangle inequality violated}\\[1pt]
      $D_{\GM}(x,z) = 1.000$\\
      $\;>\; D_{\GM}(x,y) + D_{\GM}(y,z)$\\
      $\;=\; 0.901 + 0.010 = 0.911$};

  \node[gray!55!black, font=\scriptsize\itshape, anchor=north]
        at (1.55,-0.16) {$d_1$ axis exaggerated};
\end{tikzpicture}
 \caption{Geometry of the triangle-inequality counterexample for $1 - C_{\GM}$. The three points $x, y, z$ are plotted in the normalized component-distance plane $(d_1, d_2) = (\dang/\pi, d_{\mathrm{JS}})$; the $d_1$ axis is shown on a greatly exaggerated scale, since the three points span only $d_1 \in [0, 0.02]$. The intermediate point $y$ is close to both endpoints on the angular channel ($d_1 \approx 0.01$ on both edges) and asymmetric on the topic channel ($d_2(x,y) = 0.99$ but $d_2(y,z) = 0.01$). This channel imbalance is what produces the violation $D_{\GM}(x, z) = 1.000 > 0.911 = D_{\GM}(x, y) + D_{\GM}(y, z)$.\label{fig:triangle-violation}}
\end{figure}

Figure~\ref{fig:triangle-violation} renders the configuration geometrically: the violation requires an intermediate point that is close on one channel and far on the other, which is exactly the channel-imbalance regime that scale complementarity (Remark~\ref{def:complementarity}) and smooth clustering rule out in practice.

\begin{Proposition}[Connection to the Hellinger distance]\label{prop:hellinger}
 Let $H(p,q)=\bigl(\sum_k(\sqrt{p_k}-\sqrt{q_k})^2\bigr)^{1/2}$ be the Hellinger distance and $\mathrm{BC}(p,q)=\sum_k\sqrt{p_k q_k}$ the Bhattacharyya coefficient. Then, exactly and at all distances,
 \begin{equation}\label{eq:hellinger-chord}
  H^2 \;=\; 2\,(1-\mathrm{BC})
  \qquad\text{and}\qquad
  H \;=\; 2\sin\!\bigl(\dFR/2\bigr),
 \end{equation}
 i.e.\ $H$ is the chordal distance subtended by the Fisher--Rao arc under the square-root embedding. Moreover, for all $p,q\in\Simplex^{K-1}$,
 \begin{equation}\label{eq:hellinger-sandwich}
  \frac{H}{2\sqrt{\ln 2}} \;\le\; d_{\mathrm{JS}} \;\le\; \frac{H}{\sqrt{2\ln 2}}\,,
 \end{equation}
 with the upper bound tight to first order as $p\to q$. Consequently, wherever $H<\sqrt{2\ln 2}$, the topic cost of the exact decomposition~\eqref{eq:exact-decomp} satisfies
 \begin{equation}
  -\log\Bigl(1-\tfrac{H}{2\sqrt{\ln 2}}\Bigr)
  \;\le\; -\log\bigl(1-d_{\mathrm{JS}}\bigr) \;\le\;
  -\log\Bigl(1-\tfrac{H}{\sqrt{2\ln 2}}\Bigr),
 \end{equation}
 and to first order in $\dFR$, $-\log(1-d_{\mathrm{JS}})\approx-\log\bigl(1-H/\sqrt{2\ln 2}\bigr)$.
\end{Proposition}

\begin{proof}
 Expanding the square, $H^2=\sum_k p_k-2\sum_k\sqrt{p_kq_k}+\sum_k q_k=2(1-\mathrm{BC})$. By the definition~\eqref{eq:fr-distance} of the Bhattacharyya angle, $\mathrm{BC}=\cos\dFR$, so $H^2=2(1-\cos\dFR)=4\sin^2(\dFR/2)$, proving~\eqref{eq:hellinger-chord}; under $\eta_k=\sqrt{p_k}$ the distributions lie on the unit sphere, where $\dFR$ is the arc and $H$ the chord. For~\eqref{eq:hellinger-sandwich}, the bound $\tfrac12(1-\mathrm{BC})\le\JSD_{\mathrm{nat}}\le 1-\mathrm{BC}$ \citep{topsoe2000some} reads $\tfrac14H^2\le\JSD_{\mathrm{nat}}\le\tfrac12H^2$; dividing by $\ln 2$ and taking square roots gives the claim. Near the diagonal, $\JSD_{\mathrm{nat}}=\tfrac12\dFR^2+O(\dFR^4)$ (identity~\eqref{eq:jsd-fisher}, the cubic term vanishing) while $\tfrac12H^2=\tfrac12\dFR^2+O(\dFR^4)$, so the upper bound is tight to first order. The cost sandwich follows from the monotonicity of $x\mapsto-\log(1-x)$; the lower expression is defined for all $H\le\sqrt2$ since $H/(2\sqrt{\ln 2})\le\sqrt2/(2\sqrt{\ln 2})<1$, the upper for $H<\sqrt{2\ln 2}$.
\end{proof}

\textbf{Interpretation.}
Both $d_{\mathrm{JS}}$ and $H/\sqrt{2\ln 2}$ are first-order in the Fisher--Rao distance with the same constant, and~\eqref{eq:hellinger-sandwich} bounds their finite-distance discrepancy by a factor $\sqrt2$. They are not interchangeable as similarity measures, however: $1-H/\sqrt{2\ln 2}$ becomes negative for near-disjoint supports ($H>\sqrt{2\ln 2}$), so the Hellinger substitute is not $[0,1]$-bounded---concretely the boundedness advantage that Section~\ref{sec:topic-compat} cites in favour of $d_{\mathrm{JS}}$. Combining~\eqref{eq:hellinger-chord} and~\eqref{eq:hellinger-sandwich} with $2\sin(\dFR/2)\le\dFR$ also gives the pointwise exact bound $d_{\mathrm{JS}}\le\dFR/\sqrt{2\ln 2}$: the scatter of Figure~\ref{fig:approx} can only lie on or below the identity line, as observed.

\textbf{Behavior at extremes.}
If $A=0$ (antipodal embeddings) or $T=0$ ($d_{\mathrm{JS}}=1$, maximally different topics), then $C=0$, so the geometric mean enforces that \emph{both} channels must contribute non-trivially. This ``veto'' property is desirable for coherence because a transition that is semantically coherent but topically incoherent (or vice versa) should receive zero coherence. In practice, the veto is redundant for unconstrained maximin extraction, because such paths inherently avoid low-weight edges when alternative routes exist. It does, however, provide a hard safety guarantee for extraction algorithms that impose global constraints (the coverage, flow, and cardinality requirements of the Narrative Maps linear program) which may force otherwise-avoidable edges.

\section{Connections to Classical Information Theory}\label{sec:connections}
This section connects the composite coherence metric to classical information-theoretic concepts: the data processing inequality, and the identification of each per-channel cost with an exact information-theoretic quantity. These connections are not required by the core results of Sections~\ref{sec:product} to~\ref{sec:properties}, but they make precise why the two channels, despite being structurally coupled (topics are derived from embeddings), capture genuinely different information.

\subsection{Multi-Scale Complementarity}\label{sec:dpi}
The pipeline for computing topic distributions involves successive information-reducing transformations:

\begin{equation}
 e_i \in \R^d
 \;\xrightarrow{\;\mathrm{UMAP}\;}
 \tilde{e}_i \in \R^p
 \;\xrightarrow{\;\mathrm{HDBSCAN}\;}
 \hat{e}_i \in \Simplex^{K-1}\,.
\end{equation}

Since the topic channel is derived from the embedding channel via a lossy pipeline, statistical independence between $A$ and $T$ is structurally impossible. The relevant property is not independence but \emph{complementarity}: the two channels should capture information at different \emph{scales of abstraction}, from fine-grained semantic direction ($A$) to coarse-grained thematic membership ($T$), so that combining them yields a richer similarity than either alone.

\begin{Proposition}[Information hierarchy]\label{prop:dpi}
 Let $(e_i, e_j)$ be a pair of document embeddings with arbitrary joint distribution (for instance, a uniformly random pair from the corpus, or the endpoints of a random transition), and let $\tilde{e} = \mathrm{UMAP}(e)$ and $\hat{e} = \mathrm{HDBSCAN}(\tilde{e})$ be the reduced representations and topic distributions, the fitted pipeline being treated as a fixed, measurable, deterministic map (it is fit once on the corpus, with fixed seed, Section~\ref{sec:prelim-metric}; the randomness refers only to the draw of the pair). Then
 \begin{equation}
  I(e_i; e_j) \;\ge\; I(\tilde{e}_i; \tilde{e}_j) \;\ge\;
  I(\hat{e}_i; \hat{e}_j)\,.
 \end{equation}
\end{Proposition}

\begin{proof}
 We use the following consequence of the data processing inequality \citep{cover2006elements}: if $X' = f(X)$ and $Y' = g(Y)$ are deterministic functions, then $X' - X - Y$ and $X' - Y - Y'$ are Markov chains, so applying the inequality twice,
 \[
  I(X'; Y') \;\le\; I(X'; Y) \;\le\; I(X; Y)
 \]
 (all quantities are well defined in $[0,\infty]$ and the inequalities hold there). Applying this with $(X, Y) = (e_i, e_j)$ and $f = g = \mathrm{UMAP}$ gives $I(\tilde{e}_i; \tilde{e}_j) \le I(e_i; e_j)$; applying it again with $(X, Y) = (\tilde{e}_i, \tilde{e}_j)$ and $f = g = \mathrm{HDBSCAN}$ gives $I(\hat{e}_i; \hat{e}_j) \le I(\tilde{e}_i; \tilde{e}_j)$.
\end{proof}

\textbf{Interpretation.} The proposition fixes the \emph{direction} of information flow: the topic channel, computed from the $K$-dimensional clustered representation ($K \ll d$), cannot carry more pairwise information than the angular channel, computed from the full $d$-dimensional embedding---it is a genuine compression. The inequality bounds the \emph{amount} of information retained, not its content; that the discarded information is specifically fine-grained directional structure is an interpretation, while the load-bearing empirical fact is the separately measured non-redundancy of the two channels (the scale-complementarity test, Section~\ref{sec:exp3}).

\begin{Remark}[Scale complementarity as a design principle]\label{def:complementarity}
 Two similarity channels $A$ and $T$ derived from the same underlying representations exhibit \emph{scale complementarity} when three qualitative conditions hold.
 \begin{enumerate}[nosep]
  \item \textbf{Non-redundancy:} $\mathrm{NMI}(A,T) \ll 1$, i.e., the channels share a small fraction of their information content despite the structural coupling. In our experiments, $\mathrm{NMI} \approx 0.03$ (well below $0.05$) across all corpora with $K \ge 11$ (Section~\ref{sec:exp3}).
  \item \textbf{Scale separation:} $T$ is obtained from $A$'s representation space via a many-to-one mapping (compression) that discards fine-grained structure while preserving coarse-grained structure, so the channels capture different levels of abstraction. This is verified structurally by the UMAP$\to$HDBSCAN pipeline (Proposition~\ref{prop:dpi}).
  \item \textbf{Composite informativeness:} The composite $C = F(A,T)$ is not a function of $A$ alone, nor of $T$ alone. For $C_{\GM} = \sqrt{AT}$ this is immediate: two pairs with equal $A$ but different $T$ receive different $C_{\GM}$.
 \end{enumerate}
 Scale complementarity is a \emph{design principle}, not a formal prerequisite: the axiomatic characterization (Theorem~\ref{thm:axioms}) and product metric (Theorem~\ref{thm:metric-gm}) hold at any correlation level. It guides practitioners toward pipelines where the composite is genuinely more informative than either channel alone; the scale-complementarity test (Section~\ref{sec:exp3}) confirms this for the present pipeline.
\end{Remark}

Condition~1 is verified empirically in Section~\ref{sec:exp3}, Condition~2 holds structurally because the pipeline is a deterministic compression (Proposition~\ref{prop:dpi}), and Condition~3 holds for $C_{\GM}$ by the elementary argument above; the joint $(A,T)$ distribution (Figure~\ref{fig:correlation}) further shows the corpus pairs populate a genuinely two-dimensional region rather than collapsing onto a curve $T = f(A)$, so the non-redundancy is exercised in practice and not merely possible in principle.

\subsection{Information-Theoretic Content of the Per-Channel Costs}\label{sec:mi}
The log-coherence cost decomposition
\begin{equation}\label{eq:mi-decomp}
 -\log C_{\GM} \;=\; \tfrac{1}{2}(-\log A) \;+\; \tfrac{1}{2}(-\log T)
\end{equation}
is an algebraic identity (Proposition~\ref{prop:log-additive}) that holds at any correlation level. Each of the two per-channel costs admits an exact information-theoretic identity. Both are classical results; we record them because together they pin down the information-theoretic content of~\eqref{eq:mi-decomp}, in place of an informal reading of $-\log A$ and $-\log T$ as ``rates.''

\textbf{The angular cost.} The angular similarity is exactly a hash-collision probability.

\begin{Proposition}[Angular similarity as collision probability; \citealp{goemans1995improved,charikar2002similarity}]\label{prop:simhash}
 Let $u, v \in \R^d \setminus \{0\}$ be two document embeddings separated by angle $\theta \in [0,\pi]$, so $A = 1 - \theta/\pi$. Let $g \in \R^d$ be isotropically random with $\Pr[g = 0] = 0$ (e.g.\ a standard Gaussian), inducing the one-bit hash $h_g(x) = \operatorname{sign}\langle g, x\rangle$. Then
 \begin{equation}\label{eq:simhash}
  A \;=\; \Pr\nolimits_{g}\bigl[\,h_g(u) = h_g(v)\,\bigr],
 \end{equation}
 the probability that the uniformly random hyperplane $g^\perp$ through the origin does not separate $u$ and $v$.
\end{Proposition}

\begin{proof}
 Conditioning on the radius of $g$, its law is a mixture of uniform measures on spheres, each of which assigns zero mass to any proper subspace; hence $\langle g, u\rangle \ne 0 \ne \langle g, v\rangle$ almost surely and the hashes take values in $\{\pm 1\}$. If $\theta \in \{0, \pi\}$ the identity is immediate: $v$ is a positive (resp.\ negative) multiple of $u$, so the signs agree (resp.\ differ) almost surely, matching $A = 1$ (resp.\ $A = 0$). Assume $0 < \theta < \pi$, so that $\operatorname{span}(u,v)$ is a plane. The signs $\operatorname{sign}\langle g,u\rangle$ and $\operatorname{sign}\langle g,v\rangle$ depend on $g$ only through its orthogonal projection onto this plane, which is almost surely nonzero and whose direction is uniform on the unit circle of the plane: in-plane rotations, extended by the identity on the orthogonal complement, commute with the projection and preserve the law of $g$. The hyperplane $g^\perp$ separates $u$ and $v$ exactly when that direction falls in the symmetric difference of the two half-circles $\{w : \langle w, u\rangle > 0\}$ and $\{w : \langle w, v\rangle > 0\}$, a union of two arcs of length $\theta$ each---total measure $2\theta$ out of $2\pi$. Hence $\Pr[\text{separate}] = \theta/\pi$ and $\Pr\bigl[h_g(u) = h_g(v)\bigr] = 1 - \theta/\pi = A$.
\end{proof}

This is the random-hyperplane (SimHash) identity. It gives the angular cost an exact reading: $-\log A$ is the Shannon surprisal of a hash collision---the ideal codeword length for the event that a random hyperplane fails to separate the two documents.

\textbf{The topic cost.} The Jensen-Shannon divergence underlying the topic channel is exactly a mutual information.

\begin{Proposition}[Topic dissimilarity as mutual information; \citealp{lin1991divergence}]\label{prop:jsd-mi}
 Let $p, q \in \Simplex^{K-1}$ be two topic distributions (with the convention $0\log 0 = 0$). Introduce a latent label $Z \sim \mathrm{Bernoulli}(\tfrac{1}{2})$ and a sample $X \in \{1,\dots,K\}$ drawn from $p$ if $Z = 0$ and from $q$ if $Z = 1$. Then, with base-2 logarithms,
 \begin{equation}\label{eq:jsd-mi}
  \JSD(p,q) \;=\; I(X;Z),
 \end{equation}
 the mutual information between the sample and the label, and $0 \le I(X;Z) \le 1$ bit, with $I(X;Z) = 0$ iff $p = q$ and $I(X;Z) = 1$ iff $p$ and $q$ have disjoint support.
\end{Proposition}

\begin{proof}
 The marginal law of $X$ is the mixture $M = \tfrac{1}{2}p + \tfrac{1}{2}q$, so $I(X;Z) = H(X) - H(X \mid Z) = H(M) - \tfrac{1}{2}H(p) - \tfrac{1}{2}H(q)$, the uniformly-weighted Jensen-Shannon divergence. The bounds follow from $0 \le I(X;Z) \le H(Z) = 1$. The lower bound is attained iff $X \perp Z$; since both values of $Z$ have positive probability, this holds iff the two conditional laws coincide, i.e.\ $p = q$. The upper bound is attained iff $H(Z \mid X) = 0$, i.e.\ $Z$ is almost surely determined by $X$; the posterior $\Pr[Z = 0 \mid X = k] = p_k/(p_k + q_k)$ for $k$ in the support of $M$ lies in $\{0,1\}$ for every such $k$ iff no $k$ has $p_k q_k > 0$, i.e.\ iff the supports are disjoint.
\end{proof}

\begin{Corollary}[Topic cost as distinguishability]\label{cor:topic-cost}
 On $I(X;Z) \in [0,1)$, the topic cost satisfies $-\log T = -\log\bigl(1 - \sqrt{I(X;Z)}\bigr)$, a strictly increasing function of $I(X;Z)$: it is an exact, strictly increasing function of the information a single observation carries about which of the two topic distributions generated it. As $I(X;Z) \to 1$ (disjoint supports), $T \to 0$ and the cost diverges---the information-theoretic face of the veto property (Section~\ref{sec:properties}).
\end{Corollary}

\begin{proof}
 $T = 1 - d_{\mathrm{JS}} = 1 - \sqrt{\JSD} = 1 - \sqrt{I(X;Z)}$ by Proposition~\ref{prop:jsd-mi}, so $-\log T = (g \circ f)\bigl(I(X;Z)\bigr)$ with $f(t) = \sqrt{t}$ and $g(s) = -\log(1-s)$, a composition of the strictly increasing maps $f\colon [0,1) \to [0,1)$ and $g\colon [0,1) \to [0,\infty)$.
\end{proof}

\textbf{Reading the decomposition.} Equation~\eqref{eq:mi-decomp} is thus an identity between three information-theoretic costs: the composite cost is the arithmetic mean of an angular collision-surprisal and a topic cost monotone in a mutual information. The arithmetic mean (not a sum) is forced by the calibration axiom~\textbf{A4}: $A = T = a$ must give composite cost $-\log a$, fixing the prefactor $\tfrac{1}{2}$; summation is instead the aggregation \emph{across edges of a path}, $\sum_e -\log C_e$. That the two identities are of different type---a collision probability and a mutual information---is faithful, not incidental: the angular channel measures geometric agreement of directions and the topic channel distributional distinguishability, the two ``levels of abstraction'' of scale complementarity (Remark~\ref{def:complementarity}). Under scale complementarity (Section~\ref{sec:exp3}) the costs are also approximately uncorrelated, so they can be diagnostically separated along extracted paths (Section~\ref{sec:five-arguments})---a contingent benefit of the corpus statistics, since~\eqref{eq:mi-decomp} holds at any correlation level.

\textbf{A heuristic rate-distortion reading.} The decomposition~\eqref{eq:mi-decomp} can also be read informally through a rate-distortion analogy, in which approximately independent channels with logarithmic rate-distortion functions make $(D_A D_T)^{1/2}$ a natural scalar distortion summary. That analogy rests on two non-factual assumptions (independence and a logarithmic rate-distortion form) and is not a derivation; we record it in the Supplementary Materials (Section~S6). The identities of this subsection, by contrast, are exact.

\subsection{Soft Membership and Topic Resolution}\label{sec:maxent}
\begin{Remark}[Soft membership preserves manifold structure]\label{prop:soft-membership}
 If $\hat{e}_i \in \Simplex^{K-1}_+$ has full support (all $\hat{e}_{ik} > 0$), it lies in the interior of the simplex where the Fisher-Rao metric is well-defined and non-degenerate, $T$ takes values in the full interval $(0, 1]$, and $d_\times$ inherits the full Riemannian structure. By contrast, hard assignment restricts membership vectors to the vertices of $\Simplex^{K-1}$, where the metric tensor degenerates ($1/p_k \to \infty$ as $p_k \to 0$) and $T$ collapses to a binary same/different signal ($d_{\mathrm{JS}} \in \{0,1\}$) that provides no graded similarity information.
\end{Remark}

HDBSCAN's soft membership assignment \citep{campello2013density} produces membership vectors with full support in practice, so the topic channel operates on the interior of the statistical manifold. This is a prerequisite for the product-manifold framework of Section~\ref{sec:product}: hard assignment would collapse the Fisher-Rao geometry entirely. A component ablation in the Supplementary Materials (Section~S4) confirms this empirically: replacing the soft membership vectors with their hard one-hot counterparts collapses the rank correlation between $C_{\GM}$ and the product metric $d_\times$ from $0.999$ to $0.65$, and replacing the Jensen-Shannon distance with the divergence raises the triangle-inequality violation rate from $0\%$ to $7.5\%$.

\section{Experimental Validation}\label{sec:experiments}
We validate the theoretical results with seven experiments on real-world corpora, with further supporting analyses in the Supplementary Materials (Sections~S1--S8).

\subsection{Datasets}\label{sec:dataset}
The primary corpus is a subset of the Cuban news corpus from German et al.
\citep{german2025narrative}: 418~documents with 1536-dimensional embeddings from a language model (GPT-4). Soft cluster membership distributions are computed via the UMAP$\to$HDBSCAN pipeline described in Section~\ref{sec:prelim-metric}. This yields $K$ topic clusters with per-document membership vectors $\hat{e}_i\in\Simplex^{K-1}$. Experiments~1 to 3 use this corpus. Experiment~4 (Section~\ref{sec:exp4}) extends the validation to three additional corpora from the Narrative Maps \citep{keith2020narrative} and Narrative Trails \citep{german2025narrative} repositories, covering different domains, sizes, and embedding models. Experiment~7 (Section~\ref{sec:exp7}) additionally draws on the Wikispeedia navigation corpus \citep{west2012wikispeedia}, which supplies human storyline data absent from the other corpora.

\subsection{Experiment 1: Approximation Quality}\label{sec:exp1}
Proposition~\ref{thm:geodesic} establishes the exact decomposition $-\log C = \tfrac{1}{2}[-\log(1-\dang/\pi) - \log(1-d_{\mathrm{JS}})]$, where the topic term connects to the Fisher-Rao distance via Remark~\ref{rem:fisher-connection}. We test the Fisher identity and the resulting coherence approximation on all $\binom{418}{2} = 87{,}153$ document pairs, restricted to the $80{,}950$ pairs with $C > 10^{-10}$ (i.e., neither channel collapsed to zero). The excluded pairs correspond to documents in disjoint topic clusters where $d_{\mathrm{JS}}$ saturates at~$1$ and $\dFR$ saturates at $\pi/2$. These are boundary artifacts of the metric range, and neither side of the approximation is finite there.

Figure~\ref{fig:approx}(a) directly tests the first-order Fisher metric identity $d_{\mathrm{JS}} \approx \dFR/\sqrt{2\ln 2}$ (eq.~\ref{eq:djs-fisher}). The Pearson correlation is $R = 0.994$, confirming a strong linear relationship. The scatter resolves into two visually distinct branches, which are separated by the minimum Shannon entropy of the two cluster-membership vectors in each pair (the Panel~(a) colormap). The \emph{upper branch}, consisting of pairs where both $p$ and $q$ are diffuse (high minimum entropy, mass spread across most of the $K = 11$ clusters), hugs $y = x$ closely: for these pairs, both distributions sit in the interior of the simplex where the Fisher-Rao metric tensor is well-conditioned, and the Taylor expansion of $\JSD$ around $p = q$ is accurate. The \emph{lower branch}, where at least one of $p, q$ is peaked on a small number of clusters (low minimum entropy), dips systematically below $y = x$: here the distributions approach the simplex vertices, the Fisher-Rao metric tensor begins to degenerate ($1/p_k \to \infty$ as $p_k \to 0$; Remark~\ref{prop:soft-membership}), and the linear approximation $\dFR/\sqrt{2\ln 2}$ overshoots the truncated range of $d_{\mathrm{JS}}$. The horizontal hook at the top, where $d_{\mathrm{JS}}$ has saturated at~$1$ while $\dFR/\sqrt{2\ln 2}$ continues to $\approx 1.33$, is the asymptotic limit of the lower branch: pairs whose topic distributions approach disjoint-support deltas. The two-branch structure is thus the empirical signature of the simplex-boundary degeneration the theory warns about, not an artifact. That the entire scatter lies on or below the diagonal---the upper branch tight against $y = x$, the lower branch below it, none above---is moreover mandated pointwise rather than merely observed: the first-order identity is in truth the one-sided bound $d_{\mathrm{JS}} \le \dFR/\sqrt{2\ln 2}$, exact at all distances (Proposition~\ref{prop:hellinger}). It also explains why the overall Pearson correlation is high despite a visible bias, since the branches rejoin at both endpoints (origin and saturation) and only separate in the middle of the range.

Figure~\ref{fig:approx}(b) evaluates the Fisher substitution~\eqref{eq:fisher-substitution}, which replaces $d_{\mathrm{JS}}$ with $\dFR/\sqrt{2\ln 2}$ in the exact decomposition while keeping the angular term exact. The substitution is only mathematically defined when $\dFR/\sqrt{2\ln 2} < 1$ (otherwise $-\log(1 - \dFR/\sqrt{2\ln 2})$ is undefined), so we restrict Panel~(b) to the $51{,}342$ pairs that satisfy this condition in addition to the $C > 10^{-10}$ filter. The filter retains all $36{,}242$ upper-branch pairs of Panel~(a) but removes about two thirds of the lower-branch pairs, specifically the most extreme ones, where at least one cluster distribution is essentially a delta and $\dFR$ is near $\pi/2$. The remaining $15{,}100$ lower-branch pairs form the vertical tail seen at the upper-left of Panel~(b), where the $-\log(1-\cdot)$ transform amplifies the Panel~(a) bias toward the log-singularity as $\dFR/\sqrt{2\ln 2}$ approaches $1$. The diagonal agreement band in the lower-left is the upper-branch population where the substitution is faithful. The overall Pearson correlation on Panel~(b) is $R = 0.916$, lower than Panel~(a) because the nonlinear transform stretches the lower-branch residual vertically. The exact decomposition~\eqref{eq:exact-decomp} remains perfect at all distances and is what the rest of the paper uses. Panel~(b) is reported as a diagnostic of the first-order identity, not as an alternative computational path.

The exact decomposition~\eqref{eq:exact-decomp} holds to machine precision ($< 10^{-15}$), as it is an algebraic identity rather than an approximation. The infinitesimal expansion~\eqref{eq:geodesic-approx} requires $\dang/\pi$ and $d_{\mathrm{JS}}$ to be small, but in our corpus, typical angular distances are $\dang/\pi \approx 0.35$, well outside this regime. The practical strength of the exact decomposition is that it provides the additive cost structure (angular cost plus topic cost) at \emph{all} distances, not just infinitesimally.

\begin{figure*}[t]
 \centering
 \includegraphics[width=\textwidth]{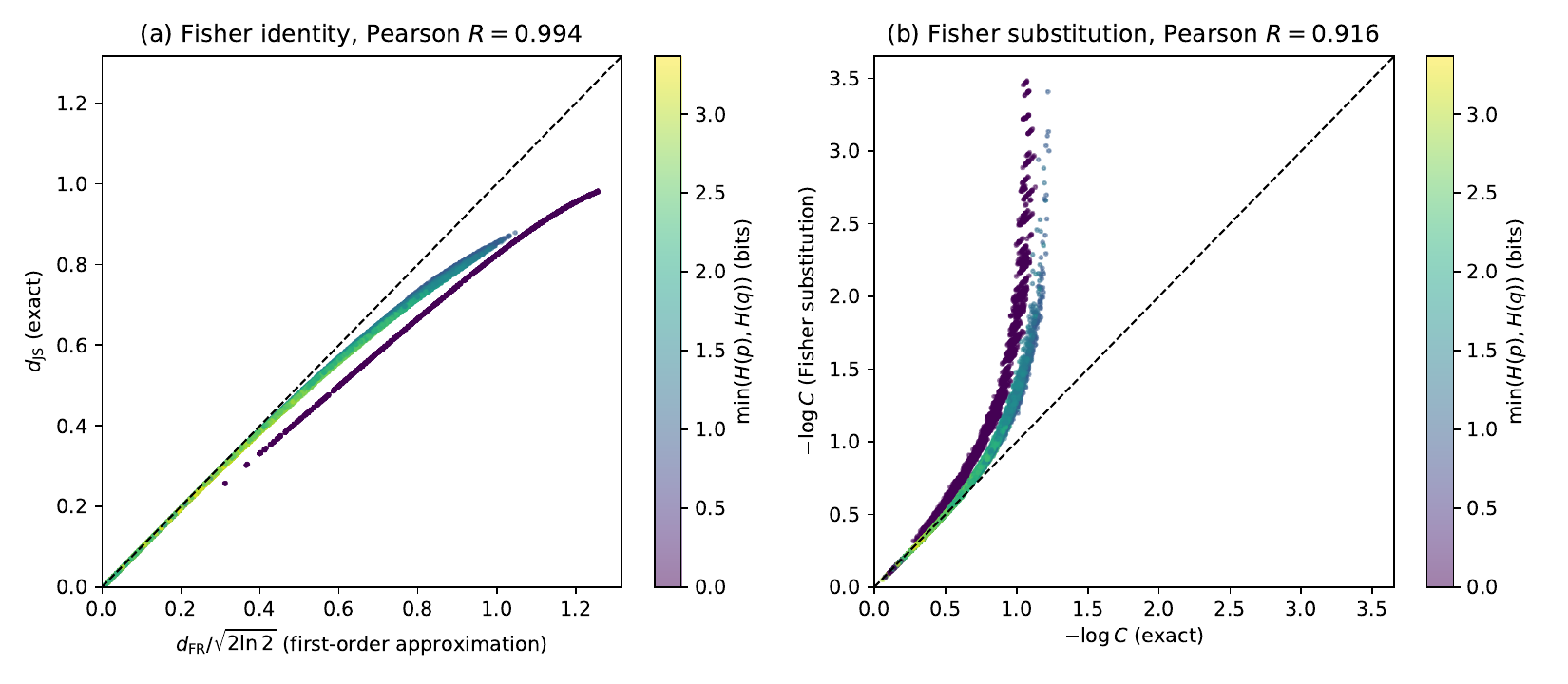}
 \caption{Approximation quality on Cuba, both panels colored by the minimum base-2 Shannon entropy of the two cluster-membership vectors in each pair. (\textbf{a})~First-order Fisher metric identity: exact $d_{\mathrm{JS}}$ vs.\ the linear approximation $\dFR/\sqrt{2\ln 2}$ on the $80{,}950$ pairs with $C > 10^{-10}$. Pearson $R = 0.994$. The scatter resolves into two branches: an upper branch (yellow/green, high minimum entropy, diffuse distributions) hugging $y = x$ where the first-order identity is tight, and a lower branch (dark, low minimum entropy, at least one near-delta distribution) dipping below $y = x$ as the Fisher-Rao metric tensor degenerates near the simplex boundary (Remark~\ref{prop:soft-membership}). The horizontal hook at the top is where $d_{\mathrm{JS}}$ has saturated at its base-2 ceiling of~$1$. (\textbf{b})~Fisher substitution~\eqref{eq:fisher-substitution}: exact angular term with $d_{\mathrm{JS}}$ replaced by $\dFR/\sqrt{2\ln 2}$, on the $51{,}342$ pairs where the substitution is defined ($\dFR/\sqrt{2\ln 2} < 1$). Pearson $R = 0.916$. The same two-branch structure appears here: the diagonal agreement band at the lower-left is the upper-branch population (high entropy, faithful substitution), and the vertical tail is the surviving $\approx 34\%$ of lower-branch pairs whose substituted cost is amplified toward the log-singularity by the $-\log(1-\cdot)$ transform. The dashed line is $y=x$.}\label{fig:approx}
\end{figure*}

\subsection{Experiment 2: Combinator Comparison}\label{sec:exp2}
We compare six combinators (GM, AM, HM, Quad, Min, Max) using the product-metric framework from Theorem~\ref{thm:metric-gm}. Two evaluation criteria align the experiment with the theoretical foundations: (a)~triangle inequality violation rate for the bounded dissimilarity $1 - F$, and (b)~rank correlation of $1 - F$ with the product metric $d_\times$.

\textbf{Product metric sanity check.} The product metric $d_\times$ (defined in Theorem~\ref{thm:metric-gm}) is a proper metric by construction (the normalized $\ell^2$-product of the angular and Jensen-Shannon distances, with $d_\times \in [0,1]$), so it cannot violate the triangle inequality. The $0.00\%$ violation rate it records on $10^5$ random triplets is accordingly not a test of the theorem but a check on the implementation---confirming that the computed $d_\times$ matches the construction, not validating a result that is already exact.

\textbf{Bounded triangle inequality.} We measure violations on the bounded dissimilarity $D_F = 1 - F(A_{ij}, T_{ij}) \in [0,1]$, which is defined for all triplets (the alternative $-\log F$ diverges at $F = 0$). Proposition~\ref{prop:bounded-metrics} establishes that $1 - \mathrm{AM}$ and $1 - \min$ are proper metrics by construction ($L^1$ and $L^\infty$ products, respectively), while $1 - \mathrm{GM}$ is not a metric in general, though no violations are observed empirically.

\textbf{Rank correlation with $d_\times$.} For each combinator, we compute the Spearman rank correlation between $1 - F$ and $d_\times$ across all $\binom{n}{2}$ document pairs. This measures how faithfully each combinator's ordering of document pairs matches the product-manifold geometry recorded in Theorem~\ref{thm:metric-gm}.

\begin{table}[t]
 \centering
 \small
 \caption{Combinator comparison on Cuba. Triangle inequality violations for the bounded dissimilarity $1 - F$ from $10^5$ sampled triplets, and Spearman rank correlation with the product metric $d_\times$ over all $\binom{n}{2}$ pairs. Wilson $95\%$ binomial intervals for the violation rates are within $\pm 0.06\%$ of the reported rates at the $0$--$1\%$ regime and $\pm 0.18\%$ at the $8\%$ regime (Quad); Spearman bootstrap CIs (B=200, document-level resampling) are within $\pm 0.001$ of the reported $\rho(d_\times)$ values for GM, AM, HM, Quad, and Min and $\pm 0.01$ for Max.\label{tab:triangle}}
 \begin{tabular}{lrr}
  \toprule
  \textbf{Combinator} & \textbf{Viol.\ (\%)} & \textbf{$\rho(d_\times)$} \\
  \midrule
  Geometric mean (GM) & 0.00 & 0.999 \\
  Arithmetic mean (AM) & 0.00 & 0.997 \\
  Harmonic mean (HM)  & 0.07 & 0.998 \\
  Metric-squared (Quad) & 8.21 & 1.000 \\
  Minimum       & 0.00 & 0.995 \\
  Maximum       & 0.53 & 0.533 \\
  \bottomrule
 \end{tabular}
\end{table}

Table~\ref{tab:triangle} reads as a contrast between the combinators whose metric status is settled and those whose is not. AM and Min are proper metrics (Proposition~\ref{prop:bounded-metrics}), so their $0.00\%$ violation rates are required rather than measured findings, and confirm only the implementation. The informative entries are the combinators with no such guarantee: GM shows $0.00\%$ empirically despite not being a metric in general, because the pathological configurations that violate the triangle inequality (extreme channel imbalance at an intermediate point) do not arise in this corpus. Quad shows $8.21\%$ violations because $1 - C_{\mathrm{Quad}} = d_\times^2$, and squaring a metric does not preserve the triangle inequality. Max shows $0.53\%$ violations, reflecting its failure to penalize weakness in either channel.

The rank correlations with $d_\times$ show that GM ($\rho = 0.999$), HM ($0.998$), AM ($0.997$), and Min ($0.995$) all track the product-manifold geometry closely, and the empirical gap among these four is narrow ($\Delta\rho < 0.005$). Quad achieves $\rho = 1.000$ by construction: since $1 - C_{\mathrm{Quad}} = d_\times^2$ is a monotone transformation of $d_\times$, Spearman rank correlation is perfect. This confirms that Quad is functionally equivalent to the product metric for ranking purposes (Section~\ref{sec:candidates}). Only the maximum combinator ($\rho = 0.533$) is clearly inferior, losing substantial geometric information by ignoring the weaker channel. The narrow gap among GM, HM, AM, and Min is itself informative: it shows that the product-manifold geometry is robust to the choice of combinator, and that the geometric structure of the document space, not the combinator, is the primary determinant of pair ordering. What distinguishes GM from the other three is not rank correlation but \emph{boundary behavior}. GM enforces $C = 0$ whenever either channel is zero (the ``veto'' property), blocking transitions between topically unrelated events regardless of semantic similarity. AM permits nonzero coherence when one channel is zero and so lacks the veto, whereas HM and Min, like GM, enforce it; HM, however, collapses pathologically toward zero for asymmetric inputs, and Min ignores the stronger channel entirely, discarding information. The GM thus occupies a unique position: it is the only combinator among the five that (a)~enforces the veto property, (b)~uses information from both channels, and (c)~satisfies the axiomatic characterization of Theorem~\ref{thm:axioms}. A veto analysis on both the full (dense) graph and MST-pruned graph confirms that unconstrained maximin extraction never traverses low-$T$ edges (0/100 paths on either graph for both corpora), making the veto redundant for that algorithm class; under the cardinality-constrained LP extraction regime it instead becomes binding (Supplementary Materials, Section~S9).

\textbf{GM vs.\ $d_\times$ for path extraction.} Since $d_\times \in [0,1]$, we define the coherence-space counterpart $C_{d_\times} = 1 - d_\times \in [0,1]$, so that all six methods (five combinators plus $d_\times$) produce coherence matrices in $[0,1]$ and use the same maximum-capacity-path algorithm. Maximin on $C_{d_\times}$ is equivalent to minimax on $d_\times$ because maximin extraction is invariant under monotone transformations of the edge weights (Corollary~\ref{cor:maximin-scale-inv}). The downstream evaluation in Section~\ref{sec:exp5} compares all methods head-to-head using this unified framework.

\subsection{Experiment 3: Scale Complementarity}\label{sec:exp3}
Scale complementarity (Remark~\ref{def:complementarity}) posits that $A$ and $T$ capture different levels of abstraction. Since $T$ is derived from the same embeddings as $A$ via lossy compression, true statistical independence is structurally impossible, and a positive correlation is guaranteed. The question is whether the channels are nonetheless \emph{non-redundant}: does the compression discard enough fine-grained information that $T$ captures something genuinely different from $A$? We test this by computing $A_{ij}$ and $T_{ij}$ for all $\binom{418}{2} = 87{,}153$ document pairs.

Figure~\ref{fig:correlation} shows the joint distribution. The Spearman rank correlation is $\rho = 0.283$ ($p < 10^{-100}$), a weak positive association, consistent with the structural coupling but far from redundancy. Since Spearman captures only monotonic dependence, we also estimate the mutual information $I(A;T)$ via histogram binning (30~bins per axis). The normalized mutual information $\mathrm{NMI} = 2I(A;T) / [H(A) + H(T)] = 0.027$ captures arbitrary nonlinear dependencies, confirming that $A$ and $T$ share less than $3\%$ of their information content, even accounting for the nonlinear UMAP+HDBSCAN pipeline. Together with the multi-scale structure verified by the data processing inequality (Proposition~\ref{prop:dpi}), these results confirm scale complementarity (Remark~\ref{def:complementarity}): the channels are non-redundant (condition~1) and capture different levels of abstraction (condition~2).

Scale complementarity validates the pipeline design by confirming that the UMAP$\to$\allowbreak{}HDBSCAN compression produces a topic channel distinct from the angular channel.

\begin{figure}[t]
 \centering
 \includegraphics[width=0.5\textwidth]{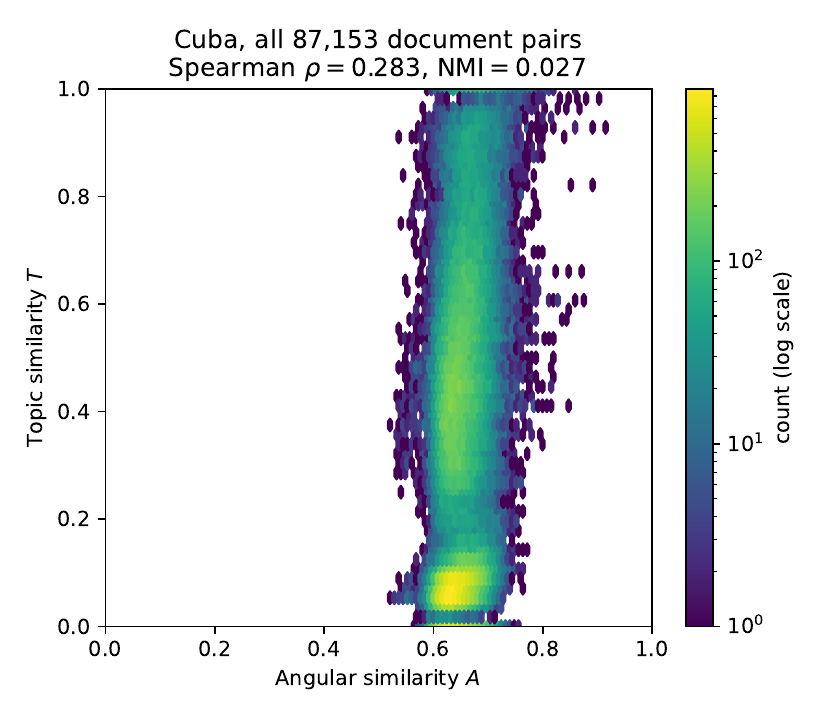}
 \caption{Joint distribution of angular similarity $A$ and topic similarity $T$ across all $87{,}153$ document pairs. The Spearman rank correlation is $\rho = 0.283$ ($\mathrm{NMI} = 0.027$): the channels are structurally coupled but capture different levels of abstraction (scale complementarity).\label{fig:correlation}}
\end{figure}

\subsection{Experiment 4: Cross-Corpus Validation}\label{sec:exp4}
Experiments~1 to 3 use a single news corpus. To verify that the information-geometric properties are not corpus-specific, we replicate the key findings across four corpora spanning different domains, sizes, and embedding models, and then test sensitivity to embedding model choice on a fixed corpus. The four corpora are: \emph{Cuba} (418 news articles on Cuban politics; GPT-4 embeddings, $d = 1536$; the primary corpus), \emph{COVID} (40 news articles on the COVID-19 pandemic; MiniLM embeddings, $d = 384$; from Narrative Maps \citep{keith2020narrative}), \emph{VisPub} (3{,}549 visualization research papers with abstracts; GPT-4 embeddings, $d = 1536$, for the metric-level analysis, and SPECTER2 embeddings, $d = 768$, over all 3{,}620 papers for downstream evaluation; from Narrative Trails \citep{german2025narrative}), and \emph{AMiner} (6{,}000 ML/AI research papers; GPT-4 embeddings, $d = 1536$; from Narrative Trails \citep{german2025narrative}).

For each corpus we compute: (1)~the Fisher identity correlation $R(d_{\mathrm{JS}}, \dFR/\sqrt{2\ln 2})$ on pairs with $C > 10^{-10}$ (matching the filter of Experiment~1), (2)~the GM rank correlation $\rho(1-C_{\GM}, d_\times)$ over all pairs, and (3)~the channel complementarity $\rho(A,T)$ (lower values indicate greater non-redundancy). We also test triangle inequality violations for $1 - C_{\GM}$. Table~\ref{tab:cross-corpus} summarizes the results.

\begin{table}[t]
 \centering
 \small
 \caption{Cross-corpus validation. Fisher identity $R \ge 0.99$, GM rank correlation $\rho \ge 0.989$, and channel correlation low (indicating scale complementarity). Bracketed values are $95\%$ bootstrap confidence intervals from $B = 200$ document-level resampling iterations on the channel correlation~$\rho_{A,T}$, excluding the same-document pairs created by resampling with replacement. The bootstrap half-widths for $R_{\mathrm{Fisher}}$ and $\rho_{\GM}$ are all within $\pm 0.01$ of the point estimates (typically $\pm 0.002$) and are omitted from the table for clarity.$^{\dagger}$\label{tab:cross-corpus}}
 \begin{tabular}{lrrrrr}
  \toprule
  \textbf{Corpus} & $n$ & $K$ & $R_{\mathrm{Fisher}}$ & $\rho_{\GM}$ & $\rho_{A,T}$ \;\; (95\% CI) \\
  \midrule
  Cuba              & 418     & 11  & 0.994 & 0.999 & 0.283 \;\; {\scriptsize $[\phantom{-}0.231,\;\phantom{-}0.337]$} \\
  COVID$^{\dagger}$ & 40      & 2   & 0.998 & 0.993 & $-$0.014 \;\; {\scriptsize $[-0.140,\;\phantom{-}0.154]$} \\
  VisPub            & 3{,}549 & 113 & 0.997 & 0.992 & 0.234 \;\; {\scriptsize $[\phantom{-}0.214,\;\phantom{-}0.255]$} \\
  AMiner            & 6{,}000 & 163 & 0.997 & 0.990 & 0.153 \;\; {\scriptsize $[\phantom{-}0.140,\;\phantom{-}0.167]$} \\
  \bottomrule
 \end{tabular}

 \smallskip
 \noindent{\footnotesize $^{\dagger}$The COVID corpus ($n=40$, $K=2$) is below the framework's recommended operating regime ($K \ge 4$); it is included to illustrate boundary behavior. The wide $\rho_{A,T}$ confidence interval on COVID reflects the small sample size; the other corpora yield half-widths $\le 0.06$. See text for discussion of its 7.6\% triangle inequality violation rate.}
\end{table}

The Fisher identity holds consistently ($R \ge 0.99$) across all corpora, indicating that the JSD-Fisher-Rao connection is not an artifact of a particular corpus or embedding model. The GM rank correlation with $d_\times$ is $\rho \ge 0.989$ in all cases (the smallest, $0.990$ to three decimals, is AMiner's), consistent with the product-metric interpretation. Channel correlation remains low across all corpora ($|\rho(A,T)| \le 0.283$ with $95\%$ bootstrap intervals bracketed in Table~\ref{tab:cross-corpus}), consistent with scale complementarity (Remark~\ref{def:complementarity}): despite the structural coupling, the channels capture different levels of abstraction. The COVID corpus shows a near-zero correlation ($\rho = -0.014$, CI $[-0.14, 0.15]$), consistent with its small size ($n=40$) and minimal cluster structure ($K=2$). With only two clusters, the coarse-grained thematic dimension is so impoverished that almost no monotonic structure couples it to the fine-grained angular dimension, and the wide confidence interval reflects the limited statistical leverage of $n = 40$ documents.

The COVID corpus ($K = 2$) exhibits $7.6\%$ triangle inequality violations for the bounded dissimilarity $1 - C_{\GM}$, while all other corpora ($K \ge 11$) show $0\%$. With only two clusters, membership vectors concentrate near the simplex vertices, producing $T \approx 0$ for documents in opposite clusters, exactly the pathological channel imbalance identified in the counterexample of Proposition~\ref{prop:bounded-metrics}. A hyperparameter sweep over HDBSCAN configurations on the COVID corpus confirms a monotonic relationship, with the violation rate dropping from $7.6\%$ at $K = 2$ to $< 0.5\%$ at $K = 4$ and $< 0.1\%$ at $K \ge 5$. For corpora with sufficient cluster structure ($K \ge 4$), the empirical metric properties hold robustly.

\begin{figure*}[t]
 \centering
 \includegraphics[width=\textwidth]{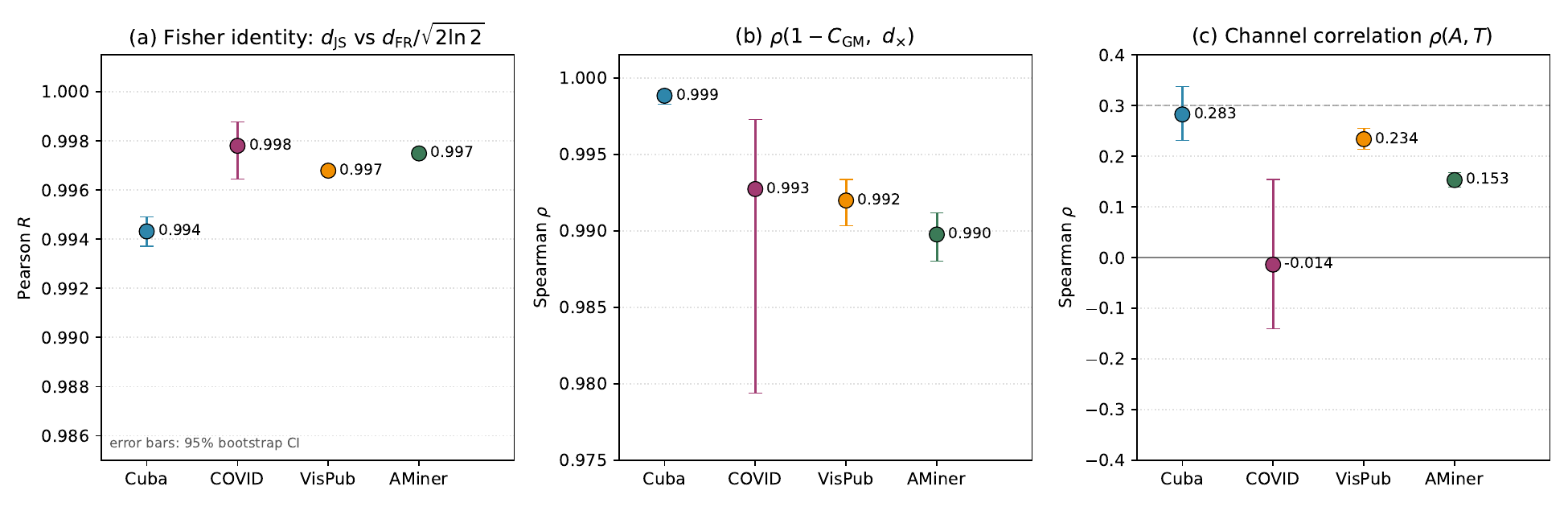}
 \caption{Cross-corpus validation across four corpora: point estimates with $95\%$ bootstrap confidence intervals ($B = 200$ document-level resamples) as error bars. Panels (a) and (b) use zoomed $\rho$/$R$ axes; the markers are points (not bars), so the zoom does not exaggerate the sub-percent spread. (\textbf{a})~Fisher identity holds consistently ($R \ge 0.99$). (\textbf{b})~GM rank correlation with $d_\times$ satisfies $\rho \ge 0.989$. (\textbf{c})~Channel correlation remains low, confirming scale complementarity.}\label{fig:cross-corpus}
\end{figure*}

\textbf{Embedding model sensitivity.} To disentangle the effect of the embedding model from the corpus, we re-embed the Cuba and VisPub corpora with two further models (MiniLM-L6, $d = 384$; MPNet, $d = 768$) and recompute the full UMAP$\to$\allowbreak{}HDBSCAN$\to$coherence pipeline. The metric-level properties are stable: GM rank correlation with $d_\times$ stays at $\rho \ge 0.99$ and triangle-inequality violations at $0\%$ across all six corpus-embedding combinations, while the channel correlation $\rho_{A,T}$ varies within $[0.23, 0.29]$, remaining in the low-correlation regime expected by scale complementarity. The full results are in the Supplementary Materials (Section~S3).

\subsection{Experiment 5: Downstream Consistency Check}\label{sec:exp5}
Experiments~1 to 4 validate the information-geometric framework; the Supplementary Materials extend the validation to alternative topic models (Section~S1) and cluster-count sensitivity (Section~S2). Here we perform a \emph{downstream consistency check} on the axiomatic recommendation of Section~\ref{sec:geomean}: we verify that the theoretically-prescribed GM is not empirically dominated by any alternative combinator or single-channel baseline on the task the metric is designed for, following the Narrative Trails framework \citep{german2025narrative}.

\textbf{Setup.} We evaluate on two corpora: the Cuba news corpus (418 documents) and the VisPub corpus (3{,}620 papers). For VisPub, we use SPECTER2 embeddings \citep{cohan2020specter,singh2023scirepeval} ($d = 768$) rather than the GPT-4 embeddings used in the metric-level experiments; SPECTER2 is trained on citation graphs, producing paths that follow research lineages rather than generic semantic similarity. The metric-level experiments use a 3{,}549-paper subset---the papers carrying abstracts, which the text-based GPT-4 embedding requires---whereas SPECTER2 represents all 3{,}620 papers, including the 71 abstract-less entries. Both corpora use a unified UMAP projection: both the angular similarity channel ($A$) and the topic similarity channel ($T$, via HDBSCAN clustering) are derived from the same UMAP($k^*$) embedding space, where $k^*$ is selected automatically via angular similarity saturation ($< 10\%$ relative change)---the smallest dimension at which the angular structure stops changing. Deriving both channels from a single projection keeps them commensurable.\footnote{This configuration preserves the metric-level guarantees established for the raw-embedding configuration in Sections~\ref{sec:product}--\ref{sec:properties}. On both downstream corpora the saturation criterion selects $k^* = 5$, and the product metric $d_\times$ (Theorem~\ref{thm:metric-gm}) remains free of triangle-inequality violations under triplet sampling; $C_{\mathrm{GM}}$ tracks $d_\times$ at Spearman $\rho = 0.98$ (Cuba) and $0.99$ (VisPub), and the angular and topic channels remain complementary (Pearson correlation $0.26$ and $0.23$). These match the raw-embedding configuration, in which $C_{\mathrm{GM}}$ tracks $d_\times$ at $\rho = 0.999$ (Experiment~2, Section~\ref{sec:exp2}) and the angular and topic channels show Pearson correlation $0.29$ on Cuba (the Pearson counterpart of the Spearman $\rho_{A,T} = 0.283$ of Experiment~3, Section~\ref{sec:exp3}).} The shared UMAP coordinates are mean-centered before the angular channel is computed, placing the corpus centroid at the origin so that $A$ measures relative direction among documents; because UMAP output sits at an arbitrary off-origin offset, omitting this step collapses the angular similarities toward~$1$. Every other angular channel in this paper---the metric-level Experiments~1--4 and the bottleneck-gap Experiment~7---is instead computed from the raw embeddings, where the UMAP step feeds only the topic clustering and the reported relational summaries are insensitive to centering. Both corpora are temporally ordered, so we build \emph{directed} sparse coherence graphs using a maximum spanning arborescence (Edmonds' algorithm \citep{edmonds1967optimum}) rooted at the earliest event, rather than the undirected MST used in the original Narrative Trails pipeline. For each of the six combinators (GM, AM, HM, Quad, Min, Max), we build a sparse graph using the arborescence-derived critical weight. Each combinator is tested with and without Dirichlet smoothing for channel balance (14~method variants total), plus two single-channel baselines ($A$-only and $T$-only).

For each corpus, we select 30 source-target pairs (source in the first temporal third, target in the last), requiring that \emph{all methods} produce a valid path of length $\geq 3$ for the same endpoints. This shared-endpoint design enables paired statistical comparisons. All methods use the same maximum-capacity-path (maximin) algorithm from Narrative Trails.

\textbf{Structural pre-filtering.} On Cuba, GM, AM, HM, and Quad form a single structural equivalence class ($>70\%$ path identity or $>85\%$ Jaccard node overlap), consistent with the narrow rank-correlation gap in Experiment~2. We use GM as the representative of this class and evaluate in three stages, with path-level deduplication across methods that produce identical paths for a given endpoint pair:
\begin{enumerate}[noitemsep,topsep=0pt]
 \item \textbf{Stage~1 (Balance):} GM vs.\ GM$_{\mathrm{bal}}$ (Dirichlet-smoothed) to determine whether channel balancing affects sequence quality.
 \item \textbf{Stage~2 (Combinator):} GM, Quad, Min, Max (using the balance variant carried forward from Stage~1) to compare structurally distinct combinator families.
 \item \textbf{Stage~3 (Channels):} The Stage~2 winner vs.\ single-channel controls ($A$-only, $T$-only) to test whether combining channels improves sequence quality over either alone.
\end{enumerate}
Path-level deduplication further reduces cost, since when two methods produce \emph{identical} paths for a given endpoint pair, the storyline is judged only once and the result is shared. This yields $458$ (Cuba) and $470$ (VisPub) LLM judge calls, roughly $35$--$45\%$ fewer than judging all method variants on every endpoint pair.

\textbf{Evaluation.} Each extracted sequence is scored independently by two frontier LLM judges, GPT-5.4 and Claude Opus 4.6, using the multi-judge protocol of Keith~\citep{keith2025judge}. The prompt presents the storyline in domain-agnostic terms (``document sequence coherence''), asks the judge whether each document follows logically from the previous one and whether the overall sequence forms a meaningful progression from source to target, and returns a single holistic coherence score on a 1 to 100 scale with a one- or two-sentence justification. The 1 to 100 scale, the absence of anchoring phrases, the single holistic score, and the brief justification are the prompt configurations that \citep{keith2025judge} finds to maximize discriminative power while controlling for the positivity bias of longer chain-of-thought protocols. Combined with the paired design (shared endpoints), this yields a score matrix of shape $30 \times k$ ($k$ methods per stage) suitable for repeated-measures analysis.

\textbf{Inter-judge agreement.} The two judges show moderate rank agreement: on Cuba, Pearson $r = 0.57$ to $0.58$ and Spearman $\rho = 0.59$ to $0.61$ across the three stages; on VisPub, $r = 0.28$ to $0.44$ and $\rho = 0.23$ to $0.49$. The Krippendorff $\alpha$ at the interval level is positive on Cuba ($\alpha \approx 0.20$) and negative on VisPub ($\alpha \approx -0.76$), but this is driven by a systematic offset between the two judges' score levels on VisPub (Claude $\mu \approx 20$ versus GPT $\mu \approx 55$) rather than by lack of rank agreement, since the Spearman correlation remains positive. For the paired Friedman/Wilcoxon analysis below the relevant signal is within-pair rank, not absolute level, so the offset does not affect the consistency-check conclusion; we treat the systematic offset as a known LLM-judge calibration effect and revisit it in Section~\ref{sec:limitations}.

\textbf{Statistical analysis.}  Since all methods share the same endpoints, the comparison is a paired (within-subjects) design. We apply the Friedman test \citep{friedman1937use}, the nonparametric analogue of repeated-measures ANOVA, with the Nemenyi post-hoc test \citep{nemenyi1963distribution} for pairwise comparisons. For two-method stages, we use the Wilcoxon signed-rank test. The Nemenyi critical difference is computed with a Bonferroni-normal approximation of the Studentized-range statistic; this is mildly conservative relative to the exact critical value (the critical difference is overstated by roughly $3\%$), which does not affect any conclusion below since no pairwise comparison is significant even under the slightly larger critical difference.

\textbf{Results.} Table~\ref{tab:method-properties} summarizes each method's theoretical profile.

\begin{table}[t]
 \centering
 \small
 \caption{Theoretical properties of each coherence model. Properties: \textbf{Veto} = $C=0$ when either channel is zero; \textbf{Metric} = $1-C$ satisfies triangle inequality; \textbf{Log-add.} = $-\log C$ decomposes additively; \textbf{Axioms} = satisfies all four axioms of Theorem~\ref{thm:axioms}; \textbf{$\rho(d_\times)$} = rank correlation with product metric.\label{tab:method-properties}}
 \begin{tabular}{lccccc}
  \toprule
  \textbf{Method} & \textbf{Veto} & \textbf{Metric} & \textbf{Log-add.} & \textbf{Axioms} & \textbf{$\rho(d_\times)$} \\
  \midrule
  GM         & \checkmark & empirical & \checkmark & \checkmark & 0.999 \\
  AM         & ---        & \checkmark & ---        & ---        & 0.997 \\
  HM         & \checkmark & ---        & ---        & ---        & 0.998 \\
  Quad       & ---        & ---        & ---        & ---        & 1.000 \\
  Min        & \checkmark & \checkmark & ---        & ---        & 0.995 \\
  Max        & ---        & ---        & ---        & ---        & 0.533 \\
  \bottomrule
 \end{tabular}
\end{table}

\textbf{Stage~1: Balance.} On both corpora, the balanced variant (Dirichlet-smoothed cluster probabilities with UMAP dimensionality reduction) ties with or narrowly exceeds the unbalanced variant: Cuba mean scores $83.18$ (GM$_{\mathrm{bal}}$) vs.\ $83.15$ (GM); VisPub $37.77$ (GM$_{\mathrm{bal}}$) vs.\ $37.38$ (GM). Wilcoxon signed-rank tests yield $p = 0.71$ (Cuba) and $p = 0.53$ (VisPub), so the balance variant has at most a negligible effect. We carry the balanced configuration forward as the default (tie-break) choice in subsequent stages. One subtlety affects the notation: Dirichlet smoothing balances the two channels only when they are not already balanced. For combinators whose channels already satisfy the balance condition ($\gamma \approx 1$), the smoothing step is a no-op and the base combinator \emph{is} the balanced configuration. This is the case for Min and Max on VisPub, where no smoothing is applied; we therefore write Min and Max (without the ``bal'' subscript) for VisPub and GM$_{\mathrm{bal}}$, Quad$_{\mathrm{bal}}$, Min$_{\mathrm{bal}}$, Max$_{\mathrm{bal}}$ for Cuba, where all four receive smoothing ($\alpha \approx 0.016$ to $0.022$).

\textbf{Stage~2: Combinator comparison.} Table~\ref{tab:downstream-stage2} reports the Friedman test results. On neither corpus do the four combinators differ significantly: Cuba $\chi^2 = 0.81$, $p = 0.846$ (all pairwise Nemenyi $p = 1.0$), VisPub $\chi^2 = 6.26$, $p = 0.0996$ (no pairwise comparison significant after Nemenyi correction). Quad$_{\mathrm{bal}}$ ranks first on Cuba (mean rank $2.40$) with GM$_{\mathrm{bal}}$ indistinguishable at $2.43$. On VisPub, Max ranks first ($2.02$) but no pairwise difference is significant.

The path-level overlap data explain these scores. On Cuba, GM$_{\mathrm{bal}}$ and Quad$_{\mathrm{bal}}$ share $60\%$ path identity and $86\%$ Jaccard node overlap, producing nearly identical narratives, while Max$_{\mathrm{bal}}$ has $0\%$ path identity with every other method. The functional equivalence of GM$_{\mathrm{bal}}$ and Quad$_{\mathrm{bal}}$ is thus structural, as identical paths receive identical scores, and the $0.03$ rank gap between them is within sampling noise. The more informative finding is that Min$_{\mathrm{bal}}$ and Max$_{\mathrm{bal}}$, which produce structurally \emph{different} narratives (Jaccard $\le 0.30$), still score equivalently. On VisPub, GM$_{\mathrm{bal}}$/Quad$_{\mathrm{bal}}$ overlap is high (Jaccard $0.82$, identity $50\%$), and Max remains a complete structural outlier ($0\%$ identity with all methods). Despite this structural divergence, no significant quality differences emerge on either corpus.

\begin{table}[t]
 \centering
 \small
 \caption{Stage~2: Combinator comparison (30 endpoint pairs per corpus). Mean scores (1 to 100 scale), Friedman mean ranks, and Friedman test statistics. Nemenyi critical difference $\mathrm{CD} = 0.879$ at $\alpha = 0.05$ (computed with the Bonferroni-normal approximation of the Studentized-range statistic; see text). Each method uses the balance configuration carried forward from Stage~1: on Cuba all four combinators are Dirichlet-smoothed, while on VisPub GM and Quad are smoothed and Min and Max are already channel-balanced and use the base combinator.\label{tab:downstream-stage2}}
 \begin{tabular}{lrrrr}
  \toprule
  & \multicolumn{2}{c}{\textbf{Cuba}} & \multicolumn{2}{c}{\textbf{VisPub}} \\
  \cmidrule(lr){2-3} \cmidrule(lr){4-5}
  \textbf{Combinator} & Mean & Rank & Mean & Rank \\
  \midrule
  GM   & 83.4 & 2.43 & 38.2 & 2.68 \\
  Quad & 83.5 & 2.40 & 37.6 & 2.72 \\
  Min  & 82.4 & 2.52 & 38.2 & 2.58 \\
  Max  & 83.2 & 2.65 & 40.1 & 2.02 \\
  \midrule
  Friedman $\chi^2$   & \multicolumn{2}{c}{$0.81$} & \multicolumn{2}{c}{$6.26$} \\
  Friedman $p$        & \multicolumn{2}{c}{$0.846$} & \multicolumn{2}{c}{$0.0996$} \\
  \bottomrule
 \end{tabular}

 \smallskip
 \noindent{\footnotesize Method column gives the combinator family; the balance configuration is corpus-dependent (Cuba: all Dirichlet-smoothed; VisPub: GM and Quad smoothed, Min and Max unsmoothed because already channel-balanced).}
\end{table}

\textbf{Stage~3: Single-channel controls.} The Stage~2 winner (Quad$_{\mathrm{bal}}$ on Cuba, Max on VisPub) is compared against single-channel baselines ($A$-only, $T$-only). On Cuba, Quad$_{\mathrm{bal}}$ achieves the best mean rank ($1.82$), followed by $A$-only ($2.07$) and $T$-only ($2.12$), but the Friedman test is not significant ($\chi^2 = 1.69$, $p = 0.429$, all pairwise Nemenyi $p \ge 0.74$). On VisPub, $A$-only achieves the best mean rank ($1.83$), followed by Max ($1.98$) and $T$-only ($2.18$), also non-significant ($\chi^2 = 1.96$, $p = 0.374$, all pairwise Nemenyi $p \ge 0.53$). Unlike Stage~2, Stage~3 compares structurally \emph{distinct} narratives: on Cuba, Quad$_{\mathrm{bal}}$ shares only $3\%$ path identity and $32\%$ Jaccard overlap with $A$-only, and $7\%$ identity / $43\%$ Jaccard with $T$-only; on VisPub, Max shares $13\%$ identity / $28\%$ Jaccard with $A$-only and $0\%$ identity / $16\%$ Jaccard with $T$-only. The two single-channel controls share $0\%$ identity with each other on both corpora. The combined metric therefore does not significantly outperform either single channel on holistic quality as measured by LLM judges, but its value lies in structural properties that a scalar quality score cannot capture, namely diagnostic decomposability (Proposition~\ref{prop:log-additive}), the veto property (Table~\ref{tab:method-properties}), and the principled reference metric $d_\times$.

\textbf{The consistency check passes.} Across all three stages, GM$_{\mathrm{bal}}$ is never significantly dominated by any alternative combinator or single-channel baseline (no stage significant; the smallest Friedman $p$ is $0.0996$, on VisPub). On Cuba it is statistically indistinguishable from Quad$_{\mathrm{bal}}$ (rank $2.43$ vs.\ $2.40$), with which it is in the same structural equivalence class ($60\%$ path identity, $86\%$ Jaccard node overlap). On VisPub the same holds relative to Max. Had GM significantly underperformed, the axiomatic recommendation of Section~\ref{sec:geomean} would be called into question. The observed parity instead confirms that its distinctions from the other compensability-spectrum members (Section~\ref{sec:candidates}) are theoretical (metric structure, log-additivity, veto behavior) and structural (the bottleneck-gap profile of Experiment~7) rather than a matter of holistic downstream quality, and that the theoretically-prescribed GM is safe to adopt in practice.

\subsection{Experiment 6: Robustness to Embedding Perturbation}\label{sec:exp6}
We test how the metric-level properties degrade when the document embeddings are perturbed by isotropic Gaussian noise. For each noise level $\sigma \in \{0.01, 0.05, 0.10\}$ (relative to the per-coordinate embedding standard deviation), we add fresh noise to the GPT-4 Cuba embeddings, run the full coherence pipeline, and compare against the clean pipeline along three axes: the metric-level correlations ($R_{\mathrm{Fisher}}$, $\rho_{\GM}$, $\rho_{A,T}$), the triangle-inequality violation rate for $1 - C_{\GM}$, and the rank stability $\rho_{\mathrm{stab}} = \rho_{\mathrm{Spearman}}\bigl(C_{\GM}^{\text{noisy}}, C_{\GM}^{\text{clean}}\bigr)$ over all pairs. Three random seeds are averaged.

We report two modes, in line with the discussion of pipeline coupling in Section~\ref{sec:dpi}. In the \emph{fixed-$K$} mode, only the angular channel is perturbed: we hold the HDBSCAN cluster labels fixed at the clean pipeline's output, so the topic channel is unaffected and only $A$ changes under the noise. This cleanly isolates the robustness of the metric structure itself. In the \emph{floating-$K$} mode, the noise is fed through the full UMAP$\to$HDBSCAN$\to$coherence pipeline, so both the cluster count $K$ and the membership vectors are allowed to shift; this is the more honest end-to-end robustness number, but it conflates metric-level effects with clustering instability. Table~\ref{tab:perturbation} summarizes the body-level results; the per-seed breakdown is in the Supplementary Materials (Section~S5).

\begin{table}[t]
 \centering
 \small
 \caption{Embedding perturbation results on Cuba ($n = 418$, GPT-4 embeddings, clean $K = 11$). Each cell is the mean over three seeds; ``stab'' is the Spearman correlation of the $C_{\GM}$ pair ranking under noise against the clean ranking. The fixed-$K$ block holds HDBSCAN output fixed at the clean labels; the floating-$K$ block re-runs the full pipeline.\label{tab:perturbation}}
 \begin{tabular}{l@{\hspace{1em}}rrrrrr}
  \toprule
  \textbf{Mode / $\sigma$} & $K$ & $R_{\mathrm{Fisher}}$ & $\rho_{\GM}$ & $\rho_{A,T}$ & stab $C_{\GM}$ & tri.\ viol.\ (\%) \\
  \midrule
  \emph{clean} (reference)       & 11   & 0.994 & 0.999 & 0.283 & ---     & 0.00 \\
  \addlinespace
  fixed-$K$, $\sigma = 0.01$    & 11   & 0.994 & 0.999 & 0.283 & 1.000   & 0.00 \\
  fixed-$K$, $\sigma = 0.05$    & 11   & 0.994 & 0.999 & 0.283 & 1.000   & 0.00 \\
  fixed-$K$, $\sigma = 0.10$    & 11   & 0.994 & 0.999 & 0.283 & 1.000   & 0.00 \\
  \addlinespace
  floating-$K$, $\sigma = 0.01$ & 11.0 & 0.996 & 0.984 & 0.223 & 0.513   & 0.18 \\
  floating-$K$, $\sigma = 0.05$ & 10.0 & 0.996 & 0.987 & 0.197 & 0.488   & 0.19 \\
  floating-$K$, $\sigma = 0.10$ & 14.3 & 0.995 & 0.998 & 0.305 & 0.610   & 0.00 \\
  \bottomrule
 \end{tabular}
\end{table}

\textbf{Takeaway.} Under fixed-$K$, the metric-level structure is exactly preserved at the rank level ($\rho_{\mathrm{stab}} = 1.000$ at all three noise levels and at the pair-Spearman precision of $10^{-4}$): the angular channel is robust to embedding noise of up to 10\% of per-coordinate standard deviation, and the channel correlation $\rho_{A,T}$, the Fisher identity correlation $R_{\mathrm{Fisher}}$, and the triangle-inequality violation rate of $1 - C_{\GM}$ are all unchanged. Under floating-$K$, the cluster count itself becomes the dominant source of variability ($K$ takes values $\{2, 11{-}16\}$ across seeds and noise levels; the standard deviation of $K$ across seeds at $\sigma = 0.01$ is $7.8$), and $\rho_{\mathrm{stab}}$ drops to roughly $0.5$ as documents are reassigned to different clusters. Even so, the within-run metric-level summaries ($R_{\mathrm{Fisher}}$, $\rho_{\GM}$, $\rho_{A,T}$, triangle violations) remain close to their clean values: the rankings change at the level of which documents are clustered together, but the geometric structure that the framework rests on is preserved at the same statistical level as in the clean pipeline.

The contrast between the two modes is informative. It shows that the robustness limit at small $\sigma$ in the floating-$K$ mode comes from clustering instability rather than from any fragility of the metric construction itself. This is consistent with the discussion in Section~\ref{sec:exp4}: the cluster count is the most impactful pipeline parameter, and any framework that relies on a soft-clustering pipeline inherits the variability of that pipeline as a downstream sensitivity. The metric-level claims of the present paper (Theorems~\ref{thm:metric-gm} and~\ref{thm:axioms}) are conditional on a fixed clustering, and the fixed-$K$ row of Table~\ref{tab:perturbation} verifies that, conditional on a fixed clustering, those claims are robust to embedding noise at the magnitudes tested.

\subsection{Experiment 7: Bottleneck-Gap Profile across Five Corpora}\label{sec:exp7}
Experiment~7 tests Proposition~\ref{prop:gap-profile} directly: we sweep the compensability exponent $\alpha$ of $M_\alpha$ and measure the bottleneck-gap profile $g(\alpha)$ of equation~\eqref{eq:gap-profile} between coherent and incoherent storylines, on five corpora.

\textbf{Setup.} On Wikispeedia \citep{west2012wikispeedia}---$3{,}928$ Wikipedia articles with $10{,}832$ human navigation paths---the coherent population is the human path and the incoherent one a length-matched random path between the same endpoints. On the four narrative corpora (Cuba, COVID, VisPub, AMiner), which carry no human navigation data, the coherent population is the maximin storyline extracted as a temporal-DAG widest path (a \emph{narrative trail} \citep{german2025narrative}) and the incoherent one a length-matched random chronological sequence between the same endpoints. For each storyline we take its bottleneck edge---the minimum-coherence edge---and score it under $M_\alpha$ across a grid of~$\alpha$; $D_{\mathrm H}$ and $D_{\mathrm R}$ are the empirical laws of these bottleneck edges. From the $\alpha=0$ bottleneck edges we also evaluate the leading-order peak prediction $\alpha^\ast=-g'(0)/g''(0)$ of Proposition~\ref{prop:gap-profile}(c).

\textbf{Results.} Figure~\ref{fig:gap-profile} and Table~\ref{tab:gap-profile} summarize the sweep; three predictions of Proposition~\ref{prop:gap-profile} hold on every corpus.

\emph{(i) The even part is maximized at the geometric mean.} On all five corpora the even part $E(\alpha)=\tfrac12(g(\alpha)+g(-\alpha))$, measured directly from the sweep, decreases monotonically away from $\alpha=0$ (plotted for Wikispeedia in Figure~\ref{fig:gap-profile}(a), red curve): the symmetric component of the coherent-versus-incoherent separation is maximized \emph{exactly} at the geometric mean. By Proposition~\ref{prop:gap-profile}(b) this is equivalent to the design-premise condition $\mathbb{E}_{\mathrm R}[Y_\alpha]\ge\mathbb{E}_{\mathrm H}[Y_\alpha]$, which therefore holds on every corpus.

\emph{(ii) The full gap has an interior maximum.} On every corpus $g(\alpha)$ has an interior maximum---never at the Min or Max endpoint (Table~\ref{tab:gap-profile}, $\alpha^\ast_{\mathrm{obs}}$). By Proposition~\ref{prop:gap-profile}(c) the peak sits at $\alpha^\ast=-g'(0)/g''(0)$. The curvature $g''(0)$ is negative on all five corpora (Table~\ref{tab:gap-profile}), so the sign of the displacement follows $g'(0)$: the peak lies above the geometric mean for Wikispeedia ($\alpha=+0.10$) and below it for the four narrative corpora. Here $g'(0)=\tfrac18(\mathbb{E}_{\mathrm H}[M_0\Delta^2]-\mathbb{E}_{\mathrm R}[M_0\Delta^2])$ compares the $M_0$-weighted channel imbalance of the coherent and random bottleneck edges (Proposition~\ref{prop:gap-profile}(c)): it is positive when the human navigation paths carry the more imbalanced bottleneck and negative when the narrative trail carries the less imbalanced one.

\emph{(iii) The predicted peak.} The leading-order estimate $\alpha^\ast_{\mathrm{pred}}=-g'(0)/g''(0)$ of Proposition~\ref{prop:gap-profile}(c) (Table~\ref{tab:gap-profile}) matches the observed peak $\alpha^\ast_{\mathrm{obs}}$ in sign on all five corpora, and in magnitude within a small factor when the peak lies near~$0$ (a factor of $3.2$, $2.1$, and $1.5$ for Wikispeedia, VisPub, and AMiner respectively); for Cuba and COVID the predicted and observed peak \emph{locations} differ more, but this reflects the flatness of $g$ on the compensatory side rather than a failure of the prediction: when the curvature at the geometric mean is small, the interior-critical-point guarantee of Proposition~\ref{prop:gap-profile}(c) lapses (its hypothesis $|g'(0)|<m\delta$ fails), so the leading-order \emph{location} is no longer pinned down. There the profile forms a broad plateau---for both corpora $g$ varies by under $4\%$ across $\alpha\in[-8,-1]$---so the point argmax $\alpha^\ast_{\mathrm{obs}}$ is only weakly identified and is itself sensitive to sampling noise, which the path-level bootstrap confirms (Section~S7). Measured instead by the quantity that matters for extraction, the gap \emph{height}, the leading-order prediction is accurate on every corpus: the gap at the predicted peak $g(\alpha^\ast_{\mathrm{pred}})$ recovers between $92\%$ and $99.7\%$ of the maximum bottleneck gap (Wikispeedia $92.1\%$, Cuba $97.8\%$, COVID $96.4\%$, VisPub $99.7\%$, AMiner $99.7\%$). The closed-form peak therefore lands on the high-gap plateau in all cases, even where its precise location differs from the noisy empirical argmax.

\emph{(iv) Separation at the geometric mean, and a non-circular anchor.} The mean gap $g(0)$ understates how cleanly coherent and incoherent storylines separate at the geometric mean. On Wikispeedia the human navigation paths separate from length-matched random paths with Cohen's $d = 0.65$ and ROC AUC $= 0.77$ on the maximin (bottleneck) objective, and more strongly ($d = 1.20$, AUC $= 0.80$) on the geometric-mean path-score (reliability) objective (both at $\alpha = 0$). On the four narrative corpora the separation is larger still ($d = 1.3$ to $4.1$, AUC $= 0.86$ to $0.99$), but these corpora are an \emph{unsupervised} setting: lacking human-annotated storylines, we take the extractor's own maximin output as the coherent population, so the four narrative corpora serve as an internal consistency check rather than independent validation, and their separation magnitudes should be read in that light. The human-grounded validation of this experiment rests on Wikispeedia, whose coherent population is revealed human navigation behaviour rather than an output of the metric; it is the non-circular anchor, and it still separates the two populations at $d = 0.65$ (bottleneck) and $d = 1.20$ (reliability). The interior, compensatory-side peak is moreover not specific to the maximin objective that Proposition~\ref{prop:gap-profile} analyses: the geometric-mean reliability gap has an interior argmax on every corpus too---negative on all four narrative corpora ($-1.5$ to $-8$) and $-0.25$ on Wikispeedia---so the same structure holds on both extraction objectives.

\textbf{Robustness.} These conclusions rest on replication across five independent corpora rather than on within-corpus resampling. For the smallest corpus, COVID ($169$ endpoint pairs), where sampling noise is most plausible, a path-level bootstrap (Supplementary Materials, Section~S7) confirms that $g(0)>0$, $g''(0)<0$, $g'(0)<0$ and $\alpha^\ast_{\mathrm{obs}}<0$ each hold across the full $95\%$ interval; on the four larger corpora the corresponding intervals are far tighter. As a cross-modal check, Supplementary Materials Section~S8 applies the same analysis to a human-curated \emph{image} narrative (the ROGER expedition-photograph corpus of \citet{german2025roger}): the expert storyline scores above its random null and the even part of $g$ is again maximized at the geometric mean, in a second modality.

\begin{table}[t]
 \centering
 \small
 \caption{Bottleneck-gap profile across five corpora (Experiment~7). Wikispeedia contrasts the human navigation path against a random path; the four narrative corpora contrast the narrative trail (the maximin storyline) against a random chronological sequence. $n$ is the number of documents (articles, for Wikispeedia) in the corpus; the gap is estimated over the corresponding endpoint-pair populations ($10{,}832$ human paths for Wikispeedia; $169$ to $2{,}000$ endpoint pairs for the narrative corpora). $\mu_{\mathrm{coh}}$ and $\mu_{\mathrm R}$ are the mean geometric-mean coherence of the coherent and random bottleneck edges, and their difference $g(0)=\mu_{\mathrm{coh}}-\mu_{\mathrm R}$ is the gap at the geometric mean. $g'(0)$ and $g''(0)$ are the first two Taylor coefficients of the profile at the geometric mean, given in closed form by Proposition~\ref{prop:gap-profile}(c); $\alpha^\ast_{\mathrm{pred}}=-g'(0)/g''(0)$ is the leading-order predicted peak and $\alpha^\ast_{\mathrm{obs}}$ the observed argmax of $g$ over the swept grid. On every corpus $g''(0)<0$, so the profile is concave at the geometric mean and has an interior maximum, and the even part of $g$ is maximized at $\alpha=0$.\label{tab:gap-profile}}
 \begin{tabular}{@{}lrrrrrrr@{}}
  \toprule
  Corpus & $n$ & $\mu_{\mathrm{coh}}$ & $\mu_{\mathrm R}$ & $g'(0)$ & $g''(0)$ & $\alpha^\ast_{\mathrm{pred}}$ & $\alpha^\ast_{\mathrm{obs}}$ \\
  \midrule
  Wikispeedia & $3{,}928$ & $0.200$ & $0.094$ & $+0.0129$ & $-0.0398$ & $+0.32$ & $+0.10$ \\
  Cuba   & $418$     & $0.702$ & $0.177$ & $-0.0695$ & $-0.0757$ & $-0.92$ & $-2.5$ \\
  COVID  & $40$      & $0.774$ & $0.588$ & $-0.0062$ & $-0.0056$ & $-1.11$ & $-5.0$ \\
  VisPub & $3{,}549$ & $0.651$ & $0.108$ & $-0.0775$ & $-0.2186$ & $-0.35$ & $-0.75$ \\
  AMiner & $6{,}000$ & $0.666$ & $0.116$ & $-0.0800$ & $-0.2367$ & $-0.34$ & $-0.5$ \\
  \bottomrule
 \end{tabular}
\end{table}

\begin{figure}[t]
 \centering
 \includegraphics[width=\textwidth]{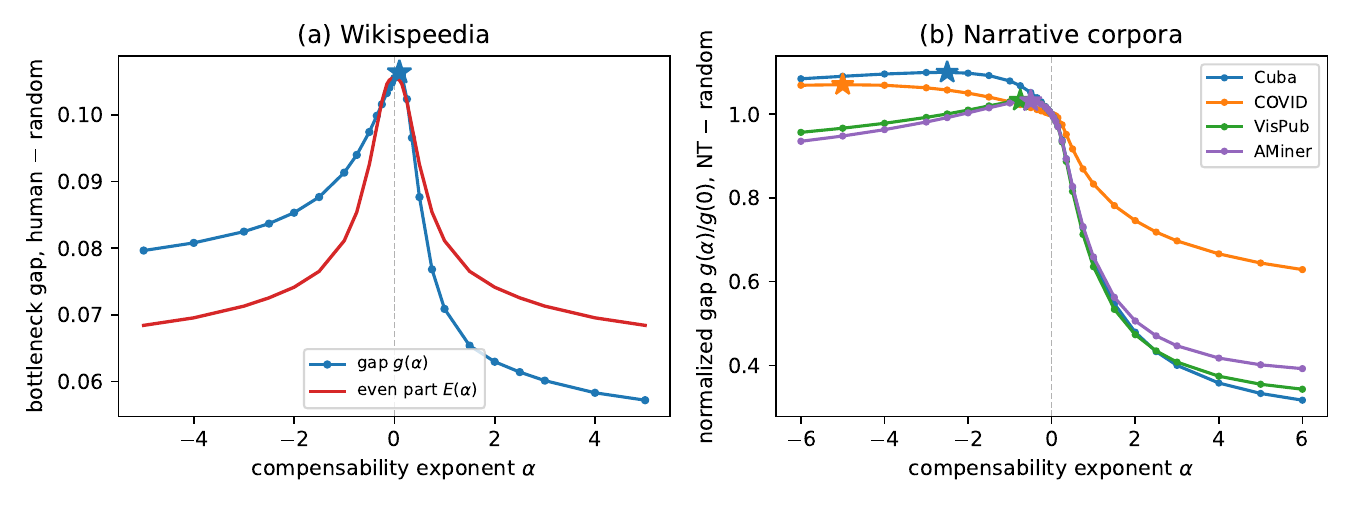}
 \caption{The bottleneck-gap profile $g(\alpha)$ (Experiment~7, Proposition~\ref{prop:gap-profile}). (\textbf{a})~Wikispeedia, human versus random navigation: the gap $g(\alpha)$ (blue, star at its peak $\alpha=+0.10$) and its even part $E(\alpha)$ (red), maximized exactly at the geometric mean $\alpha=0$. (\textbf{b})~The four narrative corpora, narrative-trail versus random chronological sequence, each gap normalized by its value $g(0)$ at the geometric mean; every profile peaks at $\alpha<0$ (stars), modestly inside the compensatory side. Dashed line: $\alpha=0$, the geometric mean.\label{fig:gap-profile}}
\end{figure}

The original Narrative Trails evaluation \citep{german2025narrative} reports the narrative-trail-versus-random bottleneck comparison at the single geometric-mean combinator; that measurement is the $\alpha=0$ point of these profiles. Experiment~7 shows that point sits in the apex region of a curve whose structural separation---the even part---is, across five corpora, pinned to the geometric mean. This is an extraction-output counterpart to the axiomatic argument of Theorem~\ref{thm:axioms}: the geometric mean is not only the unique combinator consistent with the four axioms, it is also where the symmetric component of the storyline-versus-random separation is maximized.

\section{Discussion}\label{sec:discussion}

\subsection{Interpretation}
The product-manifold decomposition (Proposition~\ref{thm:geodesic}) and Chentsov's theorem (Theorem~\ref{thm:chentsov}) together provide an information-geometric reading of the composite coherence metric as a distance on the product space $\Sphere^{d-1}\!\times\!\Simplex^{K-1}_+$ equipped with the round metric on the sphere factor and a Fisher-compatible metric on the simplex factor. Chentsov's theorem singles out the Fisher-Rao metric tensor (up to a positive scale) among Riemannian metrics on the simplex that are invariant under sufficient statistics; because $d_{\mathrm{JS}}$ induces exactly this metric tensor (identity~\eqref{eq:jsd-fisher-base2}), the topic component of the coherence metric is locally consistent with the Fisher-Rao geometry singled out by Chentsov's theorem. The axiomatic characterization (Theorem~\ref{thm:axioms}) is complementary: within the compensability spectrum, the geometric mean is the unique combinator consistent with the four stated axioms, and Section~\ref{sec:axioms} discusses which axioms are load-bearing for this conclusion. Together, the geometric reading and the axiomatic characterization provide an information-theoretic justification for an existing operational construction \citep{keith2020narrative, german2025narrative}.
\subsection{Robustness in Practice}
The mathematical pathologies of Section~\ref{sec:properties}, in particular the non-metric status of $1 - C_{\GM}$ from Proposition~\ref{prop:bounded-metrics}, require extreme channel imbalance ($A \approx 1$ with $T \approx 0$) to manifest. That configuration would need two documents to be nearly identical in embedding space yet assigned to completely disjoint topic clusters, which the UMAP$\to$HDBSCAN coupling structurally prevents. Scale complementarity (Remark~\ref{def:complementarity}) is thus doing double duty: it justifies the product-manifold model (the channels are non-redundant) and it rules out the worst-case configurations that would break metricity.

\subsection{Combinator Selection}
Three properties favor $C_{\GM}$ as the operational tool despite $d_\times$ being a proper metric and Quad ($= 1 - d_\times^2$) being its functionally equivalent similarity-domain counterpart. First, the \emph{veto property} $C_{\GM} = 0$ when $T = 0$ (Section~\ref{sec:properties}) is redundant for the unconstrained maximin (Narrative Trails) extraction used in our experiments (Section~\ref{sec:exp2}), under which every combinator already avoids disjoint-topic edges, but becomes load-bearing under the cardinality-constrained linear-programming extraction of Narrative Maps \citep{keith2020narrative}, whose storyline-size requirement can force the path through otherwise-avoidable edges. Because $C_{\GM} = \sqrt{A\cdot T}$ vanishes on disjoint-topic ($T = 0$) edges, those edges are pruned from the geometric mean's graph and its storylines can never traverse one, whereas the veto-less combinators ($d_\times$, AM) are routed through such transitions at a rate that grows with storyline length; the Supplementary Materials (Section~S9) quantify this across an extraction-length sweep. Because the composite metric is shared by both systems, this is a practical concern rather than a hypothetical one. Second, the log-additive decomposition (Proposition~\ref{prop:log-additive}) keeps both channels diagnostically visible along extracted paths, whereas the squared product metric $d_\times^2$ concentrates $78$ to $92\%$ of the path budget in the angular channel and makes the topic channel nearly invisible. Third, the geometric mean is the established coherence metric in Narrative Maps \citep{keith2020narrative} and Trails \citep{german2025narrative}. Our framework justifies this existing choice rather than proposing a replacement.

Each combinator induces a distinct coherence landscape, a different weighting of the angular and topic channels, which extraction algorithms traverse differently, yielding structurally different narrative paths at statistically indistinguishable aggregate quality (Section~\ref{sec:exp5}). The choice of combinator is therefore a \emph{modeling} decision about which narrative structures to foreground, not a free parameter to be tuned for quality. Table~\ref{tab:method-properties} makes this decision concrete: practitioners who want the axiomatic characterization and log-additive diagnostics should use GM; those who want a formal metric-space guarantee at the similarity level should use AM; those who want the product-manifold metric in similarity form should use Quad; and those who specifically want the veto property should use GM or Min. Our default recommendation of GM rests on its combination of three properties simultaneously: consistency with the four axioms of Theorem~\ref{thm:axioms}, the log-additive decomposition of Proposition~\ref{prop:log-additive}, and the veto property. It is the only combinator in the compensability spectrum (Section~\ref{sec:candidates}) that satisfies all three together under the assumptions stated.

\subsection{Limitations}\label{sec:limitations}

\textbf{LLM-as-judge evaluation, and the absence of a human anchor for the combinator comparison.} The downstream evaluation in Section~\ref{sec:exp5} uses two LLM judges (GPT-5.4 and Claude Opus 4.6) rather than human annotators. The evaluation protocol (brief justification with a single holistic score on a 1 to 100 scale, no anchoring) is designed to maximize discriminative power based on the findings of \citep{keith2025judge}, and the judge prompt uses a domain-agnostic framing (``document sequence coherence'') to avoid privileging literary narrative criteria. Even so, LLM judges are known to exhibit biases that differ from human preferences (verbosity preference, positivity bias, sensitivity to surface form; \citep{zheng2024judging}), and we did not run a human-evaluation study of combinator choices. Human grounding in this paper enters elsewhere: the bottleneck-gap analysis of Section~\ref{sec:exp7} anchors the metric's design premise to human navigation traces via the Wikispeedia corpus.

Three mitigations partially address this gap, and one limitation remains. First, the multi-judge design \citep{zheng2024judging} reduces single-model bias. Second, the role of the downstream evaluation in this paper is a \emph{consistency check} on the axiomatic recommendation rather than a quality comparison: the relevant question is whether the geometric mean is empirically dominated, and a null result is more robust to judge bias than a small positive effect would be, because systematic biases would have to coincidentally exactly cancel across all combinators and baselines to produce a flat ranking. Establishing exact equivalence rather than non-domination would instead require an equivalence test such as the two one-sided tests (TOST) \citep{schuirmann1987comparison,lakens2017equivalence} against a pre-specified margin, which the consistency check does not attempt. Third, the inter-judge agreement statistics reported in Section~\ref{sec:exp5} confirm that the two judges agree on the relative ordering of sequences even where their absolute score levels differ; the level offset that drives the negative Krippendorff $\alpha$ on VisPub does not affect the within-pair Friedman analysis, which is invariant to monotone rescaling of each judge's scores.

The remaining limitation is that all three mitigations operate within the LLM-judge family: agreement between two strong LLM judges does not establish agreement with human readers, and the absolute level of LLM scores is not a calibrated quality measurement. Two considerations bound the scope of this limitation. First, this paper is a \emph{justification} of an existing coherence metric, not a proposal of a new one; the composite metric $C = \sqrt{A \cdot T}$ is already in use in the Narrative Maps \citep{keith2020narrative} and Narrative Trails \citep{german2025narrative} systems, and the Narrative Maps representation has been assessed in human user studies, from the Amazon Mechanical Turk evaluation in the original work \citep{keith2020narrative} to subsequent sensemaking studies \citep{keith2022design,keith2026semantic}, so the metric's downstream usefulness has independent human-grounded support---although those studies evaluated the representation as a whole and did not compare combinator choices against one another. Second, the comparison most relevant here---between combinators within a single structural equivalence class---is one where the candidates are provably close: members of the class produce identical or near-identical paths (Section~\ref{sec:exp5}), so there is little mathematical room for a quality difference. Whether a human-perceptible difference nonetheless exists between mathematically near-equivalent combinators is a well-posed question in its own right, and a controlled human study of \emph{that} question is the natural next step; a full human-evaluation study of narrative quality is otherwise beyond the scope of a paper whose contribution is the formalization of the metric rather than a new extraction method.

\textbf{Prompt-domain trade-off.} The domain-agnostic prompt (``document sequence coherence'') is a deliberate choice: it keeps Cuba and VisPub comparable under one criterion, which is essential for the cross-corpus consistency check, but may under-reward qualities specific to a domain (e.g., story-arc structure in news, logical-progression structure in research paper sequences). A domain-specific prompt per corpus would break this comparability, and since the metric itself is domain-agnostic ($C = \sqrt{A \cdot T}$), an agnostic yardstick matches what the metric is trying to optimize. A per-domain evaluation study is natural future work for benchmark-style comparisons.

\textbf{Single coherence definition.} We analyze the specific composite metric $C = \sqrt{A \cdot T}$ from the Narrative Maps framework \citep{keith2020narrative}. The information-geometric framework applies whenever document similarity combines an embedding-space component with a probabilistic topic component, but we do not empirically compare against coherence definitions from other narrative extraction systems (e.g., word-influence \citep{shahaf2010connecting}, KL-based \citep{li2013evolutionary}).

\textbf{Dependence on upstream representations.} The theoretical results hold for any embedding and any soft topic assignment. The embedding sensitivity analysis (Supplementary Materials, Section~S3) shows that GM rank correlation and triangle inequality violations are stable across three embedding architectures on two corpora, though channel correlation varies with embedding choice ($\rho_{A,T} \in [0.23, 0.29]$). The number of topic clusters has a stronger effect: the COVID corpus ($K = 2$) exhibits $7.6\%$ triangle inequality violations for $1 - C_{\GM}$, with the rate dropping monotonically as the cluster count increases. The practical recommendation is $K \ge 4$ clusters, which was sufficient in all other corpora tested. The clustering pipeline (UMAP hyperparameters, HDBSCAN configuration) was not varied independently of the cluster count.

\textbf{Corpus scope.} While the corpora used span news, academic, and encyclopedic text (40 to 6{,}000 documents) and three embedding architectures, all are English-language. Generalization to other languages or modalities is untested.

\section{Conclusions}\label{sec:conclusions}

We have given an information-geometric reading of composite coherence metrics in event-based narrative extraction. On the product Riemannian manifold $\Sphere^{d-1}\!\times\!\Simplex^{K-1}_+$ the negative log-coherence decomposes additively into an angular and a topic cost; because the Jensen-Shannon distance induces a Riemannian tensor proportional to the Fisher information matrix, the topic component is locally consistent with the Fisher-Rao geometry that Chentsov's theorem singles out, and the same construction motivates the proper product metric $d_\times$ used as a reference distance. The central conclusion is a convergence: two arguments resting on different stated premises---the four operational axioms of Theorem~\ref{thm:axioms} and the design premise that incoherent transitions are the more channel-imbalanced (Proposition~\ref{prop:gap-profile})---independently single out the geometric mean within the compensability spectrum, while the remaining combinators (AM, HM, Quad, Min, Max) each occupy a distinct, characterized position and the Quad combinator $1-d_\times^2$ is recovered as the similarity-domain face of $d_\times$. The empirical study is consistent with this picture: across four corpora the metric-level experiments (Section~\ref{sec:experiments}) confirm the first-order Fisher approximation, the geometric mean's tracking of the product metric, and the non-redundancy of the two channels, while the supporting analyses establish robustness to the embedding model, the topic model, and the cluster count (Supplementary Materials). The downstream consistency check finds no alternative combinator or single-channel baseline that dominates the geometric mean on holistic quality, and the bottleneck-gap profile (Section~\ref{sec:exp7})---extended to a human-navigation corpus and, in a cross-modal case study, to a human-curated image narrative (Supplementary Materials, Section~S8)---shows the design premise holding throughout.

These results together provide an information-geometric justification for the composite coherence metric of \citep{keith2020narrative, german2025narrative} and articulate the conditions under which the geometric-mean combinator is the natural choice. The framework applies whenever document similarity combines a fine-grained embedding-space component with a coarse-grained probabilistic topic component under the design principle of \emph{scale complementarity}, and it illustrates how information-geometric tools (Chentsov's theorem, the Fisher-Rao metric, product Riemannian manifolds, and the Jensen-Shannon distance) can clarify the structure of existing constructions in the information sciences without claiming optimality beyond what the stated assumptions support.

Several questions remain open. The downstream comparison shows that combinators within a single structural equivalence class are indistinguishable to current LLM judges; whether a human-perceptible quality difference nonetheless exists between mathematically near-equivalent combinators is a well-posed question that a controlled human study could settle. Generalization beyond English-language text---to other languages, to modalities beyond the single image-narrative case study, and to coherence definitions from other extraction systems---likewise remains untested. More broadly, scale complementarity as a design principle, and the bottleneck-gap profile as a diagnostic of where a combinator sits on the compensability spectrum, may transfer to other settings in which a similarity combines a fine-grained geometric component with a coarse-grained distributional one; characterizing that transfer is natural future work.


\section*{Supplementary Materials}
The following supporting analyses are available in the Supplementary Materials appended to this preprint (Sections~S1--S9): Section~S1, validation on alternative topic models (LDA, soft $k$-means, GMM); Section~S2, cluster-count sensitivity; Section~S3, embedding-model sensitivity; Section~S4, a component ablation isolating the contributions of soft cluster membership and of the Jensen-Shannon distance; Section~S5, the per-seed breakdown of the embedding-perturbation experiment (Section~\ref{sec:exp6}); Section~S6, a heuristic rate-distortion analogy for the geometric mean; Section~S7, a path-level bootstrap quantifying the sampling uncertainty of the Experiment~7 bottleneck-gap profile on the smallest corpus; Section~S8, a cross-modal case study applying the bottleneck-gap analysis to a human-curated image narrative (the ROGER corpus); and Section~S9, a validation of the geometric-mean veto property under the cardinality-constrained Narrative Maps linear-programming extraction.

\section*{Funding}
This research is funded by the ANID FONDECYT 11250039 Project. The corresponding author is also supported by Project 202311010033-VRIDT-UCN.

\section*{Data Availability}
This preprint contains the manuscript and its supporting analyses
(Sections~S1--S9) only. The data and code used in the experiments are provided
with the version of record published in \emph{Entropy} (MDPI), as its
Supplementary Materials.

\section*{Acknowledgments}
During the preparation of this work, the author used Claude to refine writing and support literature review activities. Additionally, Writefull integrated in Overleaf was used to improve writing quality and readability. After using these tools/services, the author reviewed and edited the content as needed and takes full responsibility for the content of the article.

\section*{Conflicts of Interest}
The author declares no conflicts of interest.

\section*{Abbreviations}
The following abbreviations are used in this manuscript:
\\

\noindent
\begin{tabular}{@{}ll}
JSD & Jensen-Shannon divergence\\
KL & Kullback-Leibler divergence\\
FR & Fisher-Rao\\
GM & Geometric mean\\
AM & Arithmetic mean\\
HM & Harmonic mean\\
NMI & Normalized mutual information\\
LLM & Large language model\\
LDA & Latent Dirichlet Allocation\\
GMM & Gaussian Mixture Model\\
UMAP & Uniform Manifold Approximation and Projection\\
HDBSCAN & Hierarchical Density-Based Spatial Clustering of Applications with Noise\\
DPI & Data processing inequality\\
LP & Linear program\\
MST & Minimum spanning tree
\end{tabular}

\bibliographystyle{unsrtnat}
\bibliography{references}

\clearpage
\setcounter{section}{0}
\setcounter{table}{0}
\setcounter{figure}{0}
\setcounter{equation}{0}
\renewcommand{\thesection}{S\arabic{section}}
\renewcommand{\thesubsection}{S\arabic{section}.\arabic{subsection}}
\renewcommand{\thetable}{S\arabic{table}}
\renewcommand{\thefigure}{S\arabic{figure}}
\renewcommand{\theequation}{S\arabic{equation}}

\begin{center}
{\Large\bfseries Supplementary Materials}\\[2pt]
{\normalsize An Information-Geometric Justification for Composite Coherence in Event-Based Narrative Extraction}
\end{center}
\vspace{0.5em}

\noindent
This document collects supporting analyses for the manuscript
``An Information-Geometric Justification for Composite Coherence in
Event-Based Narrative Extraction.'' These analyses extend the empirical
validation of the main text but are not required for its central
argument; they are provided here to keep the main manuscript focused.
Sections~\ref{sec:s-alt-topic} and~\ref{sec:s-kdep} verify that the
information-geometric properties are robust to the choice of topic
model and to the cluster count. Section~\ref{sec:s-embed} reports the
embedding-model sensitivity analysis. Section~\ref{sec:s-ablation}
gives a component ablation isolating the contribution of soft cluster
membership and of the Jensen-Shannon \emph{distance}.
Section~\ref{sec:s-perturb} gives the per-seed breakdown of the
embedding-perturbation experiment (Experiment~6 of the main text).
Section~\ref{sec:s-ratedist} records a heuristic rate-distortion
analogy for the geometric mean. Section~\ref{sec:s-gapboot} reports a
path-level bootstrap quantifying the sampling uncertainty of the
Experiment~7 bottleneck-gap profile. Section~\ref{sec:s-roger} presents
a cross-modal case study, applying the bottleneck-gap analysis to a
human-curated image narrative (the ROGER corpus).
Section~\ref{sec:s-veto} validates the geometric-mean veto property under
the cardinality-constrained Narrative Maps linear program, supporting the
combinator-selection discussion of the main text. Cross-references of
the form ``the main text'' point to the main manuscript.

\section{Alternative Topic Models}\label{sec:s-alt-topic}
The theoretical results of the main text apply to any soft topic assignment, not just HDBSCAN. Any method that produces a probability vector $\hat{e}_i\in\Simplex^{K-1}_+$ for each document (e.g., LDA topic proportions, soft $k$-means memberships, neural topic model outputs) yields a topic similarity $T = 1-d_{\mathrm{JS}}$ with the same information-geometric properties. To verify this empirically, we replace the UMAP$\to$HDBSCAN pipeline with three alternative topic models (LDA, soft $k$-means, and Gaussian Mixture Models (GMM)) at $K \in \{3, 5, 6, 12, 24\}$ on the Cuba corpus, and evaluate the same metric-level properties. Table~\ref{tab:alt-topic} summarizes the results.

\begin{table}[h]
 \centering
 \small
 \caption{Alternative topic models on Cuba ($n = 418$). The Fisher identity ($R_{\mathrm{Fisher}}$), GM rank correlation with $d_\times$ ($\rho_{\GM}$), triangle inequality violations, and scale complementarity ($\rho_{A,T}$, NMI; lower values indicate greater non-redundancy) are reported for each method and cluster count. The $K \in \{3, 5\}$ rows test whether smoother (lower-$K$) GMM posteriors preserve $\rho_{\GM}$.\label{tab:alt-topic}}
 \begin{tabular}{llrrrrr}
  \toprule
  \textbf{Method} & $K$ & $R_{\mathrm{Fisher}}$ & $\rho_{\GM}$ & Viol.\ (\%) & $\rho_{A,T}$ & NMI \\
  \midrule
  HDBSCAN (baseline) & 11 & 0.994 & 0.999 & 0.0 & 0.283 & 0.027 \\
  \addlinespace
  LDA        &  3 & 0.996 & 0.999 & 0.01 & 0.165 & 0.010 \\
  LDA        &  5 & 0.996 & 1.000 & 0.03 & 0.179 & 0.011 \\
  LDA        &  6 & 0.995 & 1.000 & 0.01 & 0.227 & 0.016 \\
  LDA        & 12 & 0.995 & 1.000 & 0.0  & 0.142 & 0.009 \\
  LDA        & 24 & 0.997 & 1.000 & 0.0  & 0.198 & 0.017 \\
  \addlinespace
  Soft $k$-means &  3 & 1.000 & 0.969 & 0.0 & 0.512 & 0.059 \\
  Soft $k$-means &  5 & 1.000 & 0.974 & 0.0 & 0.485 & 0.058 \\
  Soft $k$-means &  6 & 1.000 & 0.975 & 0.0 & 0.536 & 0.071 \\
  Soft $k$-means & 12 & 1.000 & 0.984 & 0.0 & 0.610 & 0.096 \\
  Soft $k$-means & 24 & 1.000 & 0.984 & 0.0 & 0.585 & 0.090 \\
  \addlinespace
  GMM        &  3 & 0.997 & 0.980 & 2.01 & 0.215 & 0.027 \\
  GMM        &  5 & 0.997 & 0.912 & 1.30 & 0.159 & 0.019 \\
  GMM        &  6 & 0.997 & 0.838 & 0.78 & 0.161 & 0.023 \\
  GMM        & 12 & 0.998 & 0.677 & 0.24 & 0.457 & 0.042 \\
  GMM        & 24 & 1.000 & 0.434 & 0.01 & 0.378 & 0.037 \\
  \bottomrule
 \end{tabular}
\end{table}

\textbf{Takeaway.} LDA and HDBSCAN, which produce smooth (high-entropy) membership vectors, preserve the product-manifold properties across the full $K$ range. GMM exhibits a $K$-dependent trade-off: at low $K$ ($K = 3$) its $\rho_{\GM} = 0.980$ approaches the LDA / soft $k$-means level, but with $2.0\%$ triangle inequality violations on $1 - C_{\GM}$; as $K$ increases, the violations vanish ($0.01\%$ at $K = 24$) while $\rho_{\GM}$ degrades monotonically to $0.43$. The two failure modes bracket the usable GMM range to a narrow neighborhood of $K = 3$--$5$; LDA and HDBSCAN remain clean across the full range.

The Fisher identity holds across all topic models tested ($R \ge 0.99$, and $R \ge 0.995$ for the three alternatives), indicating that this property depends only on the JSD-Fisher connection. The GM rank correlation with $d_\times$ is near-perfect for LDA at every $K$ tested ($\rho \ge 0.999$) and high for soft $k$-means ($\rho \ge 0.969$), but for GMM the $K$-dependence is monotonic and steep: $\rho_{\GM}$ falls from $0.980$ at $K = 3$ to $0.434$ at $K = 24$. The degradation occurs because Gaussian mixture posteriors at high $K$ produce sharply peaked membership vectors that concentrate near simplex vertices, approaching the hard-assignment regime where the Fisher-Rao metric tensor degenerates (cf.\ the soft-membership remark of the main text). Low-$K$ GMM posteriors are smoother and preserve $\rho_{\GM}$ at the cost of higher triangle-violation rates on $1 - C_{\GM}$: at $K = 3$, $2.01\%$ of triplets violate, vs.\ $0.01\%$ at $K = 24$. The framework's practical utility therefore requires membership distributions with sufficient entropy: models whose membership vectors have low entropy (effectively hard clustering at large $K$) violate the smooth-manifold assumption underlying the Fisher-Rao connection. LDA's Dirichlet prior naturally smooths the membership vectors, and HDBSCAN's density-based soft assignments achieve similar smoothness. We recommend topic models that produce smooth membership distributions (LDA, HDBSCAN) across the full $K$ range; GMM is acceptable only at low $K$ where its posteriors remain smooth, and even there with elevated triangle-violation rates. Scale complementarity varies across these models: LDA maintains strong non-redundancy ($\rho_{A,T} \le 0.23$, NMI $\le 0.017$) at every $K$, while soft $k$-means shows moderate correlation ($\rho_{A,T} \approx 0.48$ to $0.61$), reflecting its tendency to retain more fine-grained information.

\section{Cluster-Count Sensitivity}\label{sec:s-kdep}
The number of topic clusters controls the granularity of the topic component. A hyperparameter sweep varying HDBSCAN's \texttt{min\_cluster\_size} on both Cuba and VisPub confirms that the cluster count is the most impactful pipeline parameter. Table~\ref{tab:k-dependence} reports the results.

\begin{table}[h]
 \centering
 \small
 \caption{Cluster count sensitivity. HDBSCAN sweep varying \texttt{min\_cluster\_size} on Cuba and VisPub (VisPub subsampled to 500 documents for the sweep). $w^\ast$ is the optimal angular-topic weight on the weighted geometric mean $C_w = A^w T^{1-w}$ ($0.5$ = equal weighting; $w$ is the channel weight, not the compensability exponent $\alpha$ of the main text); Cohen's $d$ measures the effect size of deviating from $0.5$; CI is the $95\%$ bootstrap confidence interval for $w^\ast$. Cuba at \texttt{mcs}${}=50$ is omitted because HDBSCAN fails to produce a valid clustering ($K < 2$).\label{tab:k-dependence}}
 \begin{tabular}{llrrrr}
  \toprule
  \textbf{Corpus} & \texttt{mcs} & $K$ & $w^\ast$ & CI & $d$ \\
  \midrule
  Cuba &  3 & 31 & 0.45 & [0.40, 0.45] & 0.25 \\
  Cuba &  5 & 14 & 0.50 & [0.50, 0.50] & 0.00 \\
  Cuba &  8 & 12 & 0.45 & [0.40, 0.55] & 0.07 \\
  Cuba & 12 &  2 & 0.55 & [0.55, 0.60] & 0.32 \\
  Cuba & 20 &  4 & 0.55 & [0.45, 0.55] & 0.06 \\
  Cuba & 30 &  2 & 0.55 & [0.55, 0.60] & 0.51 \\
  \addlinespace
  VisPub &  3 & 41 & 0.50 & [0.45, 0.55] & 0.00 \\
  VisPub &  5 & 21 & 0.45 & [0.45, 0.50] & 0.08 \\
  VisPub &  8 & 11 & 0.50 & [0.45, 0.55] & 0.00 \\
  VisPub & 12 &  2 & 0.50 & [0.45, 0.60] & 0.00 \\
  VisPub & 20 &  2 & 0.50 & [0.50, 0.65] & 0.00 \\
  VisPub & 30 &  2 & 0.55 & [0.55, 0.60] & 0.28 \\
  VisPub & 50 &  2 & 0.60 & [0.55, 0.65] & 0.42 \\
  \bottomrule
 \end{tabular}
\end{table}

\textbf{Takeaway.} The optimal weight $w^\ast$ is centered on the equal-weight value $0.5$ across all but the $K = 2$ boundary configurations, where the topic channel collapses to a binary same/different signal.

Across the sweep, $w^\ast$ stays within $[0.45, 0.60]$ on both corpora, with the $95\%$ bootstrap CIs covering $0.5$ in $8$ of the $13$ configurations. Four of the five exceptions are the degenerate $K = 2$ configurations, and the fifth (Cuba \texttt{mcs}${}=3$, $K = 31$) misses \emph{low} ($w^\ast = 0.45$); the deviations from $0.5$ thus occur exactly where the cluster count departs most from the mid-range, consistent with the $\rho(K, w^\ast) = -0.85$ trend. Effect sizes are modest ($d \le 0.51$) and strongest at the degenerate $K = 2$ regime. Cuba at \texttt{mcs}${}=30$ ($K=2$) gives $w^\ast = 0.55$ with $d = 0.51$, and VisPub at \texttt{mcs}${}=50$ ($K=2$) gives $w^\ast = 0.60$ with $d = 0.42$. Away from the $K=2$ boundary, $w^\ast$ tracks $0.5$ closely on both corpora (Cuba $K > 2$: $d \le 0.25$; VisPub $K > 2$: $d \le 0.08$), consistent with the equal-weight structure of the geometric mean. The Spearman correlation $\rho(K, w^\ast) = -0.85$ on Cuba ($p = 0.034$) and $-0.61$ on VisPub ($p = 0.147$) confirms that cluster count is the dominant sweep parameter.

\section{Embedding-Model Sensitivity}\label{sec:s-embed}
The cross-corpus validation of the main text (Experiment~4) uses different corpora with a single embedding model each. To disentangle the effect of the embedding model from the corpus, we re-embed both the Cuba ($n = 418$) and VisPub ($n = 3{,}549$) corpora using two additional models, MiniLM-L6 ($d = 384$) and MPNet ($d = 768$), and recompute the full pipeline (UMAP$\to$HDBSCAN$\to$coherence) for each. Table~\ref{tab:embedding-sensitivity} reports the key metric-level properties across all three embedding models on both corpora.

\begin{table}[h]
 \centering
 \small
 \caption{Embedding model sensitivity on Cuba and VisPub. The metric-level properties (GM rank correlation with $d_\times$, triangle inequality violations) are stable across three embedding models of different dimensionality and architecture. $K$ is the per-config cluster count (an emergent output of HDBSCAN, not a chosen value). Bracketed values are $95\%$ bootstrap confidence intervals on $\rho_{A,T}$ from $B = 200$ document-level resampling iterations (seed $42$, excluding same-document self-pairs); $\rho_{\GM}$ half-widths are within $\pm 0.002$ of the point estimates. Triangle violation counts are out of $10^5$ sampled triplets; the Wilson $95\%$ rate intervals are reported separately below.\label{tab:embedding-sensitivity}}
 \begin{tabular}{llrrrrr}
  \toprule
  \textbf{Corpus} & \textbf{Embedding} & $d$ & $K$ & $\rho_{\GM}$ & Viol.\ (\%) & $\rho_{A,T}$ \;\; (95\% CI) \\
  \midrule
  Cuba & GPT-4 (original) & 1536 & 11 & 0.999 & 0.000 & 0.283 \;\; {\scriptsize $[0.231,\;0.337]$} \\
  Cuba & MiniLM-L6        & 384  & 12 & 0.999 & 0.000 & 0.294 \;\; {\scriptsize $[0.223,\;0.355]$} \\
  Cuba & MPNet            & 768  & 18 & 0.998 & 0.000 & 0.232 \;\; {\scriptsize $[0.174,\;0.290]$} \\
  \addlinespace
  VisPub & GPT-4 (original) & 1536 & 113 & 0.992 & 0.000 & 0.234 \;\; {\scriptsize $[0.214,\;0.255]$} \\
  VisPub & MiniLM-L6        & 384  & 100 & 0.990 & 0.000 & 0.263 \;\; {\scriptsize $[0.240,\;0.281]$} \\
  VisPub & MPNet            & 768  & 101 & 0.990 & 0.000 & 0.267 \;\; {\scriptsize $[0.246,\;0.290]$} \\
  \bottomrule
 \end{tabular}

 \smallskip
 \noindent{\footnotesize Wilson $95\%$ intervals for the triangle-violation rates: all six configurations $[0.000,\,0.004]\%$ ($0$ violations in $10^5$ sampled triplets).}
\end{table}

\textbf{Takeaway.} The GM rank correlation with $d_\times$ and the triangle inequality violations are stable across embedding architectures; only the channel correlation $\rho_{A,T}$ varies, and it stays in the low-correlation regime expected by scale complementarity.

GM rank correlation with $d_\times$ remains $\rho \ge 0.99$ across all six corpus-embedding combinations, and the triangle-inequality violation rate of $1 - C_{\GM}$ is $0\%$ in every case ($0$ of $10^5$ sampled triplets; Wilson $95\%$ upper bound $0.004\%$). Every re-embedding clusters well above the degenerate regime ($K$ ranging from $11$ to $113$), so the only substantial violation rate ($7.6\%$) remains that of the degenerate $K = 2$ COVID configuration. Channel correlation varies across configurations ($\rho_{A,T} \in [0.232, 0.294]$) but remains low in all cases, consistent with scale complementarity across embedding architectures. The highest value ($\rho = 0.294$ for Cuba with MiniLM-L6) still corresponds to less than $9\%$ shared variance ($\rho^2 = 0.086$), indicating that even with smaller embedding dimensions the channels remain non-redundant. The framework does not require low channel correlation, so this variation across embeddings does not affect the validity of the geometric-mean construction. The Max combinator remains the clear outlier ($\rho = 0.46$ to $0.62$), confirming that its failure to track the product geometry is not an artifact of a particular embedding space.

\section{Component Ablation}\label{sec:s-ablation}
The combinator comparison of the main text (Experiment~2) ablates the choice of combinator. This section ablates two further design choices that the framework relies on, holding the Cuba GPT-4 corpus and the angular channel fixed:
\begin{itemize}[leftmargin=*, nosep, topsep=2pt]
 \item \textbf{(C1) Soft vs.\ hard cluster membership.} We replace the soft HDBSCAN membership vectors with their hard (one-hot) counterparts, assigning each document entirely to its most probable cluster.
 \item \textbf{(C2) Jensen-Shannon distance vs.\ divergence.} We replace the topic dissimilarity $d_{\mathrm{JS}} = \sqrt{\JSD}$ with the divergence $\JSD$ itself.
\end{itemize}
For each variant we recompute the Fisher identity correlation $R_{\mathrm{Fisher}}$, the GM rank correlation with the product metric $\rho_{\GM}$, the triangle-inequality violation rate of $1 - C_{\GM}$ on $10^5$ sampled triplets, and the mean base-2 Shannon entropy of the membership vectors. Table~\ref{tab:component-ablation} reports the results.

\begin{table}[h]
 \centering
 \small
 \caption{Component ablation on Cuba GPT-4 ($n = 418$, $K = 11$). The baseline is the framework as specified in the main text (soft membership, Jensen-Shannon distance). C1 hardens the membership vectors; C2 substitutes the Jensen-Shannon divergence for the distance. $R_{\mathrm{Fisher}}$ is degenerate under hard membership: the Fisher-Rao distance collapses to a binary signal at the simplex vertices, so the Fisher-identity correlation is not meaningful (rather than literally undefined).\label{tab:component-ablation}}
 \begin{tabular}{lrrrr}
  \toprule
  \textbf{Variant} & $R_{\mathrm{Fisher}}$ & $\rho_{\GM}$ & Viol.\ (\%) & $\bar H$ (bits) \\
  \midrule
  Baseline (soft membership, $d_{\mathrm{JS}}$) & 0.994 & 0.999 & 0.00 & 1.92 \\
  C1: hard one-hot membership                   & ---   & 0.650 & 0.00 & 0.00 \\
  C2: $\JSD$ divergence instead of $d_{\mathrm{JS}}$ & 0.990 & 0.998 & 7.49 & 1.92 \\
  \bottomrule
 \end{tabular}
\end{table}

\textbf{Takeaway.} Both design choices are load-bearing. Hardening the membership vectors collapses the GM rank correlation with the product metric from $0.999$ to $0.650$; using the Jensen-Shannon divergence in place of the distance raises the triangle-inequality violation rate from $0\%$ to $7.5\%$.

\textbf{C1: soft membership.} Hard one-hot assignment drives the mean membership entropy to exactly $0$ bits: every document sits at a simplex vertex. This is the regime where the Fisher-Rao metric tensor degenerates ($1/p_k \to \infty$ as $p_k \to 0$), so $R_{\mathrm{Fisher}}$ is degenerate rather than informative. The topic similarity collapses to a binary same-cluster / different-cluster signal, and the rank correlation of $C_{\GM}$ with the product metric $d_\times$ falls from $0.999$ to $0.650$. The soft membership of the main pipeline is therefore not a cosmetic choice: it is what places the topic channel in the interior of the statistical manifold where the product-manifold interpretation holds. The $0\%$ triangle-violation rate under one-hot membership is an empirical rather than a forced property: a cross-cluster pair has $d_{\mathrm{JS}} = 1$ and hence $D_{\GM} = 1$, so any triangle containing a cross-cluster edge satisfies the inequality automatically, and only same-cluster triples can violate it---whether they do is a corpus-contingent property of the angular channel, not a theoretical guarantee.

\textbf{C2: Jensen-Shannon distance.} Substituting the divergence $\JSD$ for the distance $d_{\mathrm{JS}} = \sqrt{\JSD}$ leaves the rank correlations almost unchanged ($\rho_{\GM}$ from $0.999$ to $0.998$), because squaring is a monotone transformation and rank-based summaries are nearly insensitive to it. The metric structure, however, is not: $\JSD$ is not a metric (only its square root is), so the triangle-inequality violation rate of $1 - C_{\GM}$ rises from $0\%$ to $7.5\%$. This is the empirical counterpart of the design argument in the main text for using the Jensen-Shannon \emph{distance}: the distance is what makes the topic channel a proper metric and keeps the product metric $d_\times$ well-defined.

\section{Per-Seed Embedding-Perturbation Results}\label{sec:s-perturb}
This section complements Experiment~6 of the main text. Table~\ref{tab:perturbation-detail} reports the per-seed values of $K$, $R_{\mathrm{Fisher}}$, $\rho_{\GM}$, the stability $\rho_{\mathrm{stab}}(C_{\GM}^{\text{noisy}}, C_{\GM}^{\text{clean}})$, and the triangle-inequality violation rate for $1 - C_{\GM}$ on Cuba GPT-4 ($n = 418$). The fixed-$K$ block is deterministic up to the precision of the floating-point pipeline: the angular-channel noise leaves the rank ordering of $C_{\GM}$ exactly preserved at all three magnitudes. The floating-$K$ block exposes the clustering instability discussed in the main text: at $\sigma \in \{0.01, 0.05\}$ one of three seeds collapses the cluster count to $K = 2$, the regime where the topic channel degenerates to a binary signal; the other seeds produce $K$ in the range $13$ to $16$, close to the clean $K = 11$.

\begin{table}[h]
 \centering
 \small
 \caption{Per-seed embedding-perturbation results on Cuba GPT-4 ($n = 418$). Three seeds per cell. Clean reference: $K = 11$, $R_{\mathrm{Fisher}} = 0.994$, $\rho_{\GM} = 0.999$, $\rho_{A,T} = 0.283$, $0.00\%$ triangle violations. The final column gives the per-seed channel correlation $\rho_{A,T}$ whose seed means are reported for Experiment~6 in the main text. The fixed-$K$ angular-ranking stability---stab$(A) = 0.999995, 0.999886, 0.999543$ at $\sigma = 0.01, 0.05, 0.10$ (mean over seeds)---confirms the perturbation reaches the angular channel even though the $C_{\GM}$ ranking is held at stab $= 1.0000$.\label{tab:perturbation-detail}}
 \begin{tabular}{l@{\hspace{0.6em}}lrrrrrrr}
  \toprule
  \textbf{Mode} & $\sigma$ & seed & $K$ & $R_{\mathrm{Fisher}}$ & $\rho_{\GM}$ & stab $C_{\GM}$ & Viol.\ (\%) & $\rho_{A,T}$ \\
  \midrule
  fixed-$K$    & 0.01 & 42 & 11 & 0.9943 & 0.9988 & 1.0000 & 0.000 & 0.2826 \\
  fixed-$K$    & 0.01 & 43 & 11 & 0.9943 & 0.9988 & 1.0000 & 0.000 & 0.2826 \\
  fixed-$K$    & 0.01 & 44 & 11 & 0.9943 & 0.9988 & 1.0000 & 0.000 & 0.2826 \\
  fixed-$K$    & 0.05 & 42 & 11 & 0.9943 & 0.9988 & 1.0000 & 0.000 & 0.2824 \\
  fixed-$K$    & 0.05 & 43 & 11 & 0.9943 & 0.9988 & 1.0000 & 0.000 & 0.2822 \\
  fixed-$K$    & 0.05 & 44 & 11 & 0.9943 & 0.9988 & 1.0000 & 0.000 & 0.2826 \\
  fixed-$K$    & 0.10 & 42 & 11 & 0.9943 & 0.9988 & 1.0000 & 0.000 & 0.2821 \\
  fixed-$K$    & 0.10 & 43 & 11 & 0.9943 & 0.9988 & 1.0000 & 0.000 & 0.2817 \\
  fixed-$K$    & 0.10 & 44 & 11 & 0.9943 & 0.9988 & 1.0000 & 0.000 & 0.2824 \\
  \addlinespace
  floating-$K$ & 0.01 & 42 & 16 & 0.9954 & 0.9986 & 0.6003 & 0.001 & 0.2553 \\
  floating-$K$ & 0.01 & 43 &  2 & 0.9970 & 0.9555 & 0.2788 & 0.534 & 0.1163 \\
  floating-$K$ & 0.01 & 44 & 15 & 0.9955 & 0.9983 & 0.6599 & 0.006 & 0.2985 \\
  floating-$K$ & 0.05 & 42 &  2 & 0.9963 & 0.9626 & 0.2380 & 0.563 & 0.0896 \\
  floating-$K$ & 0.05 & 43 & 14 & 0.9952 & 0.9989 & 0.6210 & 0.001 & 0.2440 \\
  floating-$K$ & 0.05 & 44 & 14 & 0.9950 & 0.9987 & 0.6039 & 0.003 & 0.2585 \\
  floating-$K$ & 0.10 & 42 & 13 & 0.9951 & 0.9984 & 0.6306 & 0.001 & 0.2958 \\
  floating-$K$ & 0.10 & 43 & 14 & 0.9952 & 0.9987 & 0.6082 & 0.000 & 0.3293 \\
  floating-$K$ & 0.10 & 44 & 16 & 0.9958 & 0.9979 & 0.5908 & 0.004 & 0.2891 \\
  \bottomrule
 \end{tabular}
\end{table}

The pattern is consistent with the discussion in the main text: where $K$ stays near the clean value, all metric-level quantities are close to the clean reference; the only seeds that produce visible degradation in $\rho_{\GM}$ are the two cases where the perturbation pushes HDBSCAN into the $K = 2$ regime, and these are isolated to the smaller noise magnitudes $\sigma \in \{0.01, 0.05\}$ where the clustering boundary apparently sits closer to the clean configuration. These same two seeds also carry the lowest channel correlation ($\rho_{A,T} = 0.116$ at $\sigma = 0.01$ and $0.090$ at $\sigma = 0.05$, the degenerate two-cluster value), so the floating-$K$ seed-mean $\rho_{A,T}$ reported in the main text (Experiment~6) dips at those two magnitudes and is highest at $\sigma = 0.10$---the only magnitude at which no seed collapses to $K = 2$. The framework's robustness statement is therefore conditional: it holds at the metric level, given a stable clustering.

\section{A Heuristic Rate-Distortion Analogy for the Geometric Mean}\label{sec:s-ratedist}
This section records the rate-distortion analogy referenced in the main text's case for the geometric mean (Section~4, on its log-additivity and scale-invariance properties). It is a heuristic information-theoretic reading of the log-additive cost structure, not a derivation: the geometric mean is already characterized in the main text as the unique combinator consistent with the four stated axioms, and the analogy below merely offers a parallel information-theoretic motivation for the same conclusion.

Suppose, hypothetically, that the transition from event $d_i$ to event $d_j$ is modeled as the joint transmission of two source components, one carrying semantic content and one carrying topical structure, with two simplifying assumptions:
\begin{enumerate}[leftmargin=*, nosep, topsep=2pt]
 \item[(i)] The two source components are independent, so the total rate decomposes as $R = R_A + R_T$.
 \item[(ii)] The per-component rate-distortion functions take the logarithmic form $R_c(D_c) = -\log D_c + \mathrm{const}$.
\end{enumerate}
Under these assumptions, the joint rate-distortion function is $R(D_A, D_T) = R_A(D_A) + R_T(D_T)$, and the iso-rate curves in $(D_A, D_T)$ space are level sets of $-\log D_A - \log D_T = \mathrm{const}$, i.e., $D_A \cdot D_T = \mathrm{const}$. A scalar summary of the distortion along such a curve must therefore be a function of $D_A \cdot D_T$ alone, and the choice $D = (D_A \cdot D_T)^{1/2}$ is the unique such summary satisfying the calibration $D = a$ when $D_A = D_T = a$, which is the same calibration as the normalization axiom of the main text. Because $A$ and $T$ are \emph{similarities}, the distortion variable $D_c$ here is fidelity-like and $R_c(D_c) = -\log D_c$ is \emph{decreasing} in it: the analogy has the form of the Gaussian rate-distortion curve with the distortion axis reversed.

Neither assumption is a factual claim about the real channels. Assumption~(i) requires strict statistical independence, which is structurally impossible in our setting because the topic channel is a deterministic function of the embedding channel via the UMAP$\to$HDBSCAN pipeline. Scale complementarity shows that the channels are nevertheless \emph{approximately} independent in practice ($\mathrm{NMI} = 0.027$), and this is the only sense in which we use Assumption~(i). Assumption~(ii) is motivated by consistency with the log-additive cost structure rather than by first principles. The exact log-additive identity $-\log C_{\GM} = -\tfrac{1}{2}\log A - \tfrac{1}{2}\log T$ holds at any correlation level and at any choice of underlying distributions, so the analogy above is best read as a conceptual motivation for the log-additive structure rather than as an independent argument for the geometric mean. The primary justification for the geometric mean in the main text is the axiomatic characterization.

\section{Bootstrap Uncertainty for the Bottleneck-Gap Profile}\label{sec:s-gapboot}
This section complements Experiment~7 of the main text. The bottleneck-gap profile is validated by replication across five corpora; within each corpus we additionally quantify sampling uncertainty with a path-level bootstrap: the paired endpoint-pair index is resampled with replacement ($B = 2{,}000$), and on each resample we recompute the gap at the geometric mean $g(0)$, the Taylor coefficients $g'(0)$ and $g''(0)$, the leading-order predicted peak $\alpha^\ast_{\mathrm{pred}} = -g'(0)/g''(0)$, and the observed argmax $\alpha^\ast_{\mathrm{obs}}$. Table~\ref{tab:s-gapboot-all} reports the $2.5$/$97.5$ percentile intervals for the gap $g(0)$, its first two Taylor coefficients, and the observed peak on all four narrative corpora (Wikispeedia, with $10{,}832$ human paths, is by far the largest sample and its point estimates are correspondingly the tightest; it is omitted here only because the bootstrap was run on the four narrative corpora). On every corpus both $g(0) > 0$ and $g''(0) < 0$ hold across the entire $95\%$ interval, so ``the narrative trail is more coherent than random'' and ``the profile is concave at the geometric mean, hence has an interior maximum'' are not point-estimate artifacts. The smallest corpus, COVID ($40$ documents, $169$ endpoint pairs), is the only one where the magnitudes are loosely determined, and we detail it in Table~\ref{tab:s-gapboot}.

\begin{table}[h]
 \centering
 \small
 \caption{Path-level bootstrap ($B = 2{,}000$) of the bottleneck-gap profile on the four narrative corpora of Experiment~7. Intervals are $2.5$/$97.5$ percentiles over the resamples; point estimates match the Experiment~7 table of the main text. The sign of $g(0)$ (trail more coherent than random) and of $g''(0)$ (concavity, hence an interior maximum) is stable across the full interval on every corpus.\label{tab:s-gapboot-all}}
 \begin{tabular}{lrrrr}
  \toprule
  Corpus & $g(0)$ CI & $g'(0)$ CI & $g''(0)$ CI & $\alpha^\ast_{\mathrm{obs}}$ CI \\
  \midrule
  Cuba   & $[+0.518, +0.532]$ & $[-0.072, -0.067]$ & $[-0.078, -0.074]$ & $[-3.0, -2.5]$ \\
  COVID  & $[+0.159, +0.213]$ & $[-0.008, -0.004]$ & $[-0.009, -0.003]$ & $[-6.0, -4.0]$ \\
  VisPub & $[+0.534, +0.551]$ & $[-0.081, -0.074]$ & $[-0.225, -0.212]$ & $[-0.75, -0.5]$ \\
  AMiner & $[+0.542, +0.559]$ & $[-0.083, -0.077]$ & $[-0.243, -0.230]$ & $[-0.5, -0.5]$ \\
  \bottomrule
 \end{tabular}
\end{table}

\begin{table}[h]
 \centering
 \small
 \caption{Path-level bootstrap ($B = 2{,}000$) for the narrative-trail-versus-random bottleneck-gap profile on COVID ($n = 40$, $169$ endpoint pairs), the smallest corpus of Experiment~7. Point estimates are those reported for COVID in the Experiment~7 table of the main text; intervals are $2.5$/$97.5$ percentiles over the resamples.\label{tab:s-gapboot}}
 \begin{tabular}{lrr}
  \toprule
  Quantity & Point estimate & $95\%$ bootstrap CI \\
  \midrule
  $g(0)$                        & $0.186$   & $[\,0.159,\ 0.213\,]$ \\
  $g'(0)$                       & $-0.0062$ & $[-0.0083,\ -0.0042]$ \\
  $g''(0)$                      & $-0.0056$ & $[-0.0085,\ -0.0031]$ \\
  $\alpha^\ast_{\mathrm{pred}}$ & $-1.11$   & $[-2.20,\ -0.60]$ \\
  $\alpha^\ast_{\mathrm{obs}}$  & $-5.0$    & $[-6.0,\ -4.0]$ \\
  \bottomrule
 \end{tabular}
\end{table}

Every sign-level conclusion of Experiment~7 holds across the entire $95\%$ interval, even on this smallest corpus: $g(0) > 0$ (the narrative trail is more coherent than a random chronological sequence), $g''(0) < 0$ (the profile is concave at the geometric mean and so has an interior maximum), and $g'(0) < 0$ together with $\alpha^\ast_{\mathrm{obs}} < 0$ (the peak lies on the compensatory side of the geometric mean). At $169$ endpoint pairs the \emph{magnitudes} are only loosely determined---the observed peak is bracketed to $[-6, -4]$---but the qualitative structure is stable. On the three larger narrative corpora ($418$ to $6{,}000$ documents) the same intervals are far tighter (Table~\ref{tab:s-gapboot-all}); AMiner, the largest, has an observed peak of $-0.5$ in every one of the $2{,}000$ resamples. The five-corpus replication reported in the main text is thus the principal robustness evidence, and the bootstrap confirms that the sign-level conclusions---$g(0) > 0$ and $g''(0) < 0$ on every corpus---survive within-corpus resampling, with only the smallest corpus showing wide magnitude intervals. The wide COVID interval is structural as well as statistical: with its small curvature $g''(0)$ the profile is nearly flat on the compensatory side, so the interior-critical-point guarantee of Proposition~5(c) lapses and the peak is weakly identified independently of sample size.

\section{A Cross-Modal Case Study: Image Narratives (ROGER)}\label{sec:s-roger}
The five corpora of Experiment~7 are all textual. To probe whether the
bottleneck-gap signature of the main text (the Proposition of Section~4.4) is
specific to text or reflects the geometry of the composite metric itself, we
apply the same analysis to a human-curated \emph{image} narrative. We use the
ROGER corpus of German et al.\ (\emph{Semi-Supervised Image-Based Narrative
Extraction: A Case Study with Historical Photographic Records}, ECIR 2025):
$501$ photographs of the 1928 Sacambaya Expedition (Robert Gerstmann Fonds,
Universidad Cat\'olica del Norte), released with expert-curated baseline
timelines---the same human-grounded resource and the same narrative-extraction
task as Experiment~7, in a different modality.

\textbf{Pipeline.} The angular channel $A$ is the cosine similarity of the
repository's DETR ResNet-50 image embeddings, mapped to $1-\arccos(\cdot)/\pi$
exactly as for the text corpora. The topic channel $T$ is the
UMAP$\to$HDBSCAN soft membership on the \emph{same} embeddings (the paper's own
topic channel), with a fixed UMAP seed ($K = 20$ clusters). We deliberately do
\emph{not} build $T$ from the corpus's expert theme labels: semi-supervised
label propagation collapses those labels to a near one-hot membership, which
sits at the simplex vertices where the Fisher--Rao tensor degenerates (cf.\ the
component ablation of Section~\ref{sec:s-ablation}) and vetoes the
theme-crossing transitions that a curated narrative is built from---it is not
the soft simplex geometry the topic factor assumes. Images are ordered into a
temporal DAG by their (label-propagated) expedition month, and storylines are
forward paths in that DAG, as in the main-text bottleneck-gap analysis (Experiment~7).

\textbf{Choice of storylines.} The ROGER corpus provides its expert narrative at
six nested lengths ($5$ to $30$) because the original study compared against
Narrative \emph{Maps}, whose linear program lets one fix the extracted storyline
to a target size and so match it to each expert length. Narrative \emph{Trails}
optimizes a different objective and returns the widest path whose length the
graph dictates---it cannot be length-regulated---so the matched-length design no
longer applies. We therefore adopt the natural alternative: extract the trail
once and compare it against the expert baselines closest to its length. Here the
trail has length $13$ (bottleneck coherence $0.876$, reliability $0.893$), so we
report the length-$15$ and length-$10$ expert timelines, keeping the human
reference and the extracted trail of comparable size. This proximity criterion
is independent of the outcome; as a bonus the two timelines have \emph{distinct}
bottleneck edges---the length-$15$ narrative's weakest link is a theme-crossing,
channel-imbalanced transition ($A = 0.84$, $T = 0.07$), the length-$10$
narrative's is channel-balanced ($A = 0.84$, $T = 0.47$)---so they exercise two
contrasting regimes of the metric.

\textbf{Both expert storylines score above chance.} Between the fixed
(source,\,target) endpoints of each timeline we draw a one-sided Monte Carlo
permutation null of $R = 5{,}000$ length-matched random forward sequences (the
random count is the resolution of the null estimate, not a sample size). The
expert timelines exceed their random null on every comparison---three of the four significant at $0.05$, the length-$15$ bottleneck marginal ($p = 0.053$): at length $10$,
bottleneck coherence $0.629$ versus a null mean of $0.180$ (permutation
$p = 0.021$) and reliability $0.760$ versus $0.399$ ($p = 0.007$); at length
$15$, bottleneck $0.243$ versus $0.112$ ($p = 0.053$) and reliability $0.675$
versus $0.360$ ($p = 0.025$). The extracted maximin trail scores higher still
than either expert timeline, so the ordering is \textsc{Trail} $>$
\textsc{Expert} $>$ \textsc{Random}; that the coherence-optimal extractor
outscores the human curator is expected, since the expert curated for thematic
and historical progression rather than for bottleneck coherence, and the
load-bearing comparison is the non-circular \textsc{Expert}-versus-\textsc{Random}
one.

\textbf{The bottleneck-gap signature transfers, and tracks channel balance.}
Sweeping the compensability exponent and decomposing each gap $g(\alpha)$ into
its even and odd parts (the bottleneck-gap Proposition of the main text)
reproduces the Experiment~7 pattern on both timelines: the even (symmetric) part
is maximized at the geometric mean (swept argmax $\alpha = 0$; closed-form
curvature $g''(0) < 0$), and the design-premise condition
$\mathbb{E}_{\mathrm R}[Y_\alpha] \ge \mathbb{E}_{\mathrm H}[Y_\alpha]$ holds at
every swept $\alpha$. The two regimes differ in \emph{how pronounced} the peak
is, exactly as the channel balance of the bottleneck predicts. For the
balanced-bottleneck length-$10$ narrative the peak is sharp ($g''(0) = -0.16$,
leading-order peak $\alpha^\ast = -0.6$, even part down $22\%$ by
$|\alpha| = 1$); for the imbalanced-bottleneck length-$15$ narrative it is
shallow ($g''(0) = -0.02$, $\alpha^\ast = +3.3$, even part down $60\%$ by
$|\alpha| = 1$ but very flat near the apex). Here the quadratic term alone predicts only an $\approx 8\%$ drop by $|\alpha| = 1$; the remaining decline is carried by higher-order even terms, large because the bottleneck imbalance ($\Delta \approx 2.5$) shrinks the Taylor radius, so the near-flatness holds only for $|\alpha| \lesssim 0.3$. The geometric mean is the symmetric
optimum in both cases; the strength of that optimum scales with how balanced the
narrative's weakest transition is.

\textbf{Scope.} ROGER provides a single expert narrative, so the two reported
lengths overlap and are not statistically independent; this is an illustrative
cross-modal case study, not a powered population test like Wikispeedia
($10{,}832$ independent human paths). It is reported as corroboration that the
even-part-at-the-geometric-mean structure is a property of the composite metric
rather than an artifact of the text corpora. As noted above, the
\textsc{Trail}--\textsc{Expert} comparison is between fixed endpoints rather than
at matched length, since Narrative Trails cannot be length-regulated.

\section{The Veto under Cardinality-Constrained Extraction}\label{sec:s-veto}
The veto property of the geometric mean---$C_{\GM}=\sqrt{A\cdot T}=0$ whenever
$T=0$, i.e.\ a transition between disjoint topic clusters---is, by the
combinator-selection argument of the main text (Discussion), redundant for
unconstrained maximin (Narrative Trails) extraction but load-bearing under the
cardinality-constrained linear program of Narrative Maps. This section quantifies
that contrast.

\textbf{Setup.} On the Cuba corpus ($n=418$, canonical $K=11$ soft clustering) we
draw $30$ source--target pairs (source in the first temporal third, target in the
last; seed $42$) and extract a storyline for each under two regimes and three
combinators: the geometric mean $C_{\GM}$, the product-metric complement
$1-d_\times$, and the arithmetic mean $C_{\AM}$. The first regime is the maximin
(widest-path) extractor of Narrative Trails; the second is the Narrative Maps
linear program, whose \emph{cardinality} constraint
$\sum_i \mathrm{node\_act}_i = K_{\mathrm{LP}}$ fixes the storyline to
$K_{\mathrm{LP}}$ active nodes---the ``expected-length'' constraint of the Maps
LP, and \emph{not} its separate topic-coverage constraint (see the caveats). We
sweep $K_{\mathrm{LP}}$ from $6$ to $25$. The veto acts at graph construction:
edges are instantiated only where coherence is positive, and since $C_{\GM}=0$ on
every disjoint-topic edge, all $6{,}203$ of the $87{,}153$ document pairs with
$T=0$ are absent from the geometric mean's graph. We report how many of the $30$
extracted storylines contain an edge with $T=0$ exactly---the veto's precise
boundary.

\textbf{Maximin: the veto is redundant.} Under the unconstrained widest-path
extractor, \emph{all three} combinators avoid disjoint-topic edges entirely
($0/30$ each, Table~\ref{tab:s-veto}): the bottleneck objective discards the
lowest-coherence edges regardless of combinator, so the veto adds nothing. This
is the ``redundant for maximin'' half of the claim, and it holds for the
veto-less combinators ($1-d_\times$, $C_{\AM}$) as well.

\textbf{Cardinality-constrained LP: the veto is load-bearing.}
Table~\ref{tab:s-veto} reports the sweep. The geometric mean traverses a
disjoint-topic edge in $0$ of $30$ storylines at \emph{every} length---a
structural guarantee, since its graph contains no such edge---whereas the
veto-less combinators are routed through disjoint transitions at a rate that
\emph{grows with storyline length}, from $1$--$4$ of $30$ at $K_{\mathrm{LP}}=6$
to $15$--$18$ of $30$ (half or more) at $K_{\mathrm{LP}}=25$. The cardinality
requirement forces a longer chain; without the veto that chain is increasingly
routed onto disjoint edges, whereas the geometric mean's pruned graph cannot
supply one. The result is deterministic for a fixed pipeline---the flow-pruned
extraction and the linear-programming solver are both deterministic---so the only
source of run-to-run variability is the upstream UMAP$\to$HDBSCAN clustering,
characterized in Section~\ref{sec:s-perturb} and Experiment~6 of the main text.

\begin{table}[h]
 \centering
 \small
 \caption{Veto under cardinality-constrained extraction (Cuba, $n=418$, $K=11$,
 $30$ endpoint pairs, seed $42$). Each cell is the number of extracted storylines
 containing an edge with $T=0$ exactly (a disjoint-topic transition; the veto
 boundary). Maximin is length-invariant; the linear program is swept over its
 cardinality (storyline-size) constraint $K_{\mathrm{LP}}$. The geometric mean is
 $0/30$ at every length because $C_{\GM}=0$ prunes disjoint edges from its graph,
 while the veto-less combinators traverse them at a rate rising with
 length.\label{tab:s-veto}}
 \begin{tabular}{lrrr}
  \toprule
  \textbf{Extractor} & $C_{\GM}$ & $1-d_\times$ & $C_{\AM}$ \\
  \midrule
  Maximin (Trails), any length & 0/30 & 0/30 & 0/30 \\
  \midrule
  LP, $K_{\mathrm{LP}}=6$  & 0/30 & 1/30  & 4/30 \\
  LP, $K_{\mathrm{LP}}=8$  & 0/30 & 5/30  & 5/30 \\
  LP, $K_{\mathrm{LP}}=10$ & 0/30 & 6/30  & 6/30 \\
  LP, $K_{\mathrm{LP}}=12$ & 0/30 & 11/30 & 8/30 \\
  LP, $K_{\mathrm{LP}}=15$ & 0/30 & 10/30 & 12/30 \\
  LP, $K_{\mathrm{LP}}=20$ & 0/30 & 14/30 & 14/30 \\
  LP, $K_{\mathrm{LP}}=25$ & 0/30 & 15/30 & 18/30 \\
  \bottomrule
 \end{tabular}
\end{table}

\textbf{Caveats.} Two points fix the scope. (i)~\emph{Cardinality, not
coverage.} The swept constraint is the Maps LP's node-count constraint, which
sets storyline size; it is \emph{not} the LP's separate topic-coverage
constraint. We separately implemented the actual coverage constraint (the Maps
``MinCover'' requirement, whose per-cluster credit is the Bhattacharyya
coefficient $\sum_k\sqrt{\hat e_{ik}\,\hat e_{jk}}$ of the endpoints' membership
vectors) and swept its threshold: because that credit rewards within-cluster
(high-$T$) edges and assigns disjoint-topic edges essentially zero credit,
coverage steers the path \emph{toward} topical cohesion and only reduces
disjoint-edge traversal. The cardinality-only result reported here is therefore
conservative---adding coverage strengthens, not weakens, the veto's relevance.
(ii)~\emph{The metric is the exact boundary $T=0$.} The veto is the statement
$C_{\GM}=0\iff T=0$; relaxing the threshold to $T<0.05$ admits \emph{near}-disjoint
edges ($0<T<0.05$, where $C_{\GM}=\sqrt{A\cdot T}>0$ and the veto does not act),
on which the geometric mean is small but nonzero ($0$ to $6$ of $30$ across the
sweep). $T=0$ is thus the faithful metric for the veto, and it is there that the
geometric mean is uniformly $0$.

\end{document}